\documentclass[pra,aps,twocolumn,superscriptaddress,longbibliography,nofootinbib,showpacs]{revtex4-1}

\usepackage[utf8]{inputenc}      
\usepackage[T1]{fontenc}         
\usepackage[english]{babel}      

\usepackage[colorlinks=true, citecolor=blue, urlcolor=blue]{hyperref}

\usepackage{amsmath, amssymb, amsfonts, mathtools, bm, bbm}
\usepackage{amsthm}
\usepackage{physics}
\usepackage{braket}
\usepackage{esvect}
\usepackage{latexsym}
\usepackage{mathrsfs}
\usepackage[mathscr]{euscript}

\usepackage{graphicx}
\usepackage{wrapfig}
\usepackage{float}
\usepackage{adjustbox}
\usepackage{graphics, epstopdf}
\usepackage{epsfig}
\usepackage{tabularx}
\usepackage{booktabs}
\usepackage{dcolumn}

\usepackage{xcolor}
\usepackage{soul}
\usepackage[normalem]{ulem}
\usepackage{enumitem}
\usepackage{csquotes}
\usepackage{fmtcount}
\usepackage{verbatim}
\usepackage{appendix}

\theoremstyle{plain}
\newtheorem{theorem}{Theorem}
\newtheorem{lemma}{Lemma}
\newtheorem{corollary}{Corollary}

\newtheorem{observation}{Observation}

\def\Tr{\operatorname{Tr}}

\def\bea{\begin{eqnarray}}
\def\eea{\end{eqnarray}}
\def\ba{\begin{array}}
\def\ea{\end{array}}

\def\beq{\begin{equation}}
\def\eeq{\end{equation}}

\let\emptyset\varnothing

\begin{document}


\title{Entanglement-Constrained Quantum Metrology: Rapid Low-Entanglement Gains, Tapered High-Level Growth}


\author{Debarupa Saha}
\author{Ujjwal Sen}
\affiliation{Harish-Chandra Research Institute, A CI of Homi Bhabha National Institute, Chhatnag Road, Jhunsi, Prayagraj - 211019, India}

\begin{abstract}

It is a specific type of quantum 
correlated state that achieves optimal precision in parameter estimation under unitary encoding. 
%
%
We consider the potential experimental limitation on probe entanglement, and find a relation between achievable precision and initial probe entanglement, in both bipartite and multipartite scenarios. 
For two-qubit probes, we analytically derive an exact relationship between the entanglement-constrained optimal quantum Fisher information and the limited initial entanglement, measured via both generalized geometric measure and entanglement entropy. We demonstrate that this fundamental relationship persists across the same range of the entanglement measures even when higher-dimensional bipartite probes are considered. Furthermore, we identify the specific states that realize maximum precision in these scenarios. Additionally, by considering the  geometric measure of entanglement, we extend our approach to multiqubit probes.
We find that in every case, the optimal quantum Fisher information exhibits a universal behavior: a steep increase in the low-entanglement regime, followed by a gradual and nearly-saturated improvement as the probe entanglement approaches values close to those required for achieving the Heisenberg limit.

\end{abstract}
\maketitle
\section{Introduction}
The central goal of quantum metrology~\cite{met1, met2, intf1, M1} is to identify optimal strategies that minimize the fluctuations arising in the estimation of parameters of quantum systems.
%
%
From a practical standpoint, such tasks have numerous important applications. These include determining optical phase shifts in interferometry~\cite{Enten,Enten2,intf1, intf2,intf3, Exp2,intf5, Rev2}, measuring transition frequencies in atomic clocks~\cite{CLock1,Clock2,Cl1}, sensing magnetic fields  (magnetometry)~\cite{Magn1,Magn2}, and probing temperature (thermometry)~\cite{Th1}.
Quantum metrology has been experimentally realized on platforms such as two-qubit NMR quantum simulators~\cite{Exp1}. For further theoretical and experimental developments, we refer the reader to Refs.~\cite{Rev10,Rev3,Rev18,rev24,Rev2,Rev1}.

A range of quantum  resources including entanglement~\cite{Entr1,pro1,Entr2,GGM3}, quantum coherence~\cite{co1,CoAberg}, non-Markovianity~\cite{NMRe,NMRe1,NMRe2,NMRe3}, and squeezing~\cite{SQ1,SQ2,SQ3} has been employed to enhance precision in quantum metrology. For instance, Refs.~\cite{met1,uni1,Enten,Enten2,intf3,Enten3,Toth,Ent5,Ent4,Ent6} demonstrate the role of entangled probes in achieving quantum-enhanced metrological precision. Additional improvements in precision have been reported in Refs.~\cite{aux1,Aux14,Auxn,Aux16,Aux23,aux5}, where entanglement between the probe and an auxiliary system is exploited.
 Ref.~\cite{Ent2} gave a sufficient criterion for states to be entangled in terms of the quantum Fisher information (QFI). Coherence-assisted precision enhancement is discussed in Ref.~\cite{cohf}. In addition, Refs.~\cite{NMR,NMarkov} investigated quantum metrology in the non-Markovian paradigm. Enhanced precision in estimation using  squeezing is observed in Refs.~\cite{SQZ1,Rev10}.

Our focus in this article will mainly hover around the utility of entanglement as a resource in quantum metrology, in the case of unitary encoding.
It is well established that a specific type of entangled states, viz. the Greenberger-Horne-Zeilinger states (also referred to as cat states)~\cite{st1,st2},
provides the best optimal precision in phase estimation. This (ideal) optimal precision is known as the Heisenberg limit~(HL). On the other hand, the best  precision with input states  having zero entanglement is know as the standard quantum limit~(SQL). 

In realistic scenarios, the generation of entanglement is often constrained by experimental limitations. Under such conditions, it becomes essential to understand the ultimate precision that can be achieved when only a limited amount of entanglement is available as a resource. While several studies such as Refs.~\cite{Toth,MPE} have explored the optimal precision attainable by specific classes of multipartite entangled states, a general and quantitative relationship between the available entanglement and the best achievable precision has remained elusive.

In this article, we aim to bridge this gap by establishing a direct connection between the optimal QFI and the entanglement content of the initial probe state. Specifically, we perform an optimization of the QFI under the constraint of a fixed amount of entanglement in the input probe, thereby shedding light on how precision scales with entanglement in practical metrological settings.  The optimization is over initial states of fixed entanglement, quantified by (i) the geometric measure (GM)~\cite{GM1,GM2,GGM1}, 
(ii) the generalized geometric measure (GGM)~\cite{PhysRevA.81.012308, GM3, GGM2},  and (iii) the von Neumann entanglement entropy~\cite{Entropy}.

We begin with bipartite probes consisting of two qubits, for which we derive an exact, closed-form, relation between the entanglement-constrained optimal QFI and the fixed initial entanglement, quantified either by the GM or by the von Neumann entanglement entropy. We then show that the same relation holds for two-qudit probes across the same range of entanglement values. In each case, we identify the probe state that maximizes the QFI for a given entanglement.
Next, we extend our analysis to multipartite probes composed of $N$ qubits, using both GGM and GM to quantify entanglement. 
We employ numerical optimization to again explore how the entanglement-constrained optimal precision varies with input entanglement. 

In both bipartite and multipartite settings, and for all the measures considered, we identify a universal behavior: with input entanglement varying between values corresponding to the SQL and the HL, precision enhances steeply in the low‑entanglement regime, but then the growth in precision slows down and nearly saturates as entanglement approaches the Heisenberg limit. This shows that  preparing states with entanglement levels near the HL threshold is enough to achieve near‑optimal precision. Our results thus pave the way for metrological protocols that deliver near-maximum performance under realistic entanglement constraints.

The rest of the paper is organized as follows. In Sec.~\ref{notation}, we introduce the notations used throughout this article. In Sec.~\ref{Prem}, we briefly review the necessary preliminary concepts. Specifically, in Sec.~\ref{EM} we discuss the entanglement measures employed, and in Sec.~\ref{Unit} we provide an short overview of quantum metrology under unitary encoding, highlighting the role of entanglement. In Sec.~\ref{Qubit}, we analyze the two-qubit case, where we derive a closed-form relation between the fixed input entanglement and the corresponding optimal QFI, and also identify the optimal probe states. In Sec.~\ref{Qudit}, we extend this analysis to bipartite probes in higher dimensions. Sec.~\ref{multi} addresses the case of multipartite probe states. Finally, we summarize our findings and conclude in Sec.~\ref{Con}.

\section{Notations}
\label{notation}
Before presenting our results and observations, we introduce the notations and conventions used throughout this work to denote various physical quantities in the joint Hilbert space $(\mathbb{C}^d)^{\otimes N}$ of the multipartite probe system. Here, $d$ denotes the dimension of each individual subsystem of the probe, and $N$ represents the number of parties. The generator of the unitary evolution acting on $(\mathbb{C}^d)^{\otimes N}$ is given by a Hermitian operator denoted as $H_{N,d} = h_{N,d} \theta'$, where $\theta = \theta' t$ is the parameter to be estimated, with $t$ being the evolution time. The Hermitian operator $h_{N,d}$ is given by 
$h_{N,d} = \sum_{k=1}^N \mathcal{Z}_d^{(k)}$, 
where $\mathcal{Z}_d^{(k)}$ represents a local Hermitian component acting on the $k^{\text{th}}$ subsystem of the multipartite probe. Each $\mathcal{Z}_d^{(k)}$ has a set of eigenvalues $\{\eta_j^d\}$ and corresponding eigenvectors $\{\ket{j_d}\}$, indexed by $j = 0, 1, \ldots, d-1$. The eigenstate of $H_{N,d}$ is denoted as $\ket{kl}_d$, where $k,l = 0, 1, \ldots, d-1$. The identity operator defined on a Hilbert space $\mathbb{C}^L$ of dimension $L$ is denoted by $\mathcal{I}_L$.

 The Initial probe state defined in $(\mathbb{C}^d)^{\otimes N}$ is denoted by $\ket{\psi^{\text{in}}_{N,d}}$, where ``in'' stands for ``initial.'' And the final encoded state after the unitary evolution is denoted by $\ket{\psi^{\text{f}}_{N,d}}$, where ``f'' stands for ``final.''

 Entanglement-constrained optimal QFI corresponding to the encoded states defined on $(\mathbb{C}^d)^{\otimes N}$, at fixed GGM $G$, von Neumann entanglement entropy $S$, and GM $\mathcal{GM}$, is denoted by $Q^G_{N,d}$, $Q^S_{N,d}$, and $Q^{\mathcal{GM}}_{N,d}$, respectively. The corresponding optimal initial probe states at fixed $G$, $S$, and $\mathcal{GM}$ are represented as $\ket{\psi^o_{N,d}(G)}$ and $\ket{\psi^o_{N,d}(S)}$,  where the superscript ``o'' stands for ``optimal.'' The set of initial probe states on $(\mathbb{C}^d)^{\otimes N}$ with fixed initial entanglement $\mathcal{E}_{X}(\ket{\psi^{\text{in}}_{N,d}})$ is denoted by $\chi^{\mathcal{E}}_{N,d}$, where $X$ can refer to the GGM ($G$), von Neumann entanglement entropy ($S$), or GM, $\mathcal{GM}$.

\section{Preliminaries}
\label{Prem}
\subsection{Entanglement Measures}
\label{EM}
\noindent
In this subsection we briefly review the entanglement measures used in this article, namely the GM~\cite{GM1,GM2,GM3}, the GGM~\cite{GGM1,GGM2,GGM3}, and the von Neumann entanglement entropy~\cite{Entropy}.\\

\noindent
The geometric measure of entanglement for a pure quantum state, denoted by the density matrix $\rho = \ketbra{\psi}{\psi}$ and acting on the joint Hilbert space $(\mathbb{C}^d)^{\otimes N}$ of $N$ parties, each with local subsystem dimension $d$, is defined as follows
$$
\mathcal{E}_\mathrm{GM} = \min_{\ket{\phi} \in \chi^N_{\mathrm{sep}}} \mathcal{D}(\ket{\psi}, \ket{\phi}),
$$
where $\chi^N_{\mathrm{sep}}$ is the set of fully separable pure product states of the form $\ket{\phi} = \ket{\phi_1} \otimes \ket{\phi_2} \otimes \ldots \otimes \ket{\phi_N}$, with each $\ket{\phi_i}$, $i = 1, 2, \ldots, N$, being a pure state defined in the local Hilbert space $\mathbb{C}^d$. The distance $\mathcal{D}(\ket{\psi}, \ket{\phi})$ is taken as
$$
\mathcal{D}(\ket{\psi}, \ket{\phi}) = 1 - |\braket{\psi | \phi}|^2.
$$
Using this form, the geometric measure simplifies to
$$
\mathcal{E}_\mathrm{GM} = \min_{\ket{\phi} \in \chi^N_{\mathrm{sep}}} \left( 1 - |\braket{\psi | \phi}|^2 \right).
$$

On the other hand, the generalized geometric measure of entanglement is defined as
$$
\mathcal{E}_\mathrm{GGM} = \min_{\ket{\bar{\phi}} \in \chi^2_{\mathrm{sep}}} \mathcal{D}(\ket{\psi}, \ket{\bar{\phi}}),
$$
where $\chi^2_{\mathrm{sep}}$ denotes the set of all pure 2-separable states~\cite{ksepb,ksepa,ksep, pro1}, $\ket{\bar{\phi}} = \ket{\bar{\phi}_A}\otimes \ket{\bar{\phi}_B}$, which are separable across some bipartition $A:B$ of the total system, with $A \cup B = \{1, 2, \ldots, N\}$ and $A \cap B=\emptyset$. Thus, the GGM quantifies genuine multipartite entanglement in $\ket{\psi}$. In a computable form, it can be written as
$$
\mathcal{E}_\mathrm{GGM} = 1 - \max_{A:B} \lambda_{A:B}^2,
$$
where the maximization is over all possible bipartitions $A:B$, and $\lambda_{A:B}$ is the largest Schmidt coefficient of $\ket{\psi}$ in the corresponding bipartition. Note that for pure bipartite states, defined on $(\mathbb{C}^d)^{\otimes 2}$ with $N = 2$, the geometric measure and the generalized geometric measure are equivalent.\\

\noindent
Another important measure is the von Neumann entanglement entropy. For a pure bipartite state $\rho_{AB}$ defined on $(\mathbb{C}^d)^{\otimes 2}$, the entanglement entropy is given by 
$\mathcal{E}_S = -\Tr[\rho_{A/B} \log_2 \rho_{A/B}]$, 
where $\rho_{A/B} = \Tr_{B/A}[\rho_{AB}]$ denotes the reduced density matrix obtained by tracing out subsystem $B$ (or $A$), respectively.\\

\noindent
We employ these three entanglement measures to quantify the entanglement present in the probe states used for parameter estimation under unitary encoding. In the following subsection, we present a brief overview of unitary encoding in quantum metrology and highlight the role of entanglement in enhancing estimation precision.

\subsection{Brief Overview of Unitary Encoding and Entanglement in Optimal Probe}
\label{Unit}
\noindent
The essential step in quantum metrology is encoding the parameter of interest into the initial probe state. Depending on the parameter, this may involve open system encoding~\cite{Dcoh,Qch,Therm,op1,op,op2}. In this article, however we restrict ourselves to closed system unitary encoding~\cite{uni1,M1,uni2,Toth,COref,uni5,uni4}, where the initial probe state $\ket{\psi^{\text{in}}_{N,d}}$ evolves unitarily into the encoded state $\ket{\psi^\text{f}_{N,d}}$ as
$\ket{\psi^\text{f}_{N,d}} = U \ket{\psi^{\text{in}}_{N,d}}$,
with $U = e^{-iH_{N,d}t}$ and $i = \sqrt{-1}$. Both $\ket{\psi^{\text{in}}_{N,d}}$ and $\ket{\psi^{\text{f}}_{N,d}}$ are defined on the Hilbert space $(\mathbb{C}^d)^{\otimes N}$.
Here, $H_{N,d}$ is a Hermitian operator defined on the Hilbert space $(\mathbb{C}^d)^{\otimes N}$ and has units of inverse time, i.e., $\tau^{-1}$, where $\tau$ is the unit of time and $t$ is the duration of evolution. We adopt natural units by setting Planck’s constant $\hbar = 1$.\\

\noindent
The generator $H_{N,d}$ can be written as $H_{N,d} = h_{N,d} \theta'$, where $h_{N,d}$ is a dimensionless Hermitian operator and $\theta'$ has units of $\tau^{-1}$. Thus, the evolution becomes
$\ket{\psi^\text{f}_{N,d}} = e^{-i h_{N,d} \theta} \ket{\psi^{\text{in}}_{N,d}}$,
where $\theta = \theta' t$ is the parameter to be estimated.\\

\noindent
After encoding, the next step is to measure the state $\ket{\psi^\text{f}_{N,d}}$ using a set of Positive‐Operator Valued Measure (POVM) operators $\{M_x\}$, with $M_x \geq 0$ and $\sum_x M_x=\mathcal{I}_{Nd}$. The probability of obtaining outcome $x$ is given by $\bra{\psi^\text{f}_{N,d}} M_x \ket{\psi^\text{f}_{N,d}}$. Based on these probabilities and an estimator $\Theta(x)$, the parameter $\theta$ is inferred.\\

\noindent
An estimator is said to be \textit{unbiased} if its expectation value over all outcomes equals the true value of the parameter, i.e.,
$\mathbb{E}(\Theta) = \sum_{x} \Theta(x) \bra{\psi^f_{N,d}} M_x \ket{\psi^\text{f}_{N,d}}=\theta$.
For unbiased estimators, the variance in estimating $\theta$ is bounded below by the quantum Cramér-Rao bound~\cite{Wooters,Braunstein1,holevo,Rev2,review1} as
$$\Delta^2 \theta \geq \frac{1}{m Q(\ket{\psi^\text{f}_{N,d}})},$$
where $m$ is the number of times the measurements have been repeated with same initial state, $\ket{\psi^\text{in}_{N,d}}$ and same measurement setting, $\{M_x\}$ and $Q(\ket{\psi^\text{f}_{N,d}})$ is the QFI for the state $\ket{\psi^\text{f}_{N,d}}$. For pure states and unitary encoding, the QFI takes the simple form
$Q(\ket{\psi^\text{f}_{N,d}}) = 4 \Delta^2h_{N,d}^{\text{in}}$,
where $\Delta^2h_{N,d}^{\text{in}}$ is the variance of the operator $h_{N,d}$ with respect to the initial state $\ket{\psi^{\text{in}}_{N,d}}$.\\

\noindent
To achieve minimum variance, one must further optimize $Q(\ket{\psi^\text{f}_{N,d}}$ over all possible input states. As shown in Ref.~\cite{M1}, when the QFI is maximized over all pure product states in $(\mathbb{C}^d)^{\otimes N}$, it scales linearly with $N$. This limit is known as the SQL. On the other hand, optimization over all pure quantum states yields a quadratic scaling of QFI with $N$, i.e., $Q(\ket{\psi^\text{f}_{N,d}}) \sim N^2$, achievable with special genuinely multipartite entangled states known as cat states. This defines the HL.\\

\noindent
Thus, entanglement plays a crucial role in enhancing precision under unitary encoding. Several studies, including Refs.~\cite{Toth,MPE}, have investigated this further. In particular, they derived bounds on the QFI by optimizing over states with fixed entanglement depth (i.e., $k$-separability or $k$-producibility)~\cite{ksep, ksepa, pro1}.\\

\noindent
A pure state $\ket{\psi^{\text{in}}_{N,d}}$ in $(\mathbb{C}^d)^{\otimes N}$ is said to be $k$-producible if it can be written as
$\ket{\psi^{\text{in}}_{N,d}} = \ket{\phi_{A_1}} \otimes \ket{\phi_{A_2}} \otimes \cdots \otimes \ket{\phi_{A_m}}$,
such that $A_1 \cup A_2 \cup \cdots \cup A_m = \{1,2,\dots,N\}$ and each $\ket{\phi_{A_i}}$ is a state defined on at most $k$ subsystems. It is important to note that all $k$-producible states have zero generalized geometric measure.\\

\noindent
Other works (e.g., Ref.~\cite{Ent2}) have established necessary and sufficient conditions for a state to be entangled in terms of its QFI. Ref.~\cite{Ent3} demonstrated that even positive partial transpose (PPT) states can be useful in quantum metrology. For additional insights into the connection between entanglement and meteorological advantage under unitary encoding, we refer the reader to Refs.~\cite{met1,uni1,Enten,Enten2,intf3,Enten3,Ent5,Ent4,Ent6}.\\

\noindent
In summary, these studies underscore the necessity of entanglement for achieving quantum-enhanced metrology. However, a precise and quantitative relationship between the amount of entanglement and the achievable precision remains largely unexplored. In this work, we address this gap by examining how the entanglement-constrained optimal QFI correlates with the amount of available entanglement. In particular, we optimize the QFI over states with fixed input entanglement, quantified using the entanglement measures introduced in Sec.~\ref{EM}. \\

\noindent
We begin our analysis in the next section by investigating the variation of entanglement-constrained QFI for two-qubit pure bipartite probe states.

\section{Entanglement vs QFI: The case of two qubits}
\label{Qubit}
In this section, we analytically derive an exact relation between the entanglement-constrained QFI and the initial entanglement of two-qubit probe states.

We consider a bipartite probe ($N = 2$), with each party having local dimension $d = 2$. Following the notation introduced in Sec.~\ref{notation}, and without loss of generality, we take the generator of the unitary evolution to be the Hamiltonian $H_{2,2} = h_{2,2} \theta'$, where
\begin{equation}
h_{2,2} = \sum_{k=1}^2 \sigma_z^{k}.
\label{h}
\end{equation}
Note according to the notation used in Sec.~\ref{notation}. Here we have put $\mathcal{Z}_2=\sigma_z$. Where $\sigma_z$ denotes the usual Pauli matrix and the superscript $k$ indicates that $\sigma_z$ acts on the $k^{\text{th}}$ qubit. The unitary evolution $U = e^{-iH_{2,2}t}$ transforms an arbitrary initial state $\ket{\psi^{\text{in}}_{2,2}}$ into a final state $\ket{\psi^{\text{f}}_{2,2}}$, such that
\begin{equation}
\ket{\psi^{\text{f}}_{2,2}} = e^{-i h_{2,2} \theta} \ket{\psi^{\text{in}}_{2,2}}.
\label{en}
\end{equation}
The parameter $\theta=\theta't$ is thus encoded in the evolved state $\ket{\psi^{\text{f}}_{2,2}}$.

To estimate $\theta$, one performs a measurement on $\ket{\psi^{\text{f}}_{2,2}}$ and infers its value from the resulting probability distribution. As discussed in Sec.~\ref{Unit}, for an unbiased estimator, the variance $\Delta^2\theta$ is bounded below by
\begin{equation*}
\Delta^2\theta \geq \frac{1}{m Q(\ket{\psi^{\text{f}}_{2,2}})},
\end{equation*}
where $m$ is the number of repetitions of the measurement. Since we focus on single-shot estimation, we set $m = 1$. (Note that QFI is additive under repetition.)

For unitary encoding with pure input states, the QFI is given by
\begin{equation}
Q(\ket{\psi^{\text{f}}_{2,2}}) = 4 \Delta^2 h_{2,2}^{\text{in}},
\label{QFI}
\end{equation}
where $\Delta^2 h_{2,2}^{\text{in}} = \langle h_{2,2}^2 \rangle - \langle h_{2,2} \rangle^2$ is the variance in the generator evaluated in the initial state $\ket{\psi^{\text{in}}_{2,2}}$.

To determine the maximal achievable precision under fixed entanglement, we optimize $Q(\ket{\psi^{\text{f}}_{2,2}})$ over all input states with a given entanglement. Let $\chi_{2,2}^{\mathcal{E}}$ denote the set of all pure states in $\mathbb{C}^2 \otimes \mathbb{C}^2$ with entanglement $\mathcal{E}$.   We define the entanglement constrained optimal QFI at fixed entanglement $\mathcal{E}$ as 
\begin{equation}
Q_\mathcal{E} := \max_{\ket{\psi^{\text{in}}_{2,2}} \in \chi^\mathcal{E}_{2,2}} 4 \Delta^2 h_{2,2}^{\text{in}}.
\end{equation}

This consideration naturally lead to two central questions: Does $Q_\mathcal{E}$ depend explicitly on the entanglement $\mathcal{E}$ of the initial probe state? And if so, what is the exact functional dependence of $Q_\mathcal{E}$ on $\mathcal{E}$?

Answering these questions is crucial for understanding the fundamental role of entanglement in quantum metrology. If a clear and predictable relationship exists, it could guide the efficient design of metrological protocols that achieve near-optimal precision without requiring maximal entanglement. To this end, we analyze the dependence of the QFI on initial entanglement using two well-established quantifiers: GGM and the von Neumann entanglement entropy of entanglement.

In the following subsection, we focus on GGM and present an exact analytical result: a theorem that expresses the entanglement constrained optimal QFI as an explicit function of the entanglement $\mathcal{E}$. This result provides a clear and quantitative link between the entanglement content of the input probe and the maximum achievable precision in estimating the parameter $\theta$.

\subsection{Entanglement-Constrained Optimal QFI: GGM as Measure}
\label{GMd=2}
\begin{theorem}
For a unitary encoding of the type given in Eq.~\eqref{en}, if $G$ denotes the GGM of the initial bipartite system probe, $\ket{\psi^{\text{in}}_{2,2}}$, in $\mathbb{C}^2 \otimes \mathbb{C}^2$, and $\chi_{2,2}^G$ represents the set of all such pure bipartite states with fixed GGM $G$, then the QFI $$Q^G_{2,2} := \underset{\ket{\psi^{\text{in}}_{2,2}} \in \chi_{2,2}^G}{\max} 4 \Delta^2 h_{2,2}^{in},$$ with respect to the parameter $\theta$, maximized over all initial probes with fixed GGM $G$, is given by
\begin{equation}
    Q^G_{2,2} = 8 \left( 1 + \sqrt{1 - (1 - 2G)^2} \right).
    \label{Qf}
\end{equation}
\label{t1}
\end{theorem}

\begin{proof}
To demonstrate the above statement, we decompose the initial probe state, denoted as $\ket{\psi^{\text{in}}_{2,2}}$, into the eigenbasis of the Hamiltonian $h_{2,2}$ as
\begin{equation}
\ket{\psi^{\text{in}}_{2,2}} = \sum_{p=0}^{3} \sqrt{\omega_p} \ket{p}, \label{eq1}
\end{equation}
with $0 \leq \omega_p \leq 1$, where $p = 0, 1, 2, 3$ serves as a shorthand for the pair of indices $(rs)$, with $r, s = 0, 1$. That is, $p = 0, 1, 2, 3$ corresponds to $rs = 00, 01, 10, 11$. The basis $\{\ket{0}_2, \ket{1}_2\}$ denotes the eigenbasis of $\sigma_z$, with eigenvalues $\eta_0^2 = -1$ and $\eta_1^2 = 1$, respectively. The set ${\ket{qr}_2} = {\ket{p}}$ forms the eigenbasis of $h_{2,2}$, with corresponding eigenvalues $E_p = \eta_r^2 + \eta_s^2$, for $r, s = 0, 1$.

It is worth noting that one could include local phase terms in the above decomposition. However, since the QFI depends only on the values of $\omega_p$—i.e., the probabilities of the state being in a particular eigenstate $\ket{p}$—these phases do not affect the precision. Therefore, without loss of generality, we consider the form given in Eq.~\eqref{eq1}, omitting local phase terms. Moreover, states without relative phases are experimentally favorable and less susceptible to environmental noise~\cite{Exst,Exst2,Exst3}, which further motivates our choice to consider initial states with no relative phase.

The GGM of such states, calculated in terms of the coefficients $\omega_p$, is given by

\begin{equation}
\mathcal{E}_G(\ket{\psi^{\text{in}}_{2,2}}) = \frac{1}{2} \left[ 1 - \sqrt{1 - 4 \left( \sqrt{\omega_1 \omega_2} - \sqrt{\omega_0 \omega_3} \right)^2} \right]. \label{GGM}
\end{equation}

The variance of $h_{2,2}$ with respect to the state $\ket{\psi^{\text{in}}_{2,2}}$ is given by

\begin{equation}
\Delta^2 h_{2,2}^{\text{in}} = \sum_{p} \omega_p (E_p)^2 - \left( \sum_{p} \omega_p E_p \right)^2. \label{vh}
\end{equation}

Our goal is to optimize $\Delta^2 h_{2,2}^{\text{in}}$, and hence the QFI, over the probability distribution $\{\omega_p\}$. Since $\omega_p$ denotes the probability that the initial probe is in the energy eigenstate $\ket{p}$, the normalization condition $\sum_p \omega_p = 1$ must hold. Also the optimization is subject to the constraint that the initial state belongs to the set $\chi^G_{2,2}$, comprising all pure bipartite states with a fixed GGM, i.e., $\mathcal{E}_G(\ket{\psi^{\text{in}}_{2,2}}) = G$. For two-qubit pure states, $G$ ranges from $0$ to ${1}/{2}$.

Thus, the task of maximizing the QFI,

\begin{equation}
Q^G_{2,2} = \max_{\ket{\psi^{\text{in}}_{2,2}} \in \chi^\mathcal{G}_{2,2}} 4\Delta^2 h_{2,2}^{\text{in}},
\label{QFIG}
\end{equation}

is a constrained optimization problem with two constraints. First, we examine the optimal QFI for states with fixed $G$, lying in the interior of the probability space $\{\omega_p\}$, defined by the conditions $0<\omega_p \leq 1$ for all $p$ and $\sum_p \omega_p = 1$. Later, in Appendix~\ref{Bound}, we show that the optimum indeed always lies in the interior of the probability space, and that no extrema can be found at the boundary (where at least one of the $\omega_p$ vanishes, while still satisfying $\sum_p \omega_p = 1$).

To find the optimum in the interior 
we employ the method of Lagrange multipliers, defining the Lagrangian

\begin{equation}
\mathcal{L} = \Delta^2 h_{2,2}^{\text{in}} + \mathcal{K}_1 \left( \sum_p \omega_p - 1 \right) + \mathcal{K}_2 \left( \mathcal{E}_G(\ket{\psi^{\text{in}}_{2,2}}) - G \right),
\label{Lag1}
\end{equation}

where $\mathcal{K}_1$ and $\mathcal{K}_2$ are the Lagrange multipliers. We find the stationary points of $\mathcal{L}$ by substituting Eqs.~\eqref{GGM} and \eqref{vh}, and setting the partial derivatives with respect to $\omega_p$, $\forall p$ $\mathcal{K}_1$, and $\mathcal{K}_2$ to zero. This yields a set of six equations, Eqs.\eqref{eq6a}–\eqref{eq6f}, which are presented in Appendix~\ref{A1}.

Note that if $\left( \sqrt{\omega_1 \omega_2} - \sqrt{\omega_0 \omega_3} \right) = {1}/{2}$, then $\Lambda$ given in Eq.~\ref{lA} diverges, and the Lagrange multiplier method breaks down. This corresponds to the extreme case $\mathcal{E}_G(\ket{\psi^{\text{in}}_{2,2}}) = {1}/{2}$. On the other hand, for $\Lambda = 0$—i.e., $\mathcal{E}_G(\ket{\psi^{\text{in}}_{2,2}}) = 0$ the set of equations becomes inconsistent. Therefore, the Lagrange method is valid only for $0 < \mathcal{E}_G < {1}/{2}$. The extreme cases $G = 0$ and $G = {1}/{2}$ will be treated separately.

Using the eigenvalues $E_0 = -2$, $E_1 = E_2 = 0$, and $E_3 = 2$, subtracting Eq.\eqref{eq6c} from Eq.\eqref{eq6d} gives:

\begin{equation*}
\mathcal{K}_2 \Lambda \left( \sqrt{\frac{\omega_2}{\omega_1}} - \sqrt{\frac{\omega_1}{\omega_2}} \right) = 0.
\end{equation*}

 If $\mathcal{K}_2=0$ or $\Lambda=0$, the set of Eqs.~\eqref{eq6a}–\eqref{eq6d} becomes inconsistent hence, $\mathcal{K}_2 \neq 0$ and $\Lambda \neq 0$, this implies $\omega_1 = \omega_2$.

Substituting $\omega_1 = \omega_2$ into Eqs.~\eqref{eq6a}–\eqref{eq6d} and eliminating $\mathcal{K}_1$ and $\mathcal{K}_2$ yields

\begin{equation}
\left( \sqrt{\omega_3} - \sqrt{\omega_0} \right) \left( 1 - 2 \left( \sqrt{\omega_3} + \sqrt{\omega_0} \right)^2 \right) = 0.
\label{2nd}
\end{equation}

This condition is satisfied if either or both

\begin{itemize}
\item $\omega_3 = \omega_0$, 
\item $\left( \sqrt{\omega_3} + \sqrt{\omega_0} \right)^2 = \frac{1}{2}$.
\end{itemize}

The second condition corresponds to the specific case $G = \frac{1}{2}(1 - \sqrt{\frac{3}{4}})$. Hence, for all other values of $G$ (except $G = 0$, $G = 1/2$ and this special case), the Lagrange method yields a unique solution with $\omega_0 = \omega_3$ and $\omega_1 = \omega_2$.

Now, consider the special case $G = \frac{1}{2}(1 - \sqrt{\frac{3}{4}})$. Given that $\omega_1 = \omega_2$ still holds, we use normalization $\sum_p \omega_p = 1$ and Eq.~\eqref{GGM} to show that this $G$ requires one of the following:

\begin{itemize}
\item $\left( \sqrt{\omega_0} + \sqrt{\omega_3} \right)^2 = \frac{1}{2}$,
\item $\left( \sqrt{\omega_0} + \sqrt{\omega_3} \right)^2 = \frac{3}{2}$.
\end{itemize}

Among these, the first condition automatically satisfies Eq.\eqref{2nd}. However, calculating the QFI for both cases shows that the QFI corresponding to the first condition is $Q_1 = 4$, while that of the second condition is $Q_2 = 12$. Since $Q_2 > Q_1$, the second condition yields the higher QFI. Nevertheless, our objective is to identify the optimal QFI subject to all constraints, including Eq.~\eqref{2nd}. This leads to a unique solution with $\omega_0 = \omega_3$ and $\omega_1 = \omega_2$ such that $\left(\sqrt{\omega_0} + \sqrt{\omega_3} \right)^2 = {3}/{2}$.

We now turn to the edge cases $G = 0$ and $G = {1}/{2}$, corresponding to product and maximally entangled states, respectively. Ref.~\cite{M1} has shown that the optimal state with zero entanglement is a tensor product of maximally coherent single-qubit states, yielding $\omega_0 = \omega_1 = \omega_2 = \omega_3 = {1}/{4}$. The corresponding QFI defines the SQL. On the other hand, optimization of QFI over all maximally entangled states with $G = {1}/{2}$ yields the HL, with optimal states having $\omega_0 = \omega_3 = {1}/{2}$ and $\omega_1 = \omega_2 = 0$.

Thus, combining the results of Ref.\cite{M1} with our analysis, we conclude that for all $G \in [0, {1}/{2}]$, the optimal probe state maximizing the QFI at fixed GGM must be of the form given in Eq.\eqref{eq1}, with $\omega_0 = \omega_3$ and $\omega_1 = \omega_2$.

Imposing this form and the normalization condition $\sum_{p} \omega_{p} = 1$, we express $G$ and $Q^G_{2,2}$ in terms of $\omega_0$ as follows
\begin{eqnarray}
\begin{split}
G &= \frac{1 - \sqrt{1 - (1 - 4\omega_0)^2}}{2}, \quad
Q^G_{2,2} &= 32\omega_0.
\label{Gw}
\end{split}
\end{eqnarray}

Eliminating $\omega_0$ yields a closed-form expression for the QFI in terms of the generalized geometric measure of  entanglement as
\begin{equation*}
Q^G_{2,2} = 8 \left(1 + \sqrt{1 - (1 - 2G)^2} \right).
\end{equation*}

This completes the proof of Theorem~\ref{t1}.
\end{proof}

Next, we state a corollary that follows directly from this theorem.

\begin{corollary}
For a unitary encoding of the form given in Eq.~\eqref{en}, the optimal initial state that yields the maximum quantum Fisher information $Q^G_{2,2}$ given in Eq.~\eqref{Qf}, for a fixed geometric measure of entanglement $G \in [0, {1}/{2}]$, is given by
\begin{eqnarray}
\begin{split}
\ket{\psi_{\text{2,2}}^o(G)} &= \sqrt{\frac{1+\sqrt{1-(1-2G)^2}}{4}}(\ket{00}_2+\ket{11}_2) \\
&\quad + \sqrt{\frac{1-\sqrt{1-(1-2G)^2}}{4}}(\ket{01}_2+\ket{10}_2).
\end{split}
\end{eqnarray}
\end{corollary}

Therefore, fixing the GGM of the initial probe state to $G$, the minimum standard deviation $\Delta \theta_G$ in estimating the parameter $\theta$ as given by the quantum Cramér-Rao bound is
\begin{equation*}
\Delta \theta_G = \left[ 8 \left(1 + \sqrt{1 - (1 - 2G)^2} \right) \right]^{-1/2}.
\end{equation*}

 In the next section, we show that this relationship holds even for bipartite probes with local dimensions $d = 3, 4, 5$. As evident from the plot, $\Delta \theta_G$ decreases monotonically with increasing entanglement. The enhancement in precision is most significant in the low-entanglement regime, while the rate of improvement diminishes as $G$ approaches its maximum value, that corresponds to the HL.

Note that our findings have direct implications for practical scenarios such as estimating the transition frequency of an atomic clock~\cite{Cl1}. Specifically, consider a two-qubit atomic clock governed by a Hamiltonian of the form given in Eq.~\eqref{h}. Our analysis reveals that the optimal precision in estimating the transition frequency $\theta' = \theta$ (with evolution time $t = 1$) is directly linked to the entanglement present in the clock’s initial state. In fact, we provide a closed-form expression relating the optimal precision to the amount of entanglement in the clock, showing that precision improves with increasing entanglement. To illustrate this idea, we include a schematic in Fig.~\ref{clock}. The entanglement content of the clock is metaphorically represented by the distance between two parties, Alice and Bob. In the leftmost clock, the parties are far apart, indicating no entanglement. In the middle clock, they are closer, reflecting partial entanglement. The rightmost clock shows Alice and Bob holding hands, symbolizing a maximally entangled state. This visual progression clearly conveys the central message: higher entanglement enables higher precision in parameter estimation.

\begin{figure}
\includegraphics[scale=0.3]{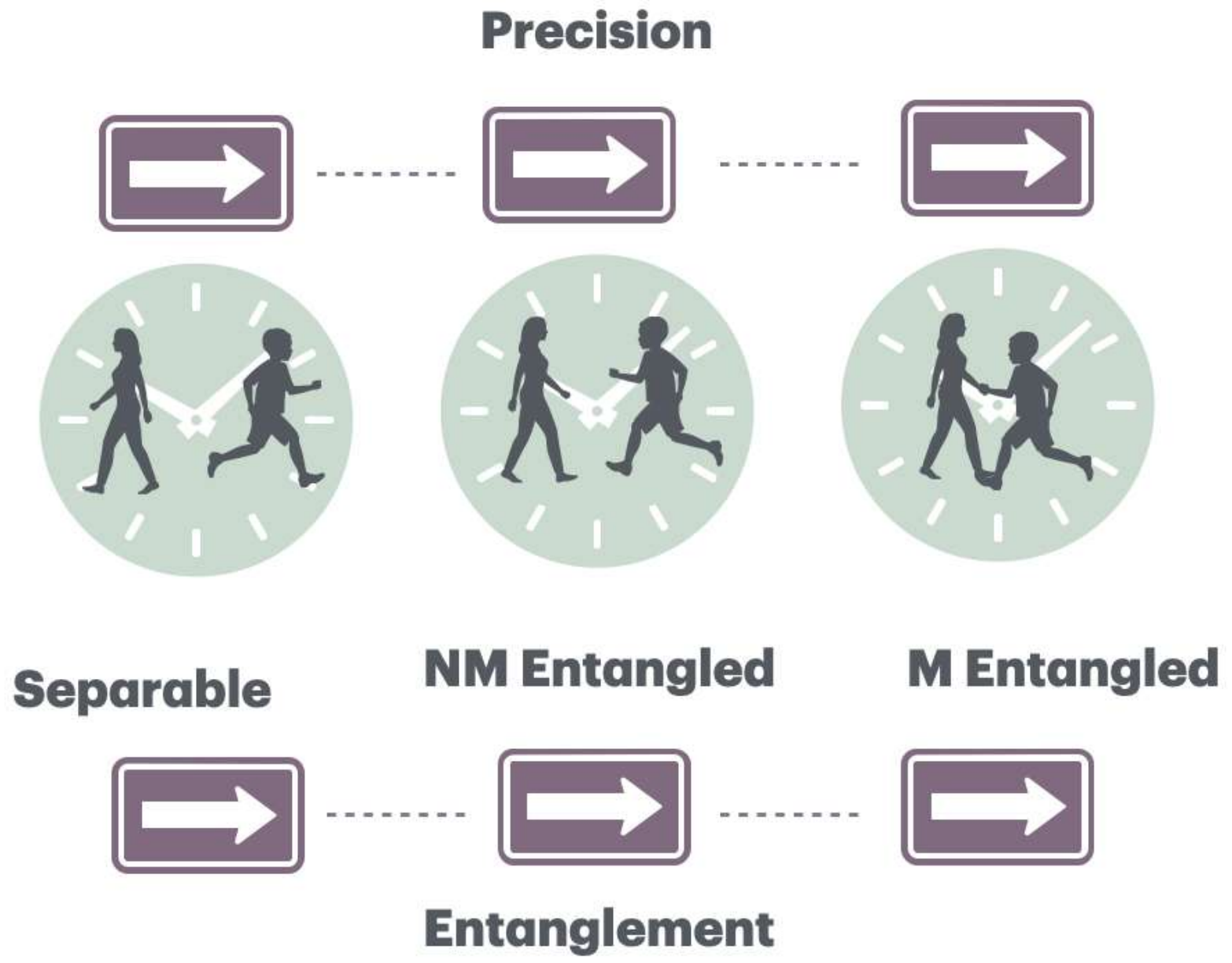}
\caption{\textbf{Role of entanglement in frequency estimation of atomic clocks}. The schematic depicts the role of entanglement  in determining the precision of estimating the transition frequency in a bipartite atomic clock composed of two qubits. The entanglement content of the clock is represented by the distance between the two parties, Alice and Bob. In the leftmost clock, where the two parties are shown to be far apart, there is no entanglement.
In the middle clock, where they are moving closer, the entanglement is nonzero.
The rightmost clock represents the maximally entangled case, visually depicted by the two parties holding hands. The figure illustrates the fact that as the degree of entanglement increases from left to right, the precision of estimation also improves.}
\label{clock}
\end{figure}

In the next subsection, we consider another type of entanglement measure, the entanglement entropy. We optimize the QFI under a fixed entanglement entropy constraint and present a theorem that establishes a relationship between the optimal QFI and fixed entanglement entropy for two-qubit bipartite states. The theorem is presented below

\subsection{Entanglement-Constrained Optimal QFI: Entanglement Entropy as Measure}
\label{S=2}
\begin{theorem}
\label{t2}
For a unitary encoding of the form given in Eq.~\eqref{en}, let $S$ denote the entanglement entropy of the initial bipartite probe state $\ket{\psi^{\text{in}}_{2,2}} \in \mathbb{C}^2 \otimes \mathbb{C}^2$. Define $\chi^S_{2,2}$ as the set of all pure bipartite states with fixed entanglement entropy $S$. Then, the quantum Fisher information optimized over all initial probes with fixed $S$, defined as $$Q^S_{2,2} := \underset{\ket{\psi^{\text{in}}_{2,2}} \in \chi_{2,2}^G}{\max} 4 \Delta^2 h_{2,2}^{in},$$ with respect to the parameter $\theta$, is related to $S$ as follows 
\begin{equation}
S = -\lambda_{Q^S_{2,2}} \log_2(\lambda_{Q^S_{2,2}}) - (1 - \lambda_{Q^S_{2,2}}) \log_2(1 - \lambda_{Q^S_{2,2}}).
\label{QS}
\end{equation}
Where $\lambda_{Q^S_{2,2}}$ is given as,
\begin{equation*}
    \lambda_{Q^S_{2,2}}=\frac{1-\sqrt{1-(1-\frac{Q^S_{2,2}}{8})^2}}{2}.
\end{equation*}
\end{theorem}
\begin{proof}
   To prove the above theorem, similar to the approach used in Theorem~\ref{t1}, we consider an initial probe state of the form given in Eq.~\eqref{eq1}.  

The entanglement entropy of such states is given by
\begin{equation}
\mathcal{E}_S(\ket{\psi^{\text{in}}_{2,2}}) = -p_s \log_2 p_s - (1 - p_s) \log_2(1 - p_s),
\end{equation}
where 
$$
p_s = \frac{1 - \sqrt{1 - 4(\sqrt{\omega_1 \omega_2} - \sqrt{\omega_0 \omega_3})^2}}{2},
$$
is one of the eigenvalues of the reduced density matrix $\Tr_{A/B}\big[\ketbra{\psi^{\text{in}}_{2,2}}\big]$. Here, $\Tr_{A/B}$ denotes the partial trace over either the first or the second subsystem.

Now consider a set $\chi^S_{2,2}$ of states having a fixed entanglement entropy $\mathcal{E}_S(\ket{\psi^{\text{in}}_{2,2}}) = S$. The goal is to optimize the quantum Fisher information (QFI) over all states in $\chi^S_{2,2}$. Thus, the optimal QFI in this scenario is defined as
\begin{equation*}
Q^S_{2,2} = \max_{\ket{\psi^{\text{in}}_{2,2}} \in \chi^S_{2,2}} 4 \Delta^2 h_{2,2}^{\text{in}}.
\end{equation*}

Similar to Theorem~\ref{t1}, this optimization must be performed under two constraints: the normalization condition $\sum_p \omega_p = 1$, and the fixed entanglement entropy constraint. To address this, we again employ the method of Lagrange multipliers. Following the methodology of Theorem~\ref{t1}, we initially investigate the optimal QFI for states with fixed $S$ located in the interior of the probability space. In Appendix~\ref{Bound}, we further establish that the optimal solution consistently lies in the interior, with no extrema arising at the boundary.

The Lagrangian for this problem becomes

\begin{equation}
\mathcal{L} = \Delta^2 h_{2,2}^\text{in} + \mathcal{K}_3 \left( \sum_p \omega_p - 1 \right) + \mathcal{K}_4 \Big( \mathcal{E}_S({\ket{\psi^{\text{in}}_{2,2}}}) - S \Big),
\label{Lag2}
\end{equation}
where $\mathcal{K}_3$ and $\mathcal{K}_4$ are the Lagrange multipliers. Derivative of the Lagrangian with respect to the parameters $\omega_p$, with $p=0,1,2,3$ leads to four sets of Eqs.~\eqref{eq7a}-~\eqref{eq7d}

One needs to solve the set of above 4 equation together with the normalization constraint and fixed entanglement entropy constraint to get the optimal set of $\omega_p$ that maximizes the QFI as given. Note that  when $\sqrt{\omega_1\omega_2}-\sqrt{\omega_0\omega_3}={1}/{2}$. $\mathcal{M}$ presented in Eq.~\eqref{Ma} diverges this corresponds to a special case when $S=1$. Also when $\mathcal{M}=0$, the above Eqs.~\eqref{eq7a}-~\eqref{eq7c}, leads to an inconsistency. This corresponds to the case of $S=0$. Thus both for $S=0$ and $S=1$, the concerned optimization cannot be performed using Lagrange multiplier method. However we employ Lagrange multiplier method for $0<S<1$. The case of $S=0$ and $S=1$ is discussed later in this section.

Also note that the problem of solving this set of four equations is similar to the one presented in Theorem~\ref{t1}. Therefore, using similar reasoning as in Theorem~\ref{t1}, and applying Eqs.\eqref{eq7c} and\eqref{eq7d}, one can argue that $\mathcal{K}_4 \neq 0$ and $\mathcal{M} \neq 0$; otherwise, this set of equations would lead to an inconsistency. Thus either of the following conditions must be true that optimal solution must have $\omega_1=\omega_2$.
 Next, using this and further solving the Eqs.\eqref{eq7a}-\eqref{eq7d} to eliminate $\mathcal{K}_3$ and $\mathcal{K}_4$.  We find that that either or both of the  following condition needs to be satisfied for the probe of the form given in Eq.\eqref{eq1}, to be optimum
\begin{itemize}
      \item $\omega_3 = \omega_0$,
    \item $1 - 2\left( \sqrt{\omega_3} + \sqrt{\omega_0} \right)^2 = 0$.
\end{itemize}
Again the second condition corresponds to a specific value or entanglement entropy, $\mathcal{E}_S(\ket{\psi^{\text{in}}_{2,2}})= S_f=-p_s \log_2 p_s - (1 - p_s) \log_2(1 - p_s)$. With $p_s= \frac{1}{2}\left(1 - \sqrt{\frac{3}{4}}\right)$. Note that $\mathcal{E}_S(\ket{\psi^{\text{in}}_{2,2}})=S_f$ requires either of the following to be true.
\begin{itemize}
    \item $\left(\sqrt{\omega_0} + \sqrt{\omega_3} \right)^2 = \frac{1}{2}$,
    \item $\left(\sqrt{\omega_0} + \sqrt{\omega_3} \right)^2 = \frac{3}{2}$.
\end{itemize}
By similar reasoning as given in Theorem~\ref{t1}, one can verify that the QFI corresponding to the second condition is greater than that obtained from the first. Since we aim to maximize the QFI, the second condition is preferred. However, this choice must also be consistent with the constraints derived from the set of Eqs.\eqref{eq7a}–\eqref{eq7d}. Therefore, when the initial entanglement is fixed to $S = S_f$, one must have $\omega_3 = \omega_0$ and $\omega_1 = \omega_2$ to ensure the validity of Eqs.\eqref{eq7a}–\eqref{eq7d}.

Next, for $S \neq S_f$ and $S \in (0,1)$, the only way the set of equations obtained via the Lagrange multiplier method can be satisfied is again when $\omega_0 = \omega_3$ and $\omega_1 = \omega_2$.

Finally, the cases $S = 0$ and $S = 1$ correspond to the SQL and HL, respectively, where the optimal states are known to have the form $\omega_0 = \omega_3$ and $\omega_1 = \omega_2$, as discussed in Theorem~\ref{t1}. Combining this with our analysis, we obtain the entanglement-constrained optimal QFI for fixed $S \in [0,1]$ as
\begin{equation*}
    Q^S_{2,2}=32\omega_0,
    S= -p_s \log_2 p_s - (1 - p_s) \log_2(1 - p_s).
\end{equation*}
with 
$$p_S=\frac{1 - \sqrt{1 - (1 - 4\omega_0 )^2}}{2}.$$

Combining the above two equations we get
\begin{equation}
S = -\lambda_{Q^S_{2,2}} \log_2(\lambda_{Q^S_{2,2}}) - (1 - \lambda_{Q^S_{2,2}}) \log_2(1 - \lambda_{Q^S_{2,2}}).
\label{QS}
\end{equation}
Where $\lambda_{Q^S_{2,2}}$ is given as,
\begin{equation*}
    \lambda_{Q^S_{2,2}}=\frac{1-\sqrt{1-(1-\frac{Q^S_{2,2}}{8})^2}}{2}.
\end{equation*}
This completes the proof of Theorem~\ref{t2}.
\end{proof}
Next we present a corollary that directly follows from this theorem.
\begin{corollary}
    For a unitary encoding of the type given in Eq.~\eqref{en}, to attain the maximum quantum Fisher information corresponding to the parameter $\theta$, while ensuring that the entanglement entropy  of the initial state is fixed at $S= -p_S \log_2 p_S - (1 - p_S) \log_2(1 - p_S),$ with  $p_S \in [0,{1}/{2}]$, then the optimal initial state,  must have the following form:
\begin{eqnarray}
\begin{split}
\ket{\psi_{\text{2,2}}^o(S)} &= \sqrt{\frac{1+\sqrt{1-(1-2p_S)^2}}{4}}(\ket{00}_2+\ket{11}_2) \\
&\quad + \sqrt{\frac{1-\sqrt{1-(1-2p_S)^2}}{4}}(\ket{01}_2+\ket{10}_2).
\end{split}
\end{eqnarray}
\end{corollary}

 Thus, we have obtained an exact relation between the optimal QFI and the fixed entanglement content of the initial probe state in $\mathbb{C}^2 \otimes \mathbb{C}^2$, considering entanglement quantified by both the GGM and the von Neumann entanglement entropy. In the following subsection, we investigate whether such a relation holds for higher-dimensional probes as well. We demonstrate that, within certain ranges of the entanglement measures, these relations remain valid even in higher dimensions. This range is marked by the SQL and HL. The detailed analysis is presented below.


\section{Higher-Dimensional Bipartite Probes}
\label{Qudit}

We now consider a bipartite probe system with subsystem dimension $d > 2$. In particular, we focus on the cases $d = 3, 4, 5$. The encoding Hamiltonian for these cases is denoted as $H_{2,d}$, following the notation introduced in Sec.~\ref{notation}, where we set $N = 2$, i.e., the probe consists of two parties.

For a given $d$, each subsystem of the bipartite probe can be regarded as a quantum particle with spin $s = ({d - 1})/{2}$. Accordingly, we consider the local Hamiltonian to be $\mathcal{Z}_d$, which corresponds to the spin angular momentum operator along the $z$-axis. This choice implies that the eigenvalues $\eta_j^d$ of the local Hamiltonian, indexed by $j = 0, 1, \ldots, d - 1$, are uniformly spaced and range from $-s$ to $s$.

To normalize the eigenvalue spectrum across different dimensions, for a given dimension $d$, we rescale the eigenvalues by dividing each by $\left| d - 1 \right|/2$, so that the resulting set of scaled eigenvalues $\{ \bar{\eta}_j^d \}$ spans the interval $[-1, 1]$ with spacing ${1}/{s}$. We denote the Hamiltonian with these rescaled eigenvalues as $\mathcal{\bar{Z}}_d$. This scaling ensures that the maximum and minimum eigenvalues of the local Hamiltonian exactly match those of the qubit case.

Remarkably, this rescaling leads to an exact overlap of the quantum Fisher information versus geometric entanglement $G \in [0, {1}/{2}]$ and entanglement entropy $S \in [0, 1]$ plots for all $d = 3, 4, 5$. In fact, the precision values become exactly equal to those of the qubit case, as we shall discuss later in this subsection.

Thus, the encoding Hamiltonian can be written as 
\begin{eqnarray*}
\begin{split}
    H_{2,d} &= h_{2,d}\, \theta' \\
    &= \left( \mathcal{\bar{Z}}_d \otimes \mathcal{I}_d + \mathcal{I}_d \otimes \mathcal{\bar{Z}}_d \right) \theta'.
\end{split}
\end{eqnarray*}

For unitary evolution governed by such a Hamiltonian, if the initial state of the system for a given $d$ is $\ket{\psi_{2,d}^{\text{in}}}$, then the final state of the system after evolution for unit time is given by
\begin{eqnarray}
\begin{split}
    \ket{\psi_{2,d}^{\text{f}}} = e^{-i H_{2,d}t} \ket{\psi_{2,d}^{\text{in}}}.
\end{split}
\label{U_d}
\end{eqnarray}
 For this type of unitary encoding and pure probe state, the QFI computed for subsystem dimension $d$ can be written as
$$
Q(\ket{\psi_{2,d}^{\text{f}}}) = 4\Delta^2 h_{2,d}^{\text{in}},
$$
where $\Delta^2 h_{2,d}$ is the variance of $h_{2,d}$ computed using the initial state $\ket{\psi_{2,d}^{\text{in}}}$.

The goal is to optimize $Q_d$ over all choices of initial states having fixed entanglement $\mathcal{E}$, and to identify whether there exists an exact relation between the optimal QFI and $\mathcal{E}$, as obtained for $d = 2$.

To maximize the variance $\Delta^2 h_{2,d}$ and, therefore, the QFI over the choices of initial states, we first consider the decomposition of $\ket{\psi_{2,d}^{\text{in}}}$ in the eigenbasis of $h_{2,d}$. Thus, we have
\begin{equation}
\begin{split}
    \ket{\psi_{2,d}^{\text{in}}} &= \sum_{l,n=0}^{d-1} \sqrt{\omega_{l,n}^d} \ket{ln}_d \\
    &= \sum_{q=0}^{d^2 - 1} \sqrt{\omega_{q}^d} \ket{q},
\end{split}
\label{std}
\end{equation}
where $q$ is a shorthand index used to represent the dual indices $(l, n)$, and $0\leq\omega_q^d \leq 1 \, \forall q$. Note that we have omitted local phases in this expansion of $\ket{\psi_{2,d}^{\text{in}}}$, as the variance of $h_{2,d}$ is independent of phase and depends only on the probabilities $\omega_q^d$ of the system being in a particular eigenstate $\ket{q}$.

The eigenvectors $\ket{q}$ have eigenvalues $E_{l,n} = E_q = \bar{\eta}_l^d + \bar{\eta}_n^d$. Thus, the variance of $h_{2,d}$ for the state $\ket{\psi_{2,d}^{\text{in}}}$ is given by
\begin{equation}
     \Delta^2 h_{2,d}^{\text{in}} = \sum_{q=0}^{d^2 - 1} \omega_q^d (E_q)^2 - \left( \sum_{q} \omega_q^d E_q \right)^2.
\end{equation}

  We optimize $\Delta^2 h_{2,d} ^{\text{in}}$ over all choices of $\ket{\psi^{\text{in}}_{2,d}}$ with fixed entanglement $\mathcal{E}$. The optimized QFI at fixed entanglement is given by
\begin{equation}
    Q^{\mathcal{E}}_{2,d} = \max_{\psi^{\text{in}}_{2,d} \in \chi^{\mathcal{E}}_{2,d}} 4 \Delta^2 h_{2,d}^{\text{in}},
    \label{QFd}
\end{equation}
where $\chi^\mathcal{E}_{2,d}$ denotes the set of all initial states of the form given in Eq.~\eqref{std} with fixed entanglement $\mathcal{E}$.

We consider two measures of entanglement, $\mathcal{E}$: the generalized geometric measure $G$ and the von Neumann entanglement entropy $S$ of the reduced state, and numerically optimize the QFI as defined in Eq.~\eqref{QFd} for both measures.

We employ the Improved Stochastic Ranking Evolution Strategy (ISRES) algorithm from the NLopt non-linear optimization library for this constrained optimization. The optimization is performed over the $d$-dimensional probability space of $\omega^d_q$, subject to the constraint that the values of $\omega^d_q$ must be such that the entanglement of the initial probe state is fixed to a specified value of either $G$ or $S$ and $\sum_q\omega^d_q=1$.

Notably, for $G \in [0, {1}/{2}]$ and $S \in [0, 1]$, we find that the precision in higher dimensions $d = 3, 4, 5$ at constant $G$ and $S$ not only remains the same, but also exactly matches the precision obtained in the qubit case.

We formalize these findings as an observation corresponding to $d = 3, 4, 5$, followed by two lemmas.

\begin{observation}
Considering a bipartite probe with subsystem dimension $d = 3, 4, 5$, the optimal state maximizing $Q^G_{2,d}$ and $Q^S_{2,d}$ in the ranges $G \in [0, {1}/{2}]$ and $S \in [0, 1]$, respectively, satisfies $\omega^d_{ln} = 0$ for all $(l, n) \notin \{(0,0), (0,d^2 - 1), (d^2 - 1,0), (d^2 - 1,d^2 - 1)\}$.
\label{Ob}
\end{observation}

Based on this observation, we propose two lemmas corresponding to the two entanglement measures: the GGM and the von Neumann entanglement entropy. These lemmas are discussed in the following subsection.

\subsection{Entanglement-Constrained Optimal QFI for Bipartite Probes in Higher Dimensions}

In this subsection, we investigate how the entanglement-constrained optimal QFI varies when the entanglement is fixed to specific values. In particular, we establish the relationship between the entanglement-constrained optimal QFI and the fixed entanglement quantified in terms of both GGM and von Neumann entanglement entropy.

Note that for a given subsystem dimension $d$, the GGM, $G$, can range from $0$ to $({d - 1})/{d}$, and the entanglement entropy can range from $0$ to $\log_2{d}$. However, within the range $G \in [0, {1}/{2}]$ and $S \in [0, 1]$, we observe a generic feature in the optimal state with fixed $G$ and $S$ that maximizes the QFI, denoted by $Q^G_{2,2}$ and $Q^S_{2,2}$ respectively, as discussed in Observation.~\ref{Ob}. Therefore, we separately consider two cases: \textbf{Case 1:} $0 \leq G \leq {1}/{2}$ and $0 \leq S \leq 1$, and \textbf{Case 2:} ${1}/{2} < G \leq ({d - 1})/{d}$ and $1 < S \leq \log_2{d}$.

We first present the scenario $0 \leq G \leq {1}/{2}$ and $0 \leq S \leq 1$ in the following subsection.

\subsubsection*{Case 1:- 
 Entanglement Held Fixed Between SQL and HL
Barriers}

Observation~\ref{Ob} indicates that the problem of optimizing $Q^G_{2,2}$ and   $Q^S_{2,2}$ within the range $G \in [0, {1}/{2}]$ and $S \in [0, 1]$ respectively, is exactly equivalent to the qubit case discussed in Sec.~\ref{GMd=2} and Sec.~\ref{S=2}.

This can be explained as follows. According to Observation~\ref{Ob}, the initial state for a fixed $G \in [0, {1}/{2}]$ and $S \in [0, 1]$ takes the form
\begin{equation}
\ket{\psi^{\text{in}}_{2,d}} = \sum_{q=0}^{3} \sqrt{\omega_q^d} \ket{q}, \label{eqd1}
\end{equation}
where $q = 0, 1, 2, 3$ is a shorthand index denoting the pair of indices $(ln)$ with $l, n = 0 \text{ or } d^2 - 1$.

Consequently, the QFI corresponding to this state is expressed as
\begin{equation}
\begin{split}
Q^{G/S}_{2,d} &= \max_{\ket{\psi^{\text{in}}_{2,d}}\in {\chi^{G/S}_{2,d}}} 4  \Delta^2 h_{2,d}^\text{in},\\
&= 4 \left[ \sum_{q} \omega_q (E_q)^2 - \left( \sum_{q} \omega_q E_q \right)^2 \right]. \label{Qod}
\end{split}
\end{equation}

Where $G/S$ in the subscript $Q^{G/S}_{2,d}$, denotes that the entanglement of the initial state is quantified either in terms of GGM with the entanglement being fixed at a specific, $G\in[0,{1}/{2}]$ or it is quantified in terms of the entanglement entropy with fixed entanglement, $S$ that can take values within the range $[0,1]$. $\chi^{G/S}_{2,d}$ denotes the set of all such initial states with fixed $G$ or $S$.

We optimize $Q^{G/S}_{2,d}$ over all input states with fixed GGM, where $G \in [0, {1}/{2}]$ and $S\in[0,1]$. The GGM of the initial state $\ket{\psi^{\text{in}}_{2,d}}$, for a given dimension $d$, when expressed in terms of the parameters $\omega_q^d$ (with $q = 0, 1, 2, 3$), takes the form:
\begin{equation}
\mathcal{E}_{G}^d(\ket{\psi^{\text{in}}_{2,d}}) = \frac{1}{2} \left[ 1 - \sqrt{1 - 4 \left( \sqrt{\omega_1^d \omega_2^d} - \sqrt{\omega_0^d \omega_3^d} \right)^2 } \right].
\label{EDG}
\end{equation}


Similarly the entanglement entropy of the state $\ket{\psi^{\text{in}}_{2,d}}$, in terms of the parameter $\omega^d_q$, can be written as.
\begin{equation}
\mathcal{E}_S^d = -p^d_s \log_2 p^d_s - (1 - p^d_s) \log_2(1 - p^d_s),
\label{ESD}
\end{equation}
where
$$
p^d_s = \frac{1 - \sqrt{1 - 4(\sqrt{\omega^d_1 \omega^d_2} - \sqrt{\omega^d_0 \omega^d_3})^2}}{2}.
$$

For a given $d$, we perform the optimization over the parameter space of $\omega_q^d$ (where $q = 0, 1, 2, 3$), subject to the constraint either $\mathcal{E}_{G}^d(\ket{\psi^{\text{in}}_{2,d}}) = G$ or $\mathcal{E}_{S}^d(\ket{\psi^{\text{in}}_{2,d}})= S$. Where $G$ can take any value within the range $[0,{1}/{2}]$, and $S$ can take any value within the range $[0,1]$ and $\sum_q\omega_q^d=1$.

Since $q, l, m$ are dummy indices, we can relabel them as $p, r, s$ respectively. This replacement shows that the expressions for the initial input state, QFI, and entanglement in terms of the parameters $\omega_q^d$ (where $q = 0, 1, 2, 3$), given in Eqs.~\eqref{eqd1}, \eqref{Qod}, \eqref{EDG} and \eqref{ESD} are structurally identical to those in the qubit case discussed in Sec.~\ref{GMd=2} and Sec.~\ref{S=2} (with $\omega_p$ replaced by $\omega_q^d$).

Therefore, by following the same analysis as in Sec.~\ref{GMd=2} and  Sec.~\ref{S=2}  we can deduce the following two lemmas and two corresponding corollaries.

\begin{figure*}
\hspace{-1cm}
\includegraphics[scale=0.58]{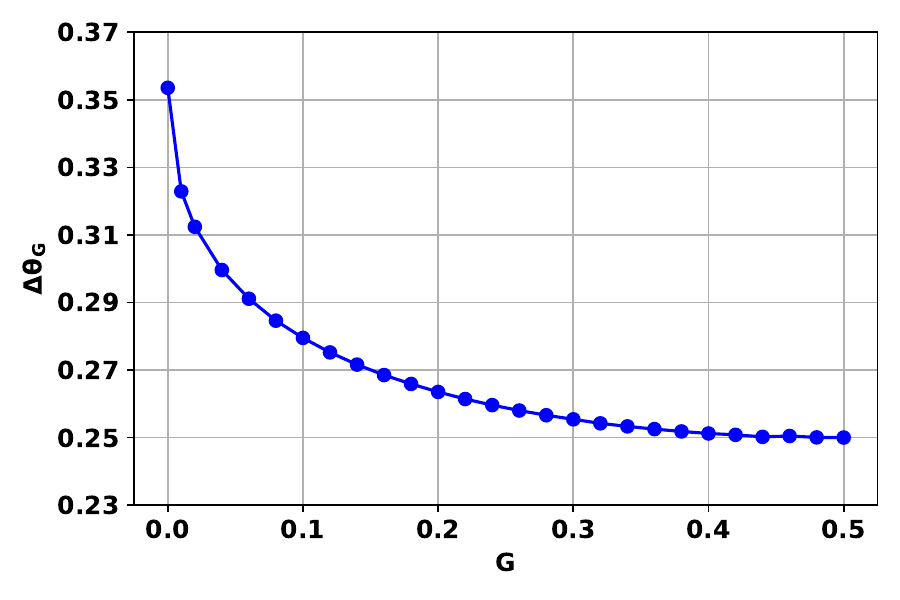}\hspace{0.8cm}
\includegraphics[scale=0.58]{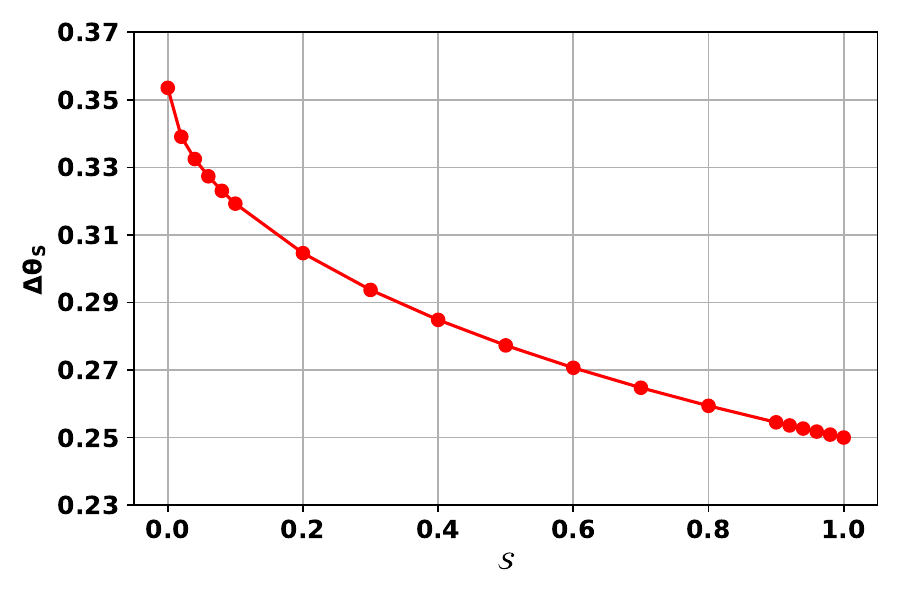}
\caption{\textbf{Minimum standard deviation vs fixed entanglement, for varying dimension}. The plot depicts entanglement-constrained optimal standard deviation in the estimation of the parameter $\theta$ versus entanglement plots are presented for a bipartite probe system with subsystem dimensions $d = 2, 3, 4, 5$. The right panel of the figure illustrates the variation of the minimum standard deviation, $\Delta{\theta_{G}}$, optimized over all choices of initial states with a fixed geometric measure of entanglement, $G \in [0, {1}/{2}]$. The left panel shows the variation of the minimum standard deviation in estimating $\theta$, denoted as $\Delta{\theta_{S}}$, maximized over all choices of initial states with a fixed von Neumann entanglement entropy, $S \in [0, 1]$. It is evident from both plots that the precision is highly sensitive to changes in entanglement in the low-entanglement regime, where even a slight increase in entanglement leads to a rapid improvement in precision. In contrast, in the high-entanglement regime, this sensitivity diminishes, and the growth of precision with increasing entanglement becomes more gradual. The vertical axis is dimensionless, and the horizontal axes are expressed in units of e-bits.}
\label{f1}
\end{figure*}

\begin{lemma}
\label{l3}
Consider a unitary encoding of the form given in Eq.\eqref{U_d}. Let the GGM, as defined in Eq.\eqref{EDG}, corresponding to the initial bipartite probe state in $\mathbb{C}^d \otimes \mathbb{C}^d$ with local dimension $d = 3, 4, 5$, be fixed to a specific value $G \in [0, {1}/{2}]$. Let $\chi^G_{2,d}$ denote the set of all such pure bipartite states with fixed GGM $G$. Then, the quantum Fisher information with respect to the parameter $\theta$ is defined as
\begin{equation*}
Q^G_{2,d} := \underset{\ket{\psi^{\text{in}}_{2,d}} \in \chi^G_{2,d}}{\max} 4\Delta^2 h_{2,d}^\text{in},
\end{equation*}
optimized over all initial probes $\ket{\psi^{\text{in}}_{2,d}}\in\chi^G_{2,d}$  is given by
\begin{equation}
Q^G_{2,d} = 8 \left( 1 + \sqrt{1 - (1 - 2G)^2} \right).
\end{equation}
\end{lemma}

The exact form of the optimal state with fixed entanglement $G \in [0, {1}/{2}]$ that maximizes the QFI can also be obtained. Since the proof of Lemma~\ref{l3} is exactly analogous to that of Theorem~\ref{GMd=2}, it follows that the optimal probe must satisfy $\omega_0^d = \omega_3^d$ and $\omega_1^d = \omega_2^d$. Additionally, the relation between $\omega_0^d$ and the fixed entanglement $G$ must match that given in Eq.~\eqref{Gw}, with $\omega_0$ replaced by $\omega_0^d$. We summarize this in the following corollary.

\begin{corollary}
For a given subsystem dimension $d$. The optimal state $\ket{\psi^{\text{o}}_{2,d}} \in \chi^G_{2,d}$ that yields the maximum QFI, $Q^G_{2,d}$, at fixed entanglement $G \in [0, {1}/{2}]$, is given by
\begin{equation}
\begin{split}
\ket{\psi^{\text{o}}_{2,d}(G)} &= \sqrt{\frac{1 + \sqrt{1 - (1 - 2G)^2}}{4}} \left( \ket{00}_d + \ket{l_d l_d}_d \right) \\
&\quad + \sqrt{\frac{1 -\sqrt{1 - (1 - 2G)^2}}{4}} \left( \ket{0 l_d}_d + \ket{l_d 0}_d \right),
\end{split}
\label{stdg}
\end{equation}
where $l_d = d^2 - 1$.
\end{corollary}
Next considering von Neumann entanglement entropy as an entanglement measure and following the exact analysis as in Sec.~\ref{S=2}. We present the following lemma.

    \begin{lemma}
For a unitary encoding of the form given in Eq.~\eqref{U_d}, if the entanglement entropy (as defined in Eq.~\eqref{ESD}) of the initial bipartite probe state in $\mathbb{C}^d \otimes \mathbb{C}^d$, with $d = 3, 4, 5$, is fixed to a value $S \in [0,1]$, and if $\chi^S_{2,d}$ denotes the set of all such pure bipartite states with entanglement entropy $S$, then the relationship between the quantum Fisher information,
\begin{equation*}
    Q^S_{2,d} := \max_{\ket{\psi^{\text{in}}_{2,d}} \in \chi^S_{2,d}} 4 \, \Delta^2 h_{2,d}^{\text{in}},
\end{equation*}
with respect to the parameter $\theta$, and the entropy $S$ is given by:
\begin{equation*}
    S = -\lambda_{Q^S_{2,d}} \log_2(\lambda_{Q^S_{2,d}}) - (1 - \lambda_{Q^S_{2,d}}) \log_2(1 - \lambda_{Q^S_{2,d}}),
\end{equation*}
where
\begin{equation*}
    \lambda_{Q^S_{2,d}} = \frac{1 - \sqrt{1 - \left(1 - \frac{Q^S_{2,d}}{8}\right)^2}}{2}.
\end{equation*}
\end{lemma}

Also since the proof of the above lemma is analogous to that presented in Sec.~\ref{S=2}. This suggests that  optimal probe at fixed $S \in[0,1]$ that gives the highest $Q^S_{2,2}$ must have a form similar to that given in Corollary 2(with $\omega_p$ being replaced by $\omega_q^d)$. This is summarized in the form of a corollary below.

\begin{corollary}
For a given subsystem dimension $d=3,4,5$. The optimal state $\ket{\psi^{\text{o}}_{2,d}} \in \chi^S_{2,d}$ that yields the maximum QFI, $Q^S_{2,d}$, at fixed entanglement $S \in [0, 1]$, given as $S= -p_s^d \log_2 p^d - (1 - p_s^d) \log_2(1 - p_s^d), p_s^d \in [0,{1}/{2}]$, must have the following form
\begin{eqnarray}
\begin{split}
\ket{\psi_{2,d}^o(S)} &= \sqrt{\frac{1+\sqrt{1-(1-2p^d_s)^2}}{4}}(\ket{00}_d+\ket{l_dl_d}_{d}) \\
&\quad + \sqrt{\frac{1-\sqrt{1-(1-2p^d_s)^2}}{4}}(\ket{0l_d}_{d}+\ket{l_d0}_{d}).
\label{Sst}
\end{split}
\end{eqnarray}
where $l_d = d^2 - 1$.
\end{corollary}

\noindent
Since $Q^{G/S}_{2,d}$ has exactly the same expression in terms of $G$ and $S$ for $d = 2, 3, 4, 5$, we omit the superscript $d$ and denote it simply as $Q^{G/S}_2$ to emphasize this commonality. Using the quantum Cramér-Rao bound, the corresponding optimal standard deviation in estimation of the parameter $\theta$ is given by
$$\Delta \theta_G = (Q^G_{2})^{-\frac{1}{2}},$$ 
$$\Delta \theta_S = (Q^S_{2})^{-\frac{1}{2}}.$$

\noindent 
In the left and right panels of Fig.~\ref{f1}, we plot the precision, $\Delta \theta_G$ and $\Delta \theta_S$, as functions of the geometric measure of entanglement, $G \in [0, {1}/{2}]$, and the entanglement entropy, $S \in [0, 1]$, respectively. The plots clearly reveal that in the low-entanglement regime, the precision is highly sensitive to changes in entanglement—small increases in $G$ or $S$ lead to significant improvements in precision. However, in the higher entanglement regime, this sensitivity diminishes, and the rate of improvement in precision gradually slows down as entanglement increases. It is well-established that the highest achievable precision, in the absence of any entanglement constraints, is attained at $G = {1}/{2}$ or $S = 1$, corresponding to the cat state. Remarkably, our analysis indicates that even without direct access to the cat state, preparing states of the form given in Eq.~\eqref{stdg} and Eq.~\eqref{Sst}, with $G$ approaching ${1}/{2}$ or $S$ approaching $1$, is sufficient to achieve near-optimal precision. This finding offers a significant practical insight: with entanglement fixed within the ranges corresponding to the SQL and HL namely, $G \in [0, {1}/{2}]$ and $S \in [0, 1]$ the precision initially increases rapidly with entanglement but the growth slows down or nearly saturates in the high-entanglement regime. This suggests that maximal entanglement is not strictly necessary to approach the ultimate precision limit, known as HL.
In the following subsection we consider the second case with higher entanglement values.
\subsubsection*{Case 2:- Entanglement Held Fixed Beyond HL Barrier}

Unlike the previous case, we did not observe any consistent or generic trends in the structure of the optimal states or in the maximum achievable precision for entanglement values beyond $G > {1}/{2}$ and $S > 1$, i.e. in the ranges ${1}/{2} < G \leq ({d-1})/{d}$ and $1 < S \leq \log_2{d}$. In this regime, we rely entirely on numerical optimization, as no exact analytical relation could be established between the entanglement-constrained optimal quantum Fisher information (QFI) and the corresponding fixed entanglement values. In Fig.~\ref{f2}, we present the numerically obtained optimal precision vs  fixed entanglement plots for subsystem dimensions $d = 3, 4, 5$. It is important to emphasize that due to strict time and computational resource limitations, the   ISRES algorithm from the NLopt optimization library did not consistently converge to the global optimum for higher values of $S$ and $G$ across all $d$. Therefore, the plotted data includes only those points where either the convergence criteria were satisfactorily met within the available resources or the entanglement constraint was exactly satisfied. For some points, although convergence to a global optimum could not be guaranteed within the allotted computational time, the fixed entanglement constraint was strictly satisfied. These points represent the best achievable precision within our resource limitations, but they are not necessarily guaranteed to correspond to the true global optima. Nonetheless, they provide valuable insight into the performance achievable under time limitations.

\begin{figure*}
\hspace{-1cm}
\includegraphics[scale=0.58]{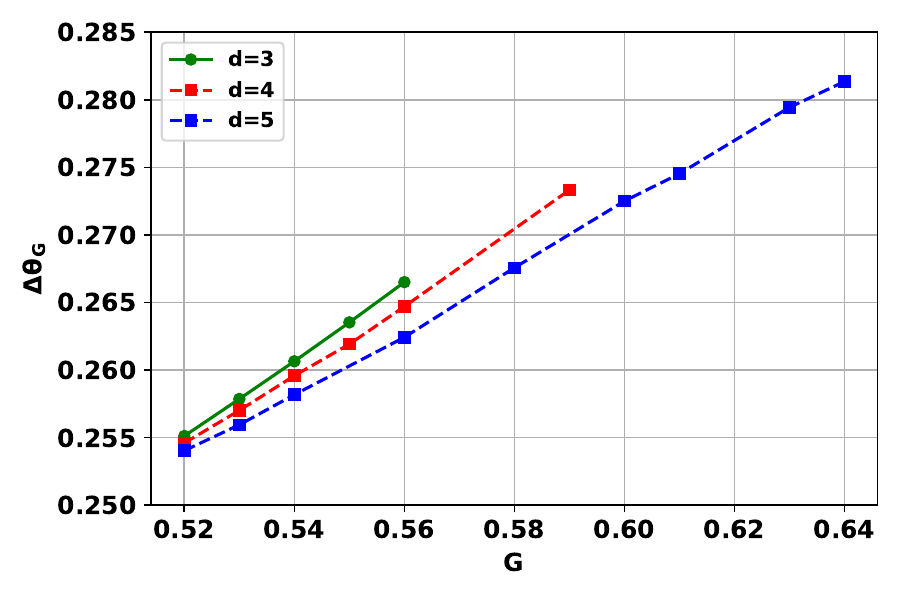}\hspace{0.8cm}
\includegraphics[scale=0.58]{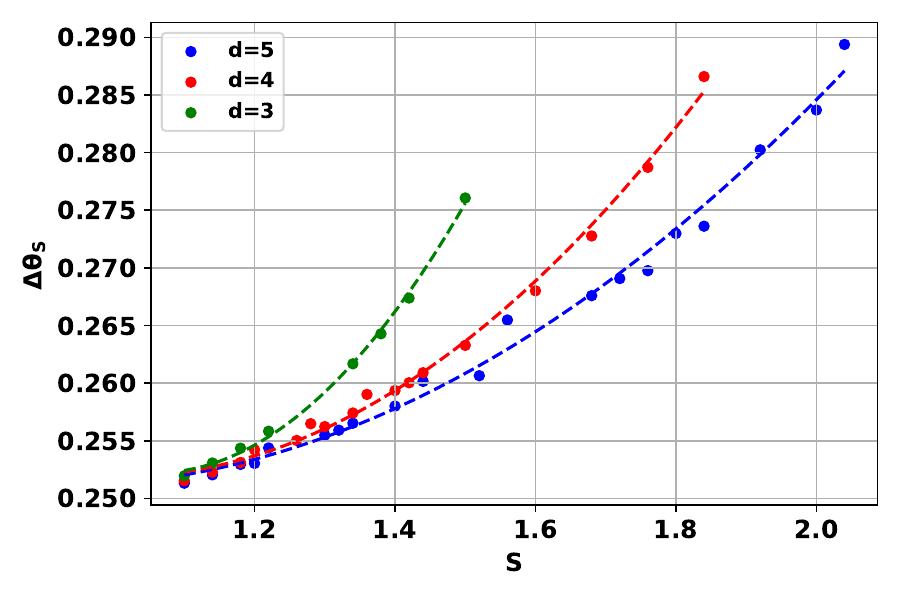}
\caption{\textbf{Minimum standard deviation vs entanglement fixed beyond HL barrier}. The above plot depicts the variation of the entanglement-constrained optimal standard deviation in the estimation of $\theta$, for fixed initial entanglement values beyond those corresponding to the HL, for subsystem dimensions $d = 3, 4, 5$. The plots corresponding to different dimensions are distinguished by color: green for $d = 3$, red for $d = 4$, and blue for $d = 5$.  The left panel shows the variation of the standard deviation $\Delta \theta_G$ with entanglement quantified by GM, denoted as $G$, where $G > \frac{1}{2}$. This plot demonstrates that, within the relevant range of convergence, $\Delta \theta_G$ increases linearly with $G$ for $d=3,4,5$, indicating that the optimal precision decreases linearly as $G$ increases. The right panel illustrates the variation of the standard deviation $\Delta \theta_S$ with entanglement quantified by the von Neumann entanglement entropy, $S>1$. The plot shows that $\Delta \theta_S$ increases quadratically with $S$ for all dimensions $d=3,4,5$. Here, the scattered points represent the numerically obtained optimized values of $\Delta \theta_S$, while the dotted lines correspond to the best-fit curves for various dimensions. The equation of the quadratic fit is given by $(a_d S^2 + b_d S + c_d)^{-\frac{1}{2}}$, where $a_d$, $b_d$, and $c_d$ are fitting parameters that depend on the subsystem dimension $d$ and the values of which are given in Table.\ref{tab1}. The vertical axis of the plots is dimensionless, and the horizontal axes are expressed in units of e-bits. }
\label{f2}
\end{figure*}

We maintain the same notations, $\Delta\theta_G$ and $\Delta\theta_S$, to denote the precision at fixed values of $G > {1}/{2}$ and $S > 1$, respectively, for all considered dimensions $d = 3, 4, 5$. However, to clearly distinguish the precision corresponding to different $d$ values, we have plotted the precision curves in Fig.~\ref{f2} using distinct colors: green for $d = 3$, red for $d = 4$, and blue for $d = 5$. The right panel of Fig.~\ref{f2} illustrates the variation of $\Delta \theta_G$ as a function of $G$. As evident from the plot, $\Delta \theta_G$ increases approximately linearly, and therefore the precision decreases linearly with $G$ within the relevant range of convergence. The left panel of Fig.~\ref{f2} shows the precision $\Delta \theta_S$ as a function of $S > 1$ for different values of $d$. The plots reveal that, for all three dimensions, $\Delta \theta_S$ increases, with the entanglement constrained optimal QFI exhibiting a quadratic dependence on $S$, and hence the precision consistently decreases as $S$ increases. Here, the scattered points represent the numerically obtained optimized values of $\Delta \theta_S$, while the dotted lines correspond to the best-fit curves for various dimensions. The equation of these curve is given by $(a_d S^2 + b_d S + c_d$, where $a_d$, $b_d)^{-\frac{1}{2}}$, and $c_d$ are fitting parameters that depend on the subsystem dimension $d$. The values of these fitting parameters are provided in the Table.~\ref{tab1}. This completes our analysis for the case of ${1}/{2} < G \leq ({d-1})/{d}$ and $1 < S \leq \log_2{d}$.

Note that, in the entire analysis above, we considered the local Hamiltonian with equally spaced eigenenergies. This  results in the complete overlap of the precision curves in the range $0 \leq G \leq {1}/{2}$ and in the range $0 \leq S \leq 1$ for $d = 2, 3, 4, 5$. 

However, what if the eigenenergies are not equally spaced? Will the precision still remain identical to that of the $d=2$ case? We address this question with an example below, considering the $d=3$ case. We show that when the eigenenergies of the local Hamiltonian are not equally spaced, the relation between the entanglement-constrained optimal QFI and the fixed initial entanglement (quantified both in terms of GM and von Neumann entanglement entropy) differs from the relations deduced in Theorem.~\ref{t1} and Theorem.~\ref{t2}. This situation is distinct from that with equally spaced eigenenergies. 

Nevertheless, it is worth noting that the generic nature of the entanglement-constrained optimal precision, for a fixed entanglement starting from $G, S = 0$ (the entanglement value corresponding to the SQL) up to $G = {1}/{2}, S = 1$ (the values corresponding to the HL), still persists. This behavior is characterized by rapid growth in the low entanglement region and slow growth in the high entanglement region. The example is presented below.

\begin{table}[h]
    \centering
   \resizebox{5cm}{1.5cm}{ \begin{tabular}{|c|c|c|c|c|}
        \hline
        & $d=3$ & $d=4$ & $d=5$ \\
        \hline
        $a_d$ & 0.028  & 0.047 & 0.12  \\
        \hline
        $b_d$ & -0.051 & -0.254 & --0.386\\
        \hline
        $c_d$ & 0.274 & 0.3 & 0.386  \\
        \hline
    \end{tabular}}
    \caption{Values of the fitting parameters $a_d$, $b_d$, and $c_d$ corresponding to the quadratic curve fitted to the $\Delta \theta_S$ vs. $S$ plot in Fig.~\ref{f2}, following the equation $(a_d S^2 + b_d S + c_d)^{-\frac{1}{2}}$. The parameters are provided for subsystem dimensions $d = 3, 4, 5$.}
  \label{tab1}
\end{table}
\subsection{ When the Eigenenergies of the Local Hamiltonian Are Not Equally Spaced}

In this subsection, we explore how the entanglement-constrained optimal QFI varies with entanglement in the range from the SQL to the HL, when the local components of the Hamiltonian are not equally spaced. To investigate this, we consider a specific example of a bipartite system probe with a subsystem dimension of $d = 3$.\\

Consider the encoding Hamiltonian of the form:
\begin{eqnarray*}
\begin{split}
    \bar{H}_{2,3} &= (\bar{h}_{2,3}) \, \theta' \\
    &= \left( \mathcal{{N}}_3 \otimes \mathcal{I}_3 + \mathcal{I}_3 \otimes \mathcal{{N}}_3 \right) \theta'.
\end{split}
\end{eqnarray*}

Where according to the notation introduced in Sec.~\ref{notation}, we set $(N, d) = (2, 3)$ in the superscript. $\mathcal{N}_3$ denotes the local component of the Hamiltonian with eigenenergies $\zeta_0 = 0$, $\zeta_1 = 2$, and $\zeta_2 = 3$, with corresponding eigenstates $\ket{k}$, where $k = 0, 1, 2$. Thus, the eigenenergies of the global Hamiltonian $h_{2,3}$ are given by $\bar{E}_{j,k} = \zeta_j + \zeta_k$, with corresponding eigenvectors $\ket{jk}_3$, where $j, k = 0, 1, 2$. We use a shorthand index $q$ to denote the collective index pair $(j,k)$. Accordingly, the set of eigenenergies and eigenvectors of the generator $h_{2,3}$ are denoted as $\{\bar{E}_q\}$ and $\{\ket{q}\}$, respectively, where $m = 0, 1, \ldots, d^2 - 1$.

Thus the optimal entanglement constrained QFI, maximized over the choice of fixed initial states,$\ket{\bar{\psi}^{\text{in}}_{2,d}}$ is given as,

\begin{equation}
\begin{split}
\bar{Q}^{G/S}_{2,3} &= \max_{\ket{\bar{\psi}^{\text{in}}_{2,3}}\in {\bar{\chi}^{G/S}_{2,3}}} 4\Delta^2 h_{2,d}^\text{in},\\
&= 4 \left[ \sum_{q} \omega_q E_q^2 - \left( \sum_{q} \omega_q E_q \right)^2 \right]. \label{Qound}
\end{split}
\end{equation}

Where $\bar{\chi}^{G/S}_{2,3}$ represents the set of all input states with fixed entanglement in the range $G\in[0,{1}/{2}]$ and $S\in[0,1]$, respectively. We use ISRES algorithm from the NLopt  library to optimize $\bar{Q}_{G/S}^3$. For which we consider the initial state with fixed entanglement of the form 
\begin{equation}
\begin{split}
    \ket{\bar{\psi}_{2,3}^{\text{in}}} 
    &= \sum_{q=0}^{d^2 - 1} \sqrt{\bar{\omega}_{q}^3} \ket{q},
\end{split}
\label{std}
\end{equation}
We observe that for all $G \in [0,{1}/{2}]$ and $S \in[0,1]$. The optimal state of fixed entanglement have a particular unique form as given in the form of an observation below
\begin{observation}
    Considering a bipartite probe with subsystem dimension $d = 3$, the optimal state maximizing $\bar{Q}^{G/S}_{2,3}$ in the ranges $G \in (0, {1}/{2})$ and $S \in (0, 1)$, respectively, satisfies $\omega^d_{jk} = 0$ for all $(j, k) \notin \{(0,0), (0,d^2 - 1), (d^2 - 1,0), (d^2 - 1,d^2 - 1)\}$.
\end{observation}

It is known from Ref.~\cite{M1} that for $G, S = 0$, the optimal state has the form $\ket{o}^{\otimes 2}$, where $\ket{o} = (\ket{0} + \ket{3})/\sqrt{2}$. Furthermore, the optimal state at $G = {1}/{2}$, $S = 1$ takes the form $(\ket{00} + \ket{33})/\sqrt{2}$. This suggests that the optimal state for $G, S = 0$ and for $G = {1}/{2}, S = 1$ satisfies the condition $\bar{\omega}_0 = \bar{\omega}_3$ and $\bar{\omega}_1 = \bar{\omega}_2$. For all intermediate values $G \in (0, {1}/{2})$ and $S \in (0, 1)$, we apply the Lagrange multiplier method within this range, following the exact steps outlined in Lemma~1 and Lemma~2, with $E_q$ replaced by $\bar{E_q}$. We find that, even in this case, the optimal state must satisfy $\bar{\omega}_0 = \bar{\omega}_3$ and $\bar{\omega}_1 = \bar{\omega}_2$. Combining this with the previously known results for $G, S = 0$ and $G ={1}/{2}, S = 1$, we obtain the optimal QFI at fixed entanglement for this case as
$$
\bar{Q}^{G/S}_{2,3} = 18 \omega_0.
$$
Thus, the relation between the optimal QFI and fixed geometric measure $G \in [0, \frac{1}{2}]$ is given by

\begin{equation}
\bar{Q}^G_{2,3} = 4.5 \left( 1 + \sqrt{1 - (1 - 2G)^2} \right).
\end{equation}

And the relation between optimal QFI at fixed entanglement is given as

\begin{equation*}
    S = -\lambda_{\bar{Q}_S^3} \log_2(\lambda_{\bar{Q}_S^3}) - (1 - \lambda_{\bar{Q}_S^3}) \log_2(1 - \lambda_{\bar{Q}_S^3}).
\end{equation*}
Where,
\begin{equation*}
    \lambda_{\bar{Q}_S^3}=\frac{1-\sqrt{1-(1-\frac{\bar{Q}_S^3}{4.5})^2}}{2}.
\end{equation*}
$\bar{Q}^{G/S}_{2,3}$ have exactly same functional form as ${Q}^{G/S}_{2,3}$. In fact $\bar{Q}^{G/S}_{2,3}=0.56{Q}^{G/S}_{2,3}$. This shows that although the entanglement constrained optimal QFI for $d=3$, considering local Hamiltonian with unequal energy spacing is not exactly same as of  $d=2$ case. However the characteristic  feature  of rapid growth of precision in the low entanglement region and slow growth in the region of high entanglement can still persist.

This complete our analysis considering the bipartite system probe, with subsystem dimension 2 or higher. We present the exact relation between the input fixed entanglement and the best optimal precision within the SQL and HL. 

In the following section, we extend our analysis to multipartite system probes and investigate how the entanglement-constrained optimal QFI varies under the practical limitation of fixed input entanglement.

\section{QFI vs entanglement for multipartite probes}
\label{multi}

Consider $N$ partite probe each with subsystem dimension $d=2$. The encoding Hamiltonian for such a  system can be written as
\begin{equation*}
\begin{split}
  H_{N,2} &= (h_{N,2})\theta' \\
    &= \sum_{k=1}^NZ_2^k.
    \end{split}
    \end{equation*}
Here $Z_2^k=\sigma_z^k$ and we specifically consider cases $N=3,4,5$. As in previous sections our aim is to find the optimal QFI when the input states are restricted to have a fixed amount of entanglement. 

If the eigenvalues of the generator $h_{N,2}$ are given as $E_p$ with corresponding eigenvectors $\ket{p}$, then the set of initial states that have a fixed variance in $h_{N,2}$ consists of states of the form
\begin{equation}
\begin{split}
    \ket{\psi_{N,2}^{\text{in}}} &= \sum_{p}^{2^N-1} \sqrt{\omega_{p}^2} \ket{p}. \\
    \label{GMBFst}
\end{split}
\end{equation}
Here, the superscript $2$ in $0 \leq \omega_{p}^2 \leq 1$ indicates that $d = 2$. Note that $\sum_p \omega_{p}^2 = 1$. The QFI corresponding to these states is given as
$$
Q(\ket{\psi_{N,2}^{\text{in}}}) = 4 \Delta^2 h_{N,2} ^{\text{in}},
$$

 If one consider initial states with fixed entanglement $\mathcal{E}$. The maximum QFI optimized over all such choices of input states is given as

 \begin{equation*}
     Q^\mathcal{E}_{N,2} = \max_{\ket{\psi^{\text{in}}_{N,2}}\chi^\mathcal{E}_{N,2}}4 \Delta^2 h_{N,2}^{\text{in}}.
 \end{equation*}

Where $\chi^\mathcal{E}_{N,2}$ is the set of all states having fixed entanglement $\mathcal{E}$.
We consider two types of entanglement measures for the multipartite scenario, first being GGM and second measure is GM. The analysis considering GGM as the entanglement measure is presented in the section below.

\subsection{When Entanglement is Quantified in Terms of GGM}
\begin{figure}
\includegraphics[scale=0.45]{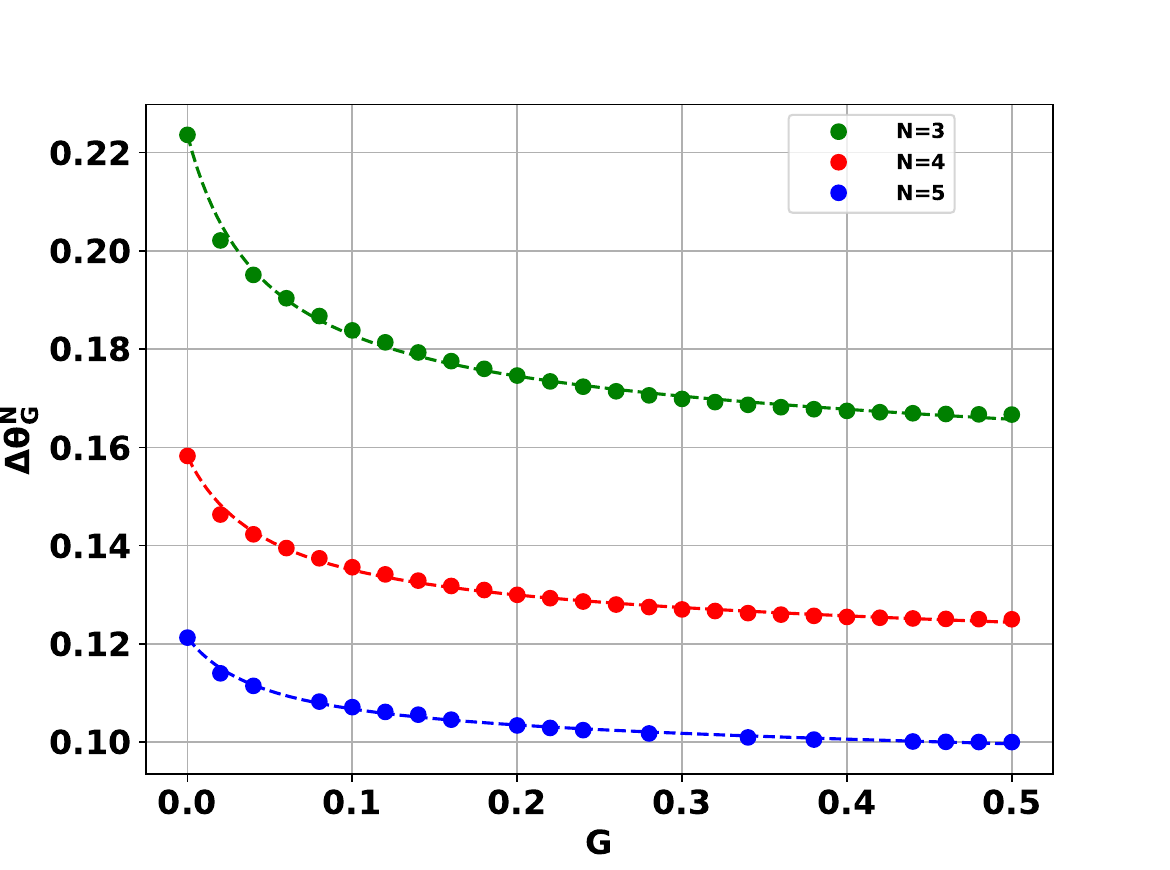}
\caption{\textbf{Minimum standard deviation vs fixed GGM for varying number of parties in the probe }. The plot shows the variation of the entanglement-constrained minimum standard deviation, $\Delta\theta_G^N$, in the estimation of the parameter $\theta$, as a function of fixed initial entanglement $G \in [0, {1}/{2}]$, using multipartite probes with local subsystem dimension $d = 2$. The scattered points in various colors, green for $N = 3$, red for $N = 4$, and blue for $N = 5$, represent the numerically obtained values of $\Delta\theta_G^N$ for fixed entanglement $G$. The corresponding dashed lines denote the best-fit curves through these points, with functional form given by $f_N(G) = \left( \frac{a_N G^2 + b_N G + c_N}{G + d_N} \right)^{-\frac{1}{2}}$, where $a_N$, $b_N$, $c_N$, and $d_N$ are fitting parameters that depend on $N$. The values of the fitting parameter corresponding to  $N=3,4,5$, is given in Table.~\ref{tab2} For all three values of $N$, it can be observed from the plot that the standard deviation decreases sharply in the low entanglement regime and approaches saturation in the high entanglement regime. The vertical axis is dimensionless, and the horizontal axis is expressed in units of e-bits.}
\label{party}
\end{figure}
GGM, $G$ for a multipartite probe state composed of qubit can vary in the range $G\in[0,{1}/{2}]$. The  entanglement constrained QFI for a fixed GGM, $G$, is given as

\begin{equation*}
    Q^G_{N,2} = \max_{\ket{\psi^{\text{in}}_{N,2}}\chi^G_{N,2}}4 \Delta^2 h_{N,2}^{\text{in}}.
\end{equation*}

Considering all possible bipartitions of $\ket{\psi^{\text{in}}_{N,2}}$, we determine the maximum Schmidt coefficient across all such bipartitions, which we use to compute the GGM of the initial state. If this maximum Schmidt coefficient is denoted by $\lambda_{M}$, then the constraint for the above optimization is given by:
\begin{equation*}
    1-\lambda_M^2=G.
\end{equation*}

 We again use ISRES algorithm to perform the above constrained optimization, over the space of $\{\omega_p\}$ considering $N=3,4,5$, with fixed $G\in[0,{1}/{2}]$. In Fig.~\ref{party}, we plot the the standard deviation, $\Delta\theta^N_G$ in the estimation of $\theta$, given as
$$\Delta\theta^N_G=(Q^G_{N,2})^{-\frac{1}{2}}$$

In Fig.~\ref{party}, the scattered points in various colors, green for $N=3$, red for $N=4$, and blue for $N=5$, represent the actual values obtained from numerical optimization, while the corresponding dashed lines represent the curves that best fit the data points. The equation of these curves is given by the inverse square root of a rational function, and is expressed as
\begin{equation*}
f_N(G) = \Bigg(\frac{a_{N} G^2 + b_{N} G + c_{N}}{G + d_{N}}\Bigg)^{-\frac{1}{2}},
\end{equation*}
where $a_N$, $b_N$, $c_N$, and $d_N$ are the fitting parameters that depend on $N$. The values of these parameters are presented in the table below.

\begin{table}[h]
    \centering
   \resizebox{5cm}{1.5cm}{ \begin{tabular}{|c|c|c|c|c|}
        \hline
        & $N=3$ & $N=4$ & $N=5$ \\
        \hline
        $a_N$ & 4.51  & 7.04 & 11.54 \\
        \hline
        $b_N$ & 36.48 & 64.58 &  99.06\\
        \hline
        $c_N$ & 1.4 & 2.8 & 4.08  \\
        \hline
        $d_N$ & 0.07 & 0.07 & 0.06  \\
        \hline
    \end{tabular}}
     \caption{Values of the fitting parameters $a_N$, $b_N$, $c_N$, and $d_N$ corresponding to the curve $
f_N(G) = \left( \frac{a_N G^2 + b_N G + c_N}{G + d_N} \right)^{-\frac{1}{2}}$ fitted to the $\Delta \theta_G$ vs. $G$ plot shown in Fig.~\ref{party}. The parameters are provided for the number of parties $N = 3, 4, 5$.
}
  \label{tab2}
\end{table}

The main observation from the plot is that the characteristic behavior of the entanglement-constrained optimal standard deviation in $\theta$, and consequently the precision in estimation, as a function of the fixed entanglement value $G$, remains consistent with the bipartite case ($N = 2$). Specifically, the precision exhibits a rapid increase in the low-entanglement region, followed by a slower, near-saturation growth in the high-entanglement regime. This indicates that, even in multipartite scenario, achieving a precision very close to the ultimate limit set by the HL does not require highly entangled states.

It is worthy to note that the value of $\Delta\theta^N_G$ at $G=0$ an $G={1}/{2}$, exactly matches to that of the bound derived in previous studies~\cite{M1,Toth}. 

In the next subsection we consider yet another measure of entanglement call the geometric measure of entanglement and analyze how the entanglement constrained optimal precision varies with fixed value of GM. 
\subsection{When Entanglement is Quantified in Terms of GM}

Recall from the definition of GM presented in Sec.~\ref{EM}: calculating the GM of any pure state $\ket{\psi}$ requires performing an optimization over the set of fully separable product states, $\ket{\phi}$. Therefore, in order to optimize the QFI corresponding to an input state $\ket{\psi_{N,2}^{\text{in}}}$ with a fixed value of GM, two optimizations are necessary: one over the set of pure product states for a fixed input state, and another over the set of all such input states to maximize the QFI. This can be expressed mathematically as:
\begin{equation*}
\mathcal{GM} = 1 - \max_{\ket{\phi} \in \chi^N_{\text{sep}}} \abs{\braket{\phi | \psi_{N,2}^{\text{in}}}}^2,
\end{equation*}
where $\chi^N_{\text{sep}}$ is the set of all fully separable pure product states in the Hilbert space $(\mathbb{C}^2)^{\otimes N}$. Thus, the optimal QFI maximized over all choices of input states with fixed GM, denoted $\mathcal{GM}$, is given by
\begin{equation*}
Q^\mathcal{GM}_{N,2} = \max_{\ket{\psi^{\text{in}}_{N,2}} \in \chi^\mathcal{GM}_{N,2}} 4 \Delta^2 h_{N,2}^{\text{in}},
\end{equation*}
where $\chi^\mathcal{GM}_{N,2}$ is the set of all $N$-partite input probe states with a fixed value of geometric entanglement $\mathcal{GM}$.

Carrying out these two nested optimizations first to compute GM, then to find the optimal QFI at fixed GM is numerically demanding and can significantly reduce the efficiency of the optimization algorithm. To circumvent this issue, we adopt a numerical approach known as the brute-force or heuristic method.

To reduce computational complexity further, in this section we restrict ourselves to the simplest case of a multipartite probe, consisting of $N=3$ qubits. The method proceeds in the following steps:

\begin{itemize}
\item First, we generate $\nu$ pure states of the form given in Eq.~\eqref{GMBFst}, where the coefficients $\omega_p^2$, for $p = 0, 1, \ldots, 7$, are drawn from a uniform distribution over the interval $[0,1]$.

\item Using the ISRES algorithm from the NLopt library, we compute the GM corresponding to these $\nu$ randomly generated states.

\item We partition the full GM range, $\mathcal{GM} \in [0,{1}/{2}]$, into bins of width 0.05. We assign each state to a bin indexed by $k$, such that the GM of the state lies within $[0.05k, 0.05k + 0.05]$, for $k = 0, 1, \ldots, 9$. Note that a given state may fall into at most two adjacent bins. We denote the set of all states in bin $k$ as $\chi_k$.

\item For each state $\ket{\psi_{3,2}^\text{in}}$ in $\chi_k$, we compute the QFI, $Q(\ket{\psi_{3,2}^\text{in}})$, and define the maximal QFI over the bin as
$$
Q_{\text{Max}} = Q^{\mathcal{GM}}_{3,2}=\max_{\ket{\psi_{3,2}^\text{in}} \in \chi_k} Q(\ket{\psi_{3,2}^\text{in}}).
$$

\item This gives the maximal QFI achievable among randomly generated states with GM in the interval $[0.05k, 0.05k + 0.05]$.

\item We repeat the entire process for another set of $\nu$ random states and check the convergence of $Q_{\text{Max}}$ within each bin to ensure stability and robustness of the results.
\end{itemize}

In Fig.~\ref{GMparty}, we present the variation of the optimal precision in parameter estimation, quantified through the maximum standard deviation $\Delta\theta_{\mathcal{GM}} = Q_{\text{Max}}^{-1/2}$, as a function of $k$, for three different values of $\nu$. The curves corresponding to $\nu = 10^5$, $10^6$, and $10^7$ are distinguished using red, green, and blue colors, respectively.
As $\nu$ increases, the values of $\Delta\theta_{\mathcal{GM}}$ for each $k$ converge, indicating statistical reliability of the maximization procedure, $k \in [0, 6]$. Beyond $k = 6$, numerical convergence of the QFI could not be reliably achieved.

The figure reveals a consistent pattern as the degree of entanglement increases, the maximum achievable precision improves, but the rate of this improvement is not uniform. In the low-entanglement regime(low $k$ values, the standard deviation $\Delta\theta_{\mathcal{GM}}$ decreases sharply with increasing entanglement, indicating a rapid gain in precision. However, this improvement becomes more gradual in the higher-entanglement region (high $k$ values).

This behavior closely mirrors the trends observed in earlier sections for entanglement quantified by other measures, such as the generalized geometric measure (GGM). Together, these findings reinforce a key insight: while entanglement is indeed a resource for quantum metrology, the enhancement in precision it offers is most significant when transitioning out of the low-entanglement regime. In contrast, additional gains from highly entangled states are relatively modest. 

\begin{figure}
\hspace{-1cm}
\includegraphics[scale=0.6]{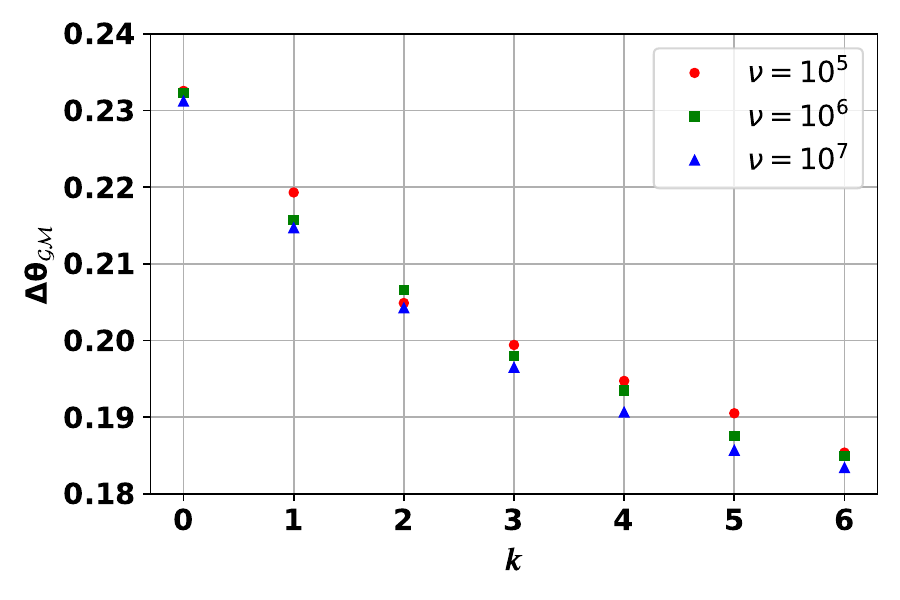}
\caption{\textbf{Minimum standard deviation vs fixed GM for varying number of parties in the probe}. The Plot illustrates the variation of the minimum standard deviation, $\Delta\theta_{\mathcal{GM}} = Q_{\text{Max}}^{-1/2}$, in the estimation of the parameter $\theta$, where the minimization is performed over randomly generated states with GM  within the interval $\mathcal{GM} \in [0.05k,, 0.05k + 0.05]$ for $k \in [0,6]$. The results are shown for three different realizations of the total number of randomly generates state samples $\nu$, with $\nu = 10^5$ (red), $\nu = 10^6$ (green), and $\nu = 10^7$ (blue). As $\nu$ increases, the values of $\Delta\theta_{\mathcal{GM}}$ for each $k$ converge, indicating statistical reliability of the minimization procedure. The plot further reveals a steep decline in $\Delta\theta_{\mathcal{GM}}$ for small $k$ (low entanglement), followed by a more gradual decrease for larger $k$ (higher entanglement), suggesting small precision gain as entanglement increases. The vertical axis is dimensionless, and the horizontal axis is expressed in units of e-bits. } 
\label{GMparty}
\end{figure}

\section{Conclusion}
\label{Con}
Entanglement is a key resource for enhancing the precision of parameter estimation under unitary encoding, enabling performance beyond the standard quantum limit, which represents the best precision achievable with product states. The ultimate optimal precision is achieved by a special class of entangled states, and is known as the Heisenberg Limit.

However, in realistic scenarios, the amount of entanglement that can be utilized is often limited. This constraint reflects practical conditions where, due to experimental limitations, one may not have access to ideal entangled states that achieve the HL. Under such circumstances, it becomes essential to determine the best possible precision that can be achieved given a fixed amount of entanglement. Until now, a direct connection between the available entanglement and the best achievable precision had not been established. In this study, we addressed this gap by directly linking achievable precision to the amount of entanglement at hand. By systematically optimizing the quantum Fisher information under entanglement constraints, we established the ultimate precision bound maximized over input probe states with a fixed amount of entanglement.

Focusing first on bipartite probe states composed of two qubits, we derived an exact relation between the entanglement-constrained optimal QFI and the corresponding fixed initial entanglement. Entanglement was quantified using both the  geometric measure and the entanglement entropy. We demonstrated that this relation also hold for bipartite states composed of two qudits, across the same range of entanglement measures. Furthermore, we identified the specific probe states that guaranteed the highest precision in both the two-qubit and two-qudit cases. We subsequently extended our analysis to multipartite probe states composed of qubits. In this setting, we used  generalized geometric measure and the geometric measure of entanglement to quantify the entanglement. Using numerical optimization methods in this case, we examined how the optimal precision varied with the level of entanglement. 


In both the bipartite and multipartite cases, we observed a consistent pattern in the relationship between optimal precision and fixed entanglement values, within the range defined by the SQL and the HL. This relationship revealed that precision increased rapidly in the low-entanglement regime but grew more slowly and nearly saturated as the entanglement approached its ideal HL value. This indicated that even without direct access to the ideal states that achieve the HL, preparing states with entanglement values close to the ideal was sufficient to attain near-optimal precision—namely, the Heisenberg limit.


\appendix
\section{Optimal QFI for Fixed GGM: The Case of Two Qubits}
\label{A1}
Here, we present the set of six equations obtained by taking the partial derivatives of the Lagrangian, given in Eq.~\eqref{Lag1}, with respect to the parameters $\omega_p$, with $p=0,1,2,3$ and the Lagrange multipliers $\mathcal{K}_1$ and $\mathcal{K}_2$.
\begin{align}
    \frac{\partial\mathcal{L}}{\partial\omega_0} &= (E_0)^2 - 2D E_0 + \mathcal{K}_1 - \mathcal{K}_2 \Lambda \sqrt{\frac{\omega_3}{\omega_0}}, \label{eq6a} \\
    \frac{\partial\mathcal{L}}{\partial\omega_3} &= (E_3)^2 - 2D E_3 + \mathcal{K}_1 - \mathcal{K}_2 \Lambda \sqrt{\frac{\omega_0}{\omega_3}}, \label{eq6b} \\
    \frac{\partial\mathcal{L}}{\partial\omega_1} &= (E_1)^2 - 2D E_1 + \mathcal{K}_1 + \mathcal{K}_2 \Lambda \sqrt{\frac{\omega_2}{\omega_1}}, \label{eq6c} \\
    \frac{\partial\mathcal{L}}{\partial\omega_2} &= (E_2)^2 - 2D E_2 + \mathcal{K}_1 + \mathcal{K}_2 \Lambda \sqrt{\frac{\omega_1}{\omega_2}}, \label{eq6d} \\
    \frac{\partial\mathcal{L}}{\partial\mathcal{K}_1} &= \sum_p \omega_p - 1, \label{eq6e} \\
    \frac{\partial\mathcal{L}}{\partial\mathcal{K}_2} &= \mathcal{E}_G(\ket{\psi^{\text{in}}_{2,2}}) - G. \label{eq6f}
\end{align}
The definitions of $D$ and $\Lambda$ are
\begin{align}
D &= \sum_p \omega_p E_p,\nonumber \\ 
\Lambda &= \frac{\left( \sqrt{\omega_1 \omega_2} - \sqrt{\omega_0 \omega_3} \right)}{\sqrt{1 - 4 \left( \sqrt{\omega_1 \omega_2} - \sqrt{\omega_0 \omega_3} \right)^2}}.
\label{lA}
\end{align}
\section{Optimal QFI for Fixed Entanglement Entropy: The Case of Two Qubits}
In this section we present the set of four equations obtained by taking the partial derivatives of the Lagrangian, given in Eq.~\eqref{Lag2}, with respect to the parameters $\omega_p$, with $p=0,1,2,3$.

\begin{align}
    \frac{\partial\mathcal{L}}{\partial\omega_0} = (E_0)^2 - 2D E_0 + \mathcal{K}_3 - \mathcal{K}_4 \mathcal{M} \sqrt{\frac{\omega_3}{\omega_0}} = 0, \label{eq7a}\\
    \frac{\partial\mathcal{L}}{\partial\omega_3} = (E_3)^2 - 2D E_3 + \mathcal{K}_3 - \mathcal{K}_4 \mathcal{M} \sqrt{\frac{\omega_0}{\omega_3}} = 0, \label{eq7b}\\
    \frac{\partial\mathcal{L}}{\partial\omega_1} = (E_1)^2 - 2D E_1 + \mathcal{K}_3 + \mathcal{K}_4 \mathcal{M} \sqrt{\frac{\omega_2}{\omega_1}} = 0, \label{eq7c}\\
    \frac{\partial\mathcal{L}}{\partial\omega_2} = (E_2)^2 - 2D E_2 + \mathcal{K}_3 + \mathcal{K}_4 \mathcal{M} \sqrt{\frac{\omega_1}{\omega_2}} = 0, \label{eq7d}
\end{align}
where 
\begin{equation*}
\mathcal{M}=\frac{\sqrt{\omega_1\omega_2}-\sqrt{\omega_0\omega_3}}{\sqrt{1-4(\sqrt{\omega_1\omega_2}-\sqrt{\omega_0\omega_3})^2}}\log_2\frac{(1-p_s)}{p_s}.
\label{Ma}
\end{equation*}
\section{Proof That No Extrema Occur on the Boundary}
\label{Bound}
 In Theorem~\ref{t1} and Theorem~\ref{t2}, we employed the Lagrange multiplier method to find the maximum QFI for fixed input entanglement, corresponding to the interior points where $\omega_p > 0$ for all $p = 0,1,2,3$. In this section, we show that the maximum indeed lies in the interior, and that no extrema occur at the boundary points where at least one of the $\omega_p$ vanishes.

To begin with, we first pinpoint the fact that for $G\in[0,1/2]$ and $S\in[0,1]$. The relation between $S$ and $G$ is given as
$$
S = -G \log_2 G - (1 - G) \log_2 (1 - G).
$$

We shall use this relation in the analysis below. We also emphasize that $\omega_p$ cannot be zero for more than two values of $p$, as this would result in a product input state. For which it is already known that the optimal QFI is obtained corresponding  to the case where $\omega_p = 1/2$ for all $p$. 

Thus, the only possible boundary states with zero entanglement are the following:

\textbf{Case 1:} $\omega_p = 0$ for a single value of $p$;

\textbf{Case 2:} $\omega_p = 0$ for any two values of $p$.

\subsection{$\omega_p = 0$ for a Single Value of $p$}
We first examine Case 1, which can be further divided into two subcases:  
1) when $\omega_0$ or $\omega_3$ equals zero;  
2) when $\omega_1$ or $\omega_2$ equals zero.

\textit{When $\omega_0$ or $\omega_3 = 0$:}  
If $\omega_0$ or $\omega_3$ is equal to zero, then the maximum QFI corresponding to such boundary states with fixed $G/S$ can again be obtained using the Lagrange multiplier method. Note that in this case, the derivatives must be performed with respect to the remaining non-zero parameters. 

Specifically, when $\omega_0 = 0$, one needs to solve Eqs.~\eqref{eq6b}–\eqref{eq6d} for fixed $G$, or Eqs.~\eqref{eq7b}–\eqref{eq7d} for fixed $S$. On the other hand when $\omega_3 = 0$ one should solve Eqs.~\eqref{eq6a},\eqref{eq6c} and \eqref{eq6d} for fixed $G$ and for fixed $S$ one must solve Eqs.~\eqref{eq7a},\eqref{eq7c} and \eqref{eq7d}. Solving these equations, along with the equations we find $\omega_1=\omega_2$ for $\omega_3=0$. And two solutions for $\omega_0=0$ which are either $\omega_3=1/2$ or $\omega_1=\omega_2$.

Note for states with $\omega_3=1/2$ and $\omega_0=0$, yields $Q(\psi_{2,2}^\text{f})=4$. Clearly for  fixed $G$, this value is smaller than the optimal QFI obtained at the interior for the same $G$. 
Next for the boundary states  with $\omega_1=\omega_2$, the  QFI using Eqs.~\eqref{QFI} and (~\eqref{GGM}-$G=0$) and $\sum_p \omega_p=1$ is given as
$$
Q^{G/S}_b = 16\left( \sqrt{1 - (1 - 2G)^2} - (1 - (1 - 2G)^2) \right).
$$
Here, the subscript $b$ denotes the boundary. 
It is easy to verify that for any given $G \in (0, 1/2)$ or $S \in (0, 1)$,
$$
Q^{G/S}_b < Q^{G/S}.
$$
Thus, in this case, the maximum QFI is achieved by states lying in the interior of the $\{\omega_p\}$ space.

\textit{When $\omega_1$ or $\omega_2 = 0$:} 
For the scenario where $\omega_1 = 0$, one needs to solve Eqs.~\eqref{eq6a}, \eqref{eq6b}, and \eqref{eq6d} for fixed $G$, or Eqs.~\eqref{eq7a}, \eqref{eq7b}, and \eqref{eq7d} for fixed $S$. On the other hand, when $\omega_2 = 0$, one should solve Eqs.~\eqref{eq6a}–\eqref{eq6c} for fixed $G$, and Eqs.~\eqref{eq7a}–\eqref{eq7c} for fixed $S$.

Solving these equations along with the normalization constraint, we find two possible solutions for fixed $G/S$:  
1) $\omega_3 = \frac{1}{2} - \omega_0$;  
2) $\omega_3 = \omega_0$.

For the first case, the QFI using Eq.~\eqref{QFI} is given by  
$$
Q^{G/S}_b = 8 - 4(1 - 4\omega_0)^2,
$$  
which is smaller than $Q^{G/S}$.  
For the second case, where $\omega_3 = \omega_0$, using Eq.~\eqref{QFI} and(~\eqref{GGM}-$G=0$) we find the QFI to be  
$$
Q^{G/S}_b = 16\sqrt{1 - (1 - 2G)^2}.
$$  
In this case as well, one can verify that for $G \in (0, 1/2)$, we have $Q^{G/S}_b < Q^{G/S}$.  
Thus, we conclude that the optimum lies in the interior of the $\omega_p$ space.

 \subsection{$\omega_p = 0$ for Any Two Values of $p$}
This can be further divided into two sub-cases:  
1) when $\omega_0 = 0$ and $\omega_3 = 0$;  
2) when $\omega_1 = \omega_2 = 0$.  

For all other choices of pairs, it can be seen that the resulting state is simply a product state, for which the optimum is already known~\cite{M1}.

\textit{When $\omega_0$ and $\omega_3$ are zero:}  
For such a case, the QFI computed using Eq.~\eqref{QFI} is always zero.

\textit{When $\omega_1$ and $\omega_2$ are zero:}  
In this case, without loss of generality, we may assume that $\omega_0 < \omega_3$. This leads to $G = \omega_0$. The QFI corresponding to such a state using Eq.~\eqref{QFI} for fixed $G/S$ is given by 
$$
Q^{G/S}_b = 16\left(1 - (1 - 2G)^2\right).
$$  
Again, it is straightforward to verify that $Q^{G/S}_b < Q^{G/S}$ for $G \in (0, 1/2)$.
This completes our analysis of the optimal QFI for fixed entanglement at the boundary of the probability space. Our results thus suggest that, for the two-qubit case, the optimal QFI always corresponds to a state lying in the interior of the probability space.

\bibliography{metref}

\begin{thebibliography}{84}%
\makeatletter
\providecommand \@ifxundefined [1]{%
 \@ifx{#1\undefined}
}%
\providecommand \@ifnum [1]{%
 \ifnum #1\expandafter \@firstoftwo
 \else \expandafter \@secondoftwo
 \fi
}%
\providecommand \@ifx [1]{%
 \ifx #1\expandafter \@firstoftwo
 \else \expandafter \@secondoftwo
 \fi
}%
\providecommand \natexlab [1]{#1}%
\providecommand \enquote  [1]{``#1''}%
\providecommand \bibnamefont  [1]{#1}%
\providecommand \bibfnamefont [1]{#1}%
\providecommand \citenamefont [1]{#1}%
\providecommand \href@noop [0]{\@secondoftwo}%
\providecommand \href [0]{\begingroup \@sanitize@url \@href}%
\providecommand \@href[1]{\@@startlink{#1}\@@href}%
\providecommand \@@href[1]{\endgroup#1\@@endlink}%
\providecommand \@sanitize@url [0]{\catcode `\\12\catcode `\$12\catcode `\&12\catcode `\#12\catcode `\^12\catcode `\_12\catcode `\%12\relax}%
\providecommand \@@startlink[1]{}%
\providecommand \@@endlink[0]{}%
\providecommand \url  [0]{\begingroup\@sanitize@url \@url }%
\providecommand \@url [1]{\endgroup\@href {#1}{\urlprefix }}%
\providecommand \urlprefix  [0]{URL }%
\providecommand \Eprint [0]{\href }%
\providecommand \doibase [0]{http://dx.doi.org/}%
\providecommand \selectlanguage [0]{\@gobble}%
\providecommand \bibinfo  [0]{\@secondoftwo}%
\providecommand \bibfield  [0]{\@secondoftwo}%
\providecommand \translation [1]{[#1]}%
\providecommand \BibitemOpen [0]{}%
\providecommand \bibitemStop [0]{}%
\providecommand \bibitemNoStop [0]{.\EOS\space}%
\providecommand \EOS [0]{\spacefactor3000\relax}%
\providecommand \BibitemShut  [1]{\csname bibitem#1\endcsname}%
\let\auto@bib@innerbib\@empty
\bibitem [{\citenamefont {Caves}(1981)}]{met1}%
  \BibitemOpen
  \bibfield  {author} {\bibinfo {author} {\bibfnamefont {C.~M.}\ \bibnamefont {Caves}},\ }\bibfield  {title} {\enquote {\bibinfo {title} {Quantum-mechanical noise in an interferometer},}\ }\href {https://link.aps.org/doi/10.1103/PhysRevD.23.1693} {\bibfield  {journal} {\bibinfo  {journal} {Phys. Rev. D}\ }\textbf {\bibinfo {volume} {23}},\ \bibinfo {pages} {1693} (\bibinfo {year} {1981})}\BibitemShut {NoStop}%
\bibitem [{\citenamefont {Braunstein}\ and\ \citenamefont {Caves}(1994{\natexlab{a}})}]{met2}%
  \BibitemOpen
  \bibfield  {author} {\bibinfo {author} {\bibfnamefont {S.~L.}\ \bibnamefont {Braunstein}}\ and\ \bibinfo {author} {\bibfnamefont {C.~M.}\ \bibnamefont {Caves}},\ }\bibfield  {title} {\enquote {\bibinfo {title} {Statistical distance and the geometry of quantum states},}\ }\href {https://link.aps.org/doi/10.1103/PhysRevLett.72.3439} {\bibfield  {journal} {\bibinfo  {journal} {Phys. Rev. Lett.}\ }\textbf {\bibinfo {volume} {72}},\ \bibinfo {pages} {3439} (\bibinfo {year} {1994}{\natexlab{a}})}\BibitemShut {NoStop}%
\bibitem [{\citenamefont {Mitchell}\ \emph {et~al.}(2004)\citenamefont {Mitchell}, \citenamefont {Lundeen},\ and\ \citenamefont {Steinberg}}]{intf1}%
  \BibitemOpen
  \bibfield  {author} {\bibinfo {author} {\bibfnamefont {M.~W.}\ \bibnamefont {Mitchell}}, \bibinfo {author} {\bibfnamefont {J.~S.}\ \bibnamefont {Lundeen}}, \ and\ \bibinfo {author} {\bibfnamefont {A.~M.}\ \bibnamefont {Steinberg}},\ }\bibfield  {title} {\enquote {\bibinfo {title} {Super-resolving phase measurements with a multiphoton entangled state},}\ }\href {https://www.nature.com/articles/nature02493} {\bibfield  {journal} {\bibinfo  {journal} {Nat.}\ }\textbf {\bibinfo {volume} {429}},\ \bibinfo {pages} {161} (\bibinfo {year} {2004})}\BibitemShut {NoStop}%
\bibitem [{\citenamefont {Giovannetti}\ \emph {et~al.}(2006)\citenamefont {Giovannetti}, \citenamefont {Lloyd},\ and\ \citenamefont {Maccone}}]{M1}%
  \BibitemOpen
  \bibfield  {author} {\bibinfo {author} {\bibfnamefont {V.}~\bibnamefont {Giovannetti}}, \bibinfo {author} {\bibfnamefont {S.}~\bibnamefont {Lloyd}}, \ and\ \bibinfo {author} {\bibfnamefont {L.}~\bibnamefont {Maccone}},\ }\bibfield  {title} {\enquote {\bibinfo {title} {Quantum metrology},}\ }\href {https://link.aps.org/doi/10.1103/PhysRevLett.96.010401} {\bibfield  {journal} {\bibinfo  {journal} {Phys. Rev. Lett.}\ }\textbf {\bibinfo {volume} {96}},\ \bibinfo {pages} {010401} (\bibinfo {year} {2006})}\BibitemShut {NoStop}%
\bibitem [{\citenamefont {Holland}\ and\ \citenamefont {Burnett}(1993)}]{Enten}%
  \BibitemOpen
  \bibfield  {author} {\bibinfo {author} {\bibfnamefont {M.~J.}\ \bibnamefont {Holland}}\ and\ \bibinfo {author} {\bibfnamefont {K.}~\bibnamefont {Burnett}},\ }\bibfield  {title} {\enquote {\bibinfo {title} {Interferometric detection of optical phase shifts at the heisenberg limit},}\ }\href {https://link.aps.org/doi/10.1103/PhysRevLett.71.1355} {\bibfield  {journal} {\bibinfo  {journal} {Phys. Rev. Lett.}\ }\textbf {\bibinfo {volume} {71}},\ \bibinfo {pages} {1355} (\bibinfo {year} {1993})}\BibitemShut {NoStop}%
\bibitem [{\citenamefont {Lee}\ \emph {et~al.}(2002)\citenamefont {Lee}, \citenamefont {Kok},\ and\ \citenamefont {Dowling}}]{Enten2}%
  \BibitemOpen
  \bibfield  {author} {\bibinfo {author} {\bibfnamefont {H.}~\bibnamefont {Lee}}, \bibinfo {author} {\bibfnamefont {P.}~\bibnamefont {Kok}}, \ and\ \bibinfo {author} {\bibfnamefont {J.~P.}\ \bibnamefont {Dowling}},\ }\bibfield  {title} {\enquote {\bibinfo {title} {A quantum rosetta stone for interferometry},}\ }\href {https://www.tandfonline.com/doi/abs/10.1080/0950034021000011536} {\bibfield  {journal} {\bibinfo  {journal} {J. Mod. Opt}\ }\textbf {\bibinfo {volume} {49}},\ \bibinfo {pages} {2325} (\bibinfo {year} {2002})}\BibitemShut {NoStop}%
\bibitem [{\citenamefont {Eisenberg}\ \emph {et~al.}(2005)\citenamefont {Eisenberg}, \citenamefont {Hodelin}, \citenamefont {Khoury},\ and\ \citenamefont {Bouwmeester}}]{intf2}%
  \BibitemOpen
  \bibfield  {author} {\bibinfo {author} {\bibfnamefont {H.~S.}\ \bibnamefont {Eisenberg}}, \bibinfo {author} {\bibfnamefont {J.~F.}\ \bibnamefont {Hodelin}}, \bibinfo {author} {\bibfnamefont {G.}~\bibnamefont {Khoury}}, \ and\ \bibinfo {author} {\bibfnamefont {D.}~\bibnamefont {Bouwmeester}},\ }\bibfield  {title} {\enquote {\bibinfo {title} {Multiphoton path entanglement by nonlocal bunching},}\ }\href {https://link.aps.org/doi/10.1103/PhysRevLett.94.090502} {\bibfield  {journal} {\bibinfo  {journal} {Phys. Rev. Lett.}\ }\textbf {\bibinfo {volume} {94}},\ \bibinfo {pages} {090502} (\bibinfo {year} {2005})}\BibitemShut {NoStop}%
\bibitem [{\citenamefont {Giovannetti}\ \emph {et~al.}(2004)\citenamefont {Giovannetti}, \citenamefont {Lloyd},\ and\ \citenamefont {Maccone}}]{intf3}%
  \BibitemOpen
  \bibfield  {author} {\bibinfo {author} {\bibfnamefont {V.}~\bibnamefont {Giovannetti}}, \bibinfo {author} {\bibfnamefont {S.}~\bibnamefont {Lloyd}}, \ and\ \bibinfo {author} {\bibfnamefont {L.}~\bibnamefont {Maccone}},\ }\bibfield  {title} {\enquote {\bibinfo {title} {Quantum-enhanced measurements: beating the standard quantum limit},}\ }\href {https://doi.org/10.1126/science.1104149} {\bibfield  {journal} {\bibinfo  {journal} {Science}\ }\textbf {\bibinfo {volume} {306}},\ \bibinfo {pages} {1330} (\bibinfo {year} {2004})}\BibitemShut {NoStop}%
\bibitem [{\citenamefont {Colombo}\ \emph {et~al.}(2022)\citenamefont {Colombo}, \citenamefont {Pedrozo-Pe{\~n}afiel}, \citenamefont {Adiyatullin}, \citenamefont {Li}, \citenamefont {Mendez}, \citenamefont {Shu},\ and\ \citenamefont {Vuleti{\'c}}}]{Exp2}%
  \BibitemOpen
  \bibfield  {author} {\bibinfo {author} {\bibfnamefont {S.}~\bibnamefont {Colombo}}, \bibinfo {author} {\bibfnamefont {E.}~\bibnamefont {Pedrozo-Pe{\~n}afiel}}, \bibinfo {author} {\bibfnamefont {A.~F.}\ \bibnamefont {Adiyatullin}}, \bibinfo {author} {\bibfnamefont {Z.}~\bibnamefont {Li}}, \bibinfo {author} {\bibfnamefont {E.}~\bibnamefont {Mendez}}, \bibinfo {author} {\bibfnamefont {C.}~\bibnamefont {Shu}}, \ and\ \bibinfo {author} {\bibfnamefont {V.}~\bibnamefont {Vuleti{\'c}}},\ }\bibfield  {title} {\enquote {\bibinfo {title} {Time-reversal-based quantum metrology with many-body entangled states},}\ }\href {https://doi.org/10.1038/s41567-022-01653-5} {\bibfield  {journal} {\bibinfo  {journal} {Nat. Phys.}\ }\textbf {\bibinfo {volume} {18}},\ \bibinfo {pages} {925} (\bibinfo {year} {2022})}\BibitemShut {NoStop}%
\bibitem [{\citenamefont {Higgins}\ \emph {et~al.}(2009)\citenamefont {Higgins}, \citenamefont {Berry}, \citenamefont {Bartlett}, \citenamefont {Mitchell}, \citenamefont {Wiseman},\ and\ \citenamefont {Pryde}}]{intf5}%
  \BibitemOpen
  \bibfield  {author} {\bibinfo {author} {\bibfnamefont {B.~L.}\ \bibnamefont {Higgins}}, \bibinfo {author} {\bibfnamefont {D.~W.}\ \bibnamefont {Berry}}, \bibinfo {author} {\bibfnamefont {S.~D.}\ \bibnamefont {Bartlett}}, \bibinfo {author} {\bibfnamefont {M.~W.}\ \bibnamefont {Mitchell}}, \bibinfo {author} {\bibfnamefont {H.~M.}\ \bibnamefont {Wiseman}}, \ and\ \bibinfo {author} {\bibfnamefont {G.~J.}\ \bibnamefont {Pryde}},\ }\bibfield  {title} {\enquote {\bibinfo {title} {Demonstrating heisenberg-limited unambiguous phase estimation without adaptive measurements},}\ }\href {https://iopscience.iop.org/article/10.1088/1367-2630/11/7/073023/meta} {\bibfield  {journal} {\bibinfo  {journal} {NJP}\ }\textbf {\bibinfo {volume} {11}},\ \bibinfo {pages} {073023} (\bibinfo {year} {2009})}\BibitemShut {NoStop}%
\bibitem [{\citenamefont {Giovannetti}\ \emph {et~al.}(2011)\citenamefont {Giovannetti}, \citenamefont {Lloyd},\ and\ \citenamefont {Maccone}}]{Rev2}%
  \BibitemOpen
  \bibfield  {author} {\bibinfo {author} {\bibfnamefont {Vittorio}\ \bibnamefont {Giovannetti}}, \bibinfo {author} {\bibfnamefont {Seth}\ \bibnamefont {Lloyd}}, \ and\ \bibinfo {author} {\bibfnamefont {Lorenzo}\ \bibnamefont {Maccone}},\ }\bibfield  {title} {\enquote {\bibinfo {title} {Advances in quantum metrology},}\ }\href {https://www.nature.com/articles/nphoton.2011.35} {\bibfield  {journal} {\bibinfo  {journal} {Nat. Photonics}\ }\textbf {\bibinfo {volume} {5}},\ \bibinfo {pages} {222} (\bibinfo {year} {2011})}\BibitemShut {NoStop}%
\bibitem [{\citenamefont {Huelga}\ \emph {et~al.}(1997{\natexlab{a}})\citenamefont {Huelga}, \citenamefont {Macchiavello}, \citenamefont {Pellizzari}, \citenamefont {Ekert}, \citenamefont {Plenio},\ and\ \citenamefont {Cirac}}]{CLock1}%
  \BibitemOpen
  \bibfield  {author} {\bibinfo {author} {\bibfnamefont {S.~F.}\ \bibnamefont {Huelga}}, \bibinfo {author} {\bibfnamefont {C.}~\bibnamefont {Macchiavello}}, \bibinfo {author} {\bibfnamefont {T.}~\bibnamefont {Pellizzari}}, \bibinfo {author} {\bibfnamefont {A.~K.}\ \bibnamefont {Ekert}}, \bibinfo {author} {\bibfnamefont {M.~B.}\ \bibnamefont {Plenio}}, \ and\ \bibinfo {author} {\bibfnamefont {J.~I.}\ \bibnamefont {Cirac}},\ }\bibfield  {title} {\enquote {\bibinfo {title} {Improvement of frequency standards with quantum entanglement},}\ }\href {\doibase 10.1103/PhysRevLett.79.3865} {\bibfield  {journal} {\bibinfo  {journal} {Phys. Rev. Lett.}\ }\textbf {\bibinfo {volume} {79}},\ \bibinfo {pages} {3865} (\bibinfo {year} {1997}{\natexlab{a}})}\BibitemShut {NoStop}%
\bibitem [{\citenamefont {Andr\'e}\ \emph {et~al.}(2004)\citenamefont {Andr\'e}, \citenamefont {S\o{}rensen},\ and\ \citenamefont {Lukin}}]{Clock2}%
  \BibitemOpen
  \bibfield  {author} {\bibinfo {author} {\bibfnamefont {A.}~\bibnamefont {Andr\'e}}, \bibinfo {author} {\bibfnamefont {A.~S.}\ \bibnamefont {S\o{}rensen}}, \ and\ \bibinfo {author} {\bibfnamefont {M.~D.}\ \bibnamefont {Lukin}},\ }\bibfield  {title} {\enquote {\bibinfo {title} {Stability of atomic clocks based on entangled atoms},}\ }\href {https://link.aps.org/doi/10.1103/PhysRevLett.92.230801} {\bibfield  {journal} {\bibinfo  {journal} {Phys. Rev. Lett.}\ }\textbf {\bibinfo {volume} {92}},\ \bibinfo {pages} {230801} (\bibinfo {year} {2004})}\BibitemShut {NoStop}%
\bibitem [{\citenamefont {Ludlow}\ \emph {et~al.}(2015)\citenamefont {Ludlow}, \citenamefont {Boyd}, \citenamefont {Ye}, \citenamefont {Peik},\ and\ \citenamefont {Schmidt}}]{Cl1}%
  \BibitemOpen
  \bibfield  {author} {\bibinfo {author} {\bibfnamefont {A.~D.}\ \bibnamefont {Ludlow}}, \bibinfo {author} {\bibfnamefont {M.~M.}\ \bibnamefont {Boyd}}, \bibinfo {author} {\bibfnamefont {Jun}\ \bibnamefont {Ye}}, \bibinfo {author} {\bibfnamefont {E.}~\bibnamefont {Peik}}, \ and\ \bibinfo {author} {\bibfnamefont {P.~O.}\ \bibnamefont {Schmidt}},\ }\bibfield  {title} {\enquote {\bibinfo {title} {Optical atomic clocks},}\ }\href {\doibase 10.1103/RevModPhys.87.637} {\bibfield  {journal} {\bibinfo  {journal} {Rev. Mod. Phys.}\ }\textbf {\bibinfo {volume} {87}},\ \bibinfo {pages} {637} (\bibinfo {year} {2015})}\BibitemShut {NoStop}%
\bibitem [{\citenamefont {Auzinsh}\ \emph {et~al.}(2004)\citenamefont {Auzinsh}, \citenamefont {Budker}, \citenamefont {Kimball}, \citenamefont {Rochester}, \citenamefont {Stalnaker}, \citenamefont {Sushkov},\ and\ \citenamefont {Yashchuk}}]{Magn1}%
  \BibitemOpen
  \bibfield  {author} {\bibinfo {author} {\bibfnamefont {M.}~\bibnamefont {Auzinsh}}, \bibinfo {author} {\bibfnamefont {D.}~\bibnamefont {Budker}}, \bibinfo {author} {\bibfnamefont {D.~F.}\ \bibnamefont {Kimball}}, \bibinfo {author} {\bibfnamefont {S.~M.}\ \bibnamefont {Rochester}}, \bibinfo {author} {\bibfnamefont {J.~E.}\ \bibnamefont {Stalnaker}}, \bibinfo {author} {\bibfnamefont {A.~O.}\ \bibnamefont {Sushkov}}, \ and\ \bibinfo {author} {\bibfnamefont {V.~V.}\ \bibnamefont {Yashchuk}},\ }\bibfield  {title} {\enquote {\bibinfo {title} {Can a quantum nondemolition measurement improve the sensitivity of an atomic magnetometer?}}\ }\href {\doibase 10.1103/PhysRevLett.93.173002} {\bibfield  {journal} {\bibinfo  {journal} {Phys. Rev. Lett.}\ }\textbf {\bibinfo {volume} {93}},\ \bibinfo {pages} {173002} (\bibinfo {year} {2004})}\BibitemShut {NoStop}%
\bibitem [{\citenamefont {Petersen}\ \emph {et~al.}(2005)\citenamefont {Petersen}, \citenamefont {Madsen},\ and\ \citenamefont {M\o{}lmer}}]{Magn2}%
  \BibitemOpen
  \bibfield  {author} {\bibinfo {author} {\bibfnamefont {V.}~\bibnamefont {Petersen}}, \bibinfo {author} {\bibfnamefont {L.~B.}\ \bibnamefont {Madsen}}, \ and\ \bibinfo {author} {\bibfnamefont {K.}~\bibnamefont {M\o{}lmer}},\ }\bibfield  {title} {\enquote {\bibinfo {title} {Magnetometry with entangled atomic samples},}\ }\href {\doibase 10.1103/PhysRevA.71.012312} {\bibfield  {journal} {\bibinfo  {journal} {Phys. Rev. A}\ }\textbf {\bibinfo {volume} {71}},\ \bibinfo {pages} {012312} (\bibinfo {year} {2005})}\BibitemShut {NoStop}%
\bibitem [{\citenamefont {Mehboudi}\ \emph {et~al.}(2019)\citenamefont {Mehboudi}, \citenamefont {Sanpera},\ and\ \citenamefont {Correa}}]{Th1}%
  \BibitemOpen
  \bibfield  {author} {\bibinfo {author} {\bibfnamefont {M.}~\bibnamefont {Mehboudi}}, \bibinfo {author} {\bibfnamefont {A.}~\bibnamefont {Sanpera}}, \ and\ \bibinfo {author} {\bibfnamefont {L.~A.}\ \bibnamefont {Correa}},\ }\bibfield  {title} {\enquote {\bibinfo {title} {Thermometry in the quantum regime: recent theoretical progress},}\ }\href {https://doi.org/10.1088/1751-8121/ab2828} {\bibfield  {journal} {\bibinfo  {journal} {J. Phys. A. Theor. Math. Phys.}\ }\textbf {\bibinfo {volume} {52}},\ \bibinfo {pages} {303001} (\bibinfo {year} {2019})}\BibitemShut {NoStop}%
\bibitem [{\citenamefont {Yang}\ \emph {et~al.}(2021)\citenamefont {Yang}, \citenamefont {Chen}, \citenamefont {Li}, \citenamefont {Peng},\ and\ \citenamefont {Laflamme}}]{Exp1}%
  \BibitemOpen
  \bibfield  {author} {\bibinfo {author} {\bibfnamefont {X.}~\bibnamefont {Yang}}, \bibinfo {author} {\bibfnamefont {X.}~\bibnamefont {Chen}}, \bibinfo {author} {\bibfnamefont {J.}~\bibnamefont {Li}}, \bibinfo {author} {\bibfnamefont {X.}~\bibnamefont {Peng}}, \ and\ \bibinfo {author} {\bibfnamefont {R.}~\bibnamefont {Laflamme}},\ }\bibfield  {title} {\enquote {\bibinfo {title} {Hybrid quantum-classical approach to enhanced quantum metrology},}\ }\href {https://doi.org/10.1038/s41598-020-80070-1} {\bibfield  {journal} {\bibinfo  {journal} {Sci. Rep.}\ }\textbf {\bibinfo {volume} {11}},\ \bibinfo {pages} {672} (\bibinfo {year} {2021})}\BibitemShut {NoStop}%
\bibitem [{\citenamefont {Schnabel}\ \emph {et~al.}(2010)\citenamefont {Schnabel}, \citenamefont {Mavalvala}, \citenamefont {McClelland},\ and\ \citenamefont {Lam}}]{Rev10}%
  \BibitemOpen
  \bibfield  {author} {\bibinfo {author} {\bibfnamefont {R.}~\bibnamefont {Schnabel}}, \bibinfo {author} {\bibfnamefont {N.}~\bibnamefont {Mavalvala}}, \bibinfo {author} {\bibfnamefont {D.~E.}\ \bibnamefont {McClelland}}, \ and\ \bibinfo {author} {\bibfnamefont {P.~K.}\ \bibnamefont {Lam}},\ }\bibfield  {title} {\enquote {\bibinfo {title} {Quantum metrology for gravitational wave astronomy},}\ }\href {https://www.nature.com/articles/ncomms1122} {\bibfield  {journal} {\bibinfo  {journal} {Nat. Commun.}\ }\textbf {\bibinfo {volume} {1}},\ \bibinfo {pages} {121} (\bibinfo {year} {2010})}\BibitemShut {NoStop}%
\bibitem [{\citenamefont {T{\'o}th}\ and\ \citenamefont {Apellaniz}(2014)}]{Rev3}%
  \BibitemOpen
  \bibfield  {author} {\bibinfo {author} {\bibfnamefont {G.}~\bibnamefont {T{\'o}th}}\ and\ \bibinfo {author} {\bibfnamefont {I.}~\bibnamefont {Apellaniz}},\ }\bibfield  {title} {\enquote {\bibinfo {title} {Quantum metrology from a quantum information science perspective},}\ }\href {https://iopscience.iop.org/article/10.1088/1751-8113/47/42/424006/meta} {\bibfield  {journal} {\bibinfo  {journal} {J. Phys. A: Math. Theor.}\ }\textbf {\bibinfo {volume} {47}},\ \bibinfo {pages} {424006} (\bibinfo {year} {2014})}\BibitemShut {NoStop}%
\bibitem [{\citenamefont {Laurenza}\ \emph {et~al.}(2018)\citenamefont {Laurenza}, \citenamefont {Lupo}, \citenamefont {Spedalieri}, \citenamefont {Braunstein},\ and\ \citenamefont {Pirandola}}]{Rev18}%
  \BibitemOpen
  \bibfield  {author} {\bibinfo {author} {\bibfnamefont {R.}~\bibnamefont {Laurenza}}, \bibinfo {author} {\bibfnamefont {C.}~\bibnamefont {Lupo}}, \bibinfo {author} {\bibfnamefont {G.}~\bibnamefont {Spedalieri}}, \bibinfo {author} {\bibfnamefont {S.~L.}\ \bibnamefont {Braunstein}}, \ and\ \bibinfo {author} {\bibfnamefont {S.}~\bibnamefont {Pirandola}},\ }\bibfield  {title} {\enquote {\bibinfo {title} {Channel simulation in quantum metrology},}\ }\href {https://www.degruyterbrill.com/document/doi/10.1515/qmetro-2018-0001/html} {\bibfield  {journal} {\bibinfo  {journal} {Quantum Meas. Quantum Metrol}\ }\textbf {\bibinfo {volume} {5}},\ \bibinfo {pages} {1} (\bibinfo {year} {2018})}\BibitemShut {NoStop}%
\bibitem [{\citenamefont {Fadel}\ \emph {et~al.}(2024)\citenamefont {Fadel}, \citenamefont {Roux},\ and\ \citenamefont {Gessner}}]{rev24}%
  \BibitemOpen
  \bibfield  {author} {\bibinfo {author} {\bibfnamefont {M.}~\bibnamefont {Fadel}}, \bibinfo {author} {\bibfnamefont {N.}~\bibnamefont {Roux}}, \ and\ \bibinfo {author} {\bibfnamefont {M.}~\bibnamefont {Gessner}},\ }\bibfield  {title} {\enquote {\bibinfo {title} {Quantum metrology with a continuous-variable system},}\ }\href {https://doi.org/10.48550/arXiv.2411.04122} {\bibfield  {journal} {\bibinfo  {journal} {arXiv preprint arXiv:2411.04122}\ } (\bibinfo {year} {2024})}\BibitemShut {NoStop}%
\bibitem [{\citenamefont {Mukhopadhyay}\ \emph {et~al.}(2024)\citenamefont {Mukhopadhyay}, \citenamefont {Montenegro},\ and\ \citenamefont {Bayat}}]{Rev1}%
  \BibitemOpen
  \bibfield  {author} {\bibinfo {author} {\bibfnamefont {C.}~\bibnamefont {Mukhopadhyay}}, \bibinfo {author} {\bibfnamefont {V.}~\bibnamefont {Montenegro}}, \ and\ \bibinfo {author} {\bibfnamefont {Ab.}\ \bibnamefont {Bayat}},\ }\bibfield  {title} {\enquote {\bibinfo {title} {Current trends in global quantum metrology},}\ }\href {https://iopscience.iop.org/article/10.1088/1751-8121/adb112} {\bibfield  {journal} {\bibinfo  {journal} {J. Phys. A: Math. Theor.}\ }\textbf {\bibinfo {volume} {58}} (\bibinfo {year} {2024})}\BibitemShut {NoStop}%
\bibitem [{\citenamefont {Horodecki}\ \emph {et~al.}(2009)\citenamefont {Horodecki}, \citenamefont {Horodecki}, \citenamefont {Horodecki},\ and\ \citenamefont {Horodecki}}]{Entr1}%
  \BibitemOpen
  \bibfield  {author} {\bibinfo {author} {\bibfnamefont {R.}~\bibnamefont {Horodecki}}, \bibinfo {author} {\bibfnamefont {P.}~\bibnamefont {Horodecki}}, \bibinfo {author} {\bibfnamefont {M.}~\bibnamefont {Horodecki}}, \ and\ \bibinfo {author} {\bibfnamefont {K.}~\bibnamefont {Horodecki}},\ }\bibfield  {title} {\enquote {\bibinfo {title} {Quantum entanglement},}\ }\href {\doibase 10.1103/RevModPhys.81.865} {\bibfield  {journal} {\bibinfo  {journal} {Rev. Mod. Phys.}\ }\textbf {\bibinfo {volume} {81}},\ \bibinfo {pages} {865} (\bibinfo {year} {2009})}\BibitemShut {NoStop}%
\bibitem [{\citenamefont {G{\"u}hne}\ and\ \citenamefont {T{\'o}th}(2009)}]{pro1}%
  \BibitemOpen
  \bibfield  {author} {\bibinfo {author} {\bibfnamefont {O.}~\bibnamefont {G{\"u}hne}}\ and\ \bibinfo {author} {\bibfnamefont {G.}~\bibnamefont {T{\'o}th}},\ }\bibfield  {title} {\enquote {\bibinfo {title} {Entanglement detection},}\ }\href {https://doi.org/10.1016/j.physrep.2009.02.004} {\bibfield  {journal} {\bibinfo  {journal} {Phys. Rep.}\ }\textbf {\bibinfo {volume} {474}},\ \bibinfo {pages} {1} (\bibinfo {year} {2009})}\BibitemShut {NoStop}%
\bibitem [{\citenamefont {Plenio}\ and\ \citenamefont {Virmani}(2014)}]{Entr2}%
  \BibitemOpen
  \bibfield  {author} {\bibinfo {author} {\bibfnamefont {M.~B.}\ \bibnamefont {Plenio}}\ and\ \bibinfo {author} {\bibfnamefont {S.~S.}\ \bibnamefont {Virmani}},\ }\href {https://doi.org/10.1007/978-3-319-04063-9_8} {\emph {\bibinfo {title} {An introduction to entanglement theory}}}\ (\bibinfo  {publisher} {Springer},\ \bibinfo {year} {2014})\ p.\ \bibinfo {pages} {173}\BibitemShut {NoStop}%
\bibitem [{\citenamefont {Das}\ \emph {et~al.}(2016)\citenamefont {Das}, \citenamefont {Chanda}, \citenamefont {Lewenstein}, \citenamefont {Sanpera}, \citenamefont {Sen~De},\ and\ \citenamefont {Sen}}]{GGM3}%
  \BibitemOpen
  \bibfield  {author} {\bibinfo {author} {\bibfnamefont {S.}~\bibnamefont {Das}}, \bibinfo {author} {\bibfnamefont {T.}~\bibnamefont {Chanda}}, \bibinfo {author} {\bibfnamefont {M.}~\bibnamefont {Lewenstein}}, \bibinfo {author} {\bibfnamefont {A.}~\bibnamefont {Sanpera}}, \bibinfo {author} {\bibfnamefont {A.}~\bibnamefont {Sen~De}}, \ and\ \bibinfo {author} {\bibfnamefont {U.}~\bibnamefont {Sen}},\ }\bibfield  {title} {\enquote {\bibinfo {title} {The separability versus entanglement problem},}\ }\href {https://doi.org/10.48550/arXiv.1701.02187} {\bibfield  {journal} {\bibinfo  {journal} {QIF-QTA or Quant. Inf. from Found. to Tech. Appl.}\ ,\ \bibinfo {pages} {127}} (\bibinfo {year} {2016})}\BibitemShut {NoStop}%
\bibitem [{\citenamefont {Baumgratz}\ \emph {et~al.}(2014)\citenamefont {Baumgratz}, \citenamefont {Cramer},\ and\ \citenamefont {Plenio}}]{co1}%
  \BibitemOpen
  \bibfield  {author} {\bibinfo {author} {\bibfnamefont {T.}~\bibnamefont {Baumgratz}}, \bibinfo {author} {\bibfnamefont {M.}~\bibnamefont {Cramer}}, \ and\ \bibinfo {author} {\bibfnamefont {M.~B.}\ \bibnamefont {Plenio}},\ }\bibfield  {title} {\enquote {\bibinfo {title} {Quantifying coherence},}\ }\href {https://link.aps.org/doi/10.1103/PhysRevLett.113.140401} {\bibfield  {journal} {\bibinfo  {journal} {Phys. Rev. Lett.}\ }\textbf {\bibinfo {volume} {113}},\ \bibinfo {pages} {140401} (\bibinfo {year} {2014})}\BibitemShut {NoStop}%
\bibitem [{\citenamefont {Aberg}(2006)}]{CoAberg}%
  \BibitemOpen
  \bibfield  {author} {\bibinfo {author} {\bibfnamefont {J.}~\bibnamefont {Aberg}},\ }\bibfield  {title} {\enquote {\bibinfo {title} {Quantifying superposition},}\ }\href {https://doi.org/10.48550/arXiv.quant-ph/0612146} {\bibfield  {journal} {\bibinfo  {journal} {arXiv preprint quant-ph/0612146}\ } (\bibinfo {year} {2006})}\BibitemShut {NoStop}%
\bibitem [{\citenamefont {Rivas}\ \emph {et~al.}(2014)\citenamefont {Rivas}, \citenamefont {Huelga},\ and\ \citenamefont {Plenio}}]{NMRe}%
  \BibitemOpen
  \bibfield  {author} {\bibinfo {author} {\bibfnamefont {{\'A}.}~\bibnamefont {Rivas}}, \bibinfo {author} {\bibfnamefont {S.~F.}\ \bibnamefont {Huelga}}, \ and\ \bibinfo {author} {\bibfnamefont {M.~B.}\ \bibnamefont {Plenio}},\ }\bibfield  {title} {\enquote {\bibinfo {title} {Quantum non-{M}arkovianity: characterization, quantification and detection},}\ }\href {https://doi.org/10.1088/0034-4885/77/9/094001} {\bibfield  {journal} {\bibinfo  {journal} {Rep. Prog. Phys.}\ }\textbf {\bibinfo {volume} {77}},\ \bibinfo {pages} {094001} (\bibinfo {year} {2014})}\BibitemShut {NoStop}%
\bibitem [{\citenamefont {Breuer}\ \emph {et~al.}(2016)\citenamefont {Breuer}, \citenamefont {Laine}, \citenamefont {Piilo},\ and\ \citenamefont {Vacchini}}]{NMRe1}%
  \BibitemOpen
  \bibfield  {author} {\bibinfo {author} {\bibfnamefont {H.-P.}\ \bibnamefont {Breuer}}, \bibinfo {author} {\bibfnamefont {E.-M.}\ \bibnamefont {Laine}}, \bibinfo {author} {\bibfnamefont {J.}~\bibnamefont {Piilo}}, \ and\ \bibinfo {author} {\bibfnamefont {B.}~\bibnamefont {Vacchini}},\ }\bibfield  {title} {\enquote {\bibinfo {title} {Colloquium: Non-{M}arkovian dynamics in open quantum systems},}\ }\href {\doibase 10.1103/RevModPhys.88.021002} {\bibfield  {journal} {\bibinfo  {journal} {Rev. Mod. Phys.}\ }\textbf {\bibinfo {volume} {88}},\ \bibinfo {pages} {021002} (\bibinfo {year} {2016})}\BibitemShut {NoStop}%
\bibitem [{\citenamefont {de~Vega}\ and\ \citenamefont {Alonso}(2017)}]{NMRe2}%
  \BibitemOpen
  \bibfield  {author} {\bibinfo {author} {\bibfnamefont {I.}~\bibnamefont {de~Vega}}\ and\ \bibinfo {author} {\bibfnamefont {D.}~\bibnamefont {Alonso}},\ }\bibfield  {title} {\enquote {\bibinfo {title} {Dynamics of non-{M}arkovian open quantum systems},}\ }\href {https://link.aps.org/doi/10.1103/RevModPhys.89.015001} {\bibfield  {journal} {\bibinfo  {journal} {Rev. Mod. Phys.}\ }\textbf {\bibinfo {volume} {89}},\ \bibinfo {pages} {015001} (\bibinfo {year} {2017})}\BibitemShut {NoStop}%
\bibitem [{\citenamefont {Shrikant}\ and\ \citenamefont {Mandayam}(2023)}]{NMRe3}%
  \BibitemOpen
  \bibfield  {author} {\bibinfo {author} {\bibfnamefont {U.}~\bibnamefont {Shrikant}}\ and\ \bibinfo {author} {\bibfnamefont {P.}~\bibnamefont {Mandayam}},\ }\bibfield  {title} {\enquote {\bibinfo {title} {Quantum non-{M}arkovianity: Overview and recent developments},}\ }\href {https://www.frontiersin.org/journals/quantum-science-and-technology/articles/10.3389/frqst.2023.1134583/full} {\bibfield  {journal} {\bibinfo  {journal} {Front. Quantum Sci. Technol.}\ }\textbf {\bibinfo {volume} {2}},\ \bibinfo {pages} {1134583} (\bibinfo {year} {2023})}\BibitemShut {NoStop}%
\bibitem [{\citenamefont {Walls}(1983)}]{SQ1}%
  \BibitemOpen
  \bibfield  {author} {\bibinfo {author} {\bibfnamefont {D.~F.}\ \bibnamefont {Walls}},\ }\bibfield  {title} {\enquote {\bibinfo {title} {Squeezed states of light},}\ }\href {https://www.nature.com/articles/306141a0} {\bibfield  {journal} {\bibinfo  {journal} {Nat.}\ }\textbf {\bibinfo {volume} {306}},\ \bibinfo {pages} {141} (\bibinfo {year} {1983})}\BibitemShut {NoStop}%
\bibitem [{\citenamefont {Slusher}\ \emph {et~al.}(1985)\citenamefont {Slusher}, \citenamefont {Hollberg}, \citenamefont {Yurke}, \citenamefont {Mertz},\ and\ \citenamefont {Valley}}]{SQ2}%
  \BibitemOpen
  \bibfield  {author} {\bibinfo {author} {\bibfnamefont {R.~E.}\ \bibnamefont {Slusher}}, \bibinfo {author} {\bibfnamefont {L.~W.}\ \bibnamefont {Hollberg}}, \bibinfo {author} {\bibfnamefont {B.}~\bibnamefont {Yurke}}, \bibinfo {author} {\bibfnamefont {J.~C.}\ \bibnamefont {Mertz}}, \ and\ \bibinfo {author} {\bibfnamefont {J.~F.}\ \bibnamefont {Valley}},\ }\bibfield  {title} {\enquote {\bibinfo {title} {Observation of squeezed states generated by four-wave mixing in an optical cavity},}\ }\href {https://link.aps.org/doi/10.1103/PhysRevLett.55.2409} {\bibfield  {journal} {\bibinfo  {journal} {Phys. Rev. Lett.}\ }\textbf {\bibinfo {volume} {55}},\ \bibinfo {pages} {2409} (\bibinfo {year} {1985})}\BibitemShut {NoStop}%
\bibitem [{\citenamefont {Kitagawa}\ and\ \citenamefont {Ueda}(1993)}]{SQ3}%
  \BibitemOpen
  \bibfield  {author} {\bibinfo {author} {\bibfnamefont {M.}~\bibnamefont {Kitagawa}}\ and\ \bibinfo {author} {\bibfnamefont {M.}~\bibnamefont {Ueda}},\ }\bibfield  {title} {\enquote {\bibinfo {title} {Squeezed spin states},}\ }\href {https://link.aps.org/doi/10.1103/PhysRevA.47.5138} {\bibfield  {journal} {\bibinfo  {journal} {Phys. Rev. A}\ }\textbf {\bibinfo {volume} {47}},\ \bibinfo {pages} {5138} (\bibinfo {year} {1993})}\BibitemShut {NoStop}%
\bibitem [{\citenamefont {Wineland}\ \emph {et~al.}(1992)\citenamefont {Wineland}, \citenamefont {Bollinger}, \citenamefont {Itano}, \citenamefont {Moore},\ and\ \citenamefont {Heinzen}}]{uni1}%
  \BibitemOpen
  \bibfield  {author} {\bibinfo {author} {\bibfnamefont {D.~J.}\ \bibnamefont {Wineland}}, \bibinfo {author} {\bibfnamefont {J.~J.}\ \bibnamefont {Bollinger}}, \bibinfo {author} {\bibfnamefont {W.~M.}\ \bibnamefont {Itano}}, \bibinfo {author} {\bibfnamefont {F.~L.}\ \bibnamefont {Moore}}, \ and\ \bibinfo {author} {\bibfnamefont {D.~J.}\ \bibnamefont {Heinzen}},\ }\bibfield  {title} {\enquote {\bibinfo {title} {Spin squeezing and reduced quantum noise in spectroscopy},}\ }\href {\doibase 10.1103/PhysRevA.46.R6797} {\bibfield  {journal} {\bibinfo  {journal} {Phys. Rev. A}\ }\textbf {\bibinfo {volume} {46}},\ \bibinfo {pages} {R6797} (\bibinfo {year} {1992})}\BibitemShut {NoStop}%
\bibitem [{\citenamefont {Braunstein}(1992)}]{Enten3}%
  \BibitemOpen
  \bibfield  {author} {\bibinfo {author} {\bibfnamefont {S.~L.}\ \bibnamefont {Braunstein}},\ }\bibfield  {title} {\enquote {\bibinfo {title} {Quantum limits on precision measurements of phase},}\ }\href {https://link.aps.org/doi/10.1103/PhysRevLett.69.3598} {\bibfield  {journal} {\bibinfo  {journal} {Phys. Rev. Lett.}\ }\textbf {\bibinfo {volume} {69}},\ \bibinfo {pages} {3598} (\bibinfo {year} {1992})}\BibitemShut {NoStop}%
\bibitem [{\citenamefont {T\'oth}(2012)}]{Toth}%
  \BibitemOpen
  \bibfield  {author} {\bibinfo {author} {\bibfnamefont {G.}~\bibnamefont {T\'oth}},\ }\bibfield  {title} {\enquote {\bibinfo {title} {Multipartite entanglement and high-precision metrology},}\ }\href {https://link.aps.org/doi/10.1103/PhysRevA.85.022322} {\bibfield  {journal} {\bibinfo  {journal} {Phys. Rev. A}\ }\textbf {\bibinfo {volume} {85}},\ \bibinfo {pages} {022322} (\bibinfo {year} {2012})}\BibitemShut {NoStop}%
\bibitem [{\citenamefont {Erol}\ \emph {et~al.}(2014)\citenamefont {Erol}, \citenamefont {Ozaydin},\ and\ \citenamefont {Altintas}}]{Ent5}%
  \BibitemOpen
  \bibfield  {author} {\bibinfo {author} {\bibfnamefont {V.}~\bibnamefont {Erol}}, \bibinfo {author} {\bibfnamefont {F.}~\bibnamefont {Ozaydin}}, \ and\ \bibinfo {author} {\bibfnamefont {A.~A.}\ \bibnamefont {Altintas}},\ }\bibfield  {title} {\enquote {\bibinfo {title} {Analysis of entanglement measures and locc maximized quantum fisher information of general two qubit systems},}\ }\href {https://www.nature.com/articles/srep05422} {\bibfield  {journal} {\bibinfo  {journal} {Sci. Rep.}\ }\textbf {\bibinfo {volume} {4}},\ \bibinfo {pages} {5422} (\bibinfo {year} {2014})}\BibitemShut {NoStop}%
\bibitem [{\citenamefont {Apellaniz}\ \emph {et~al.}(2015)\citenamefont {Apellaniz}, \citenamefont {L{\"u}cke}, \citenamefont {Peise}, \citenamefont {Klempt},\ and\ \citenamefont {T{\'o}th}}]{Ent4}%
  \BibitemOpen
  \bibfield  {author} {\bibinfo {author} {\bibfnamefont {I.}~\bibnamefont {Apellaniz}}, \bibinfo {author} {\bibfnamefont {B.}~\bibnamefont {L{\"u}cke}}, \bibinfo {author} {\bibfnamefont {J.}~\bibnamefont {Peise}}, \bibinfo {author} {\bibfnamefont {C.}~\bibnamefont {Klempt}}, \ and\ \bibinfo {author} {\bibfnamefont {G.}~\bibnamefont {T{\'o}th}},\ }\bibfield  {title} {\enquote {\bibinfo {title} {Detecting metrologically useful entanglement in the vicinity of dicke states},}\ }\href {https://doi.org/10.1088/1367-2630/17/8/083027 Focus to learn more} {\bibfield  {journal} {\bibinfo  {journal} {NJP}\ }\textbf {\bibinfo {volume} {17}},\ \bibinfo {pages} {083027} (\bibinfo {year} {2015})}\BibitemShut {NoStop}%
\bibitem [{\citenamefont {Tr{\'e}nyi}\ \emph {et~al.}(2024)\citenamefont {Tr{\'e}nyi}, \citenamefont {Luk{\'a}cs}, \citenamefont {Horodecki}, \citenamefont {Horodecki}, \citenamefont {V{\'e}rtesi},\ and\ \citenamefont {T{\'o}th}}]{Ent6}%
  \BibitemOpen
  \bibfield  {author} {\bibinfo {author} {\bibfnamefont {R.}~\bibnamefont {Tr{\'e}nyi}}, \bibinfo {author} {\bibfnamefont {{\'A}.}~\bibnamefont {Luk{\'a}cs}}, \bibinfo {author} {\bibfnamefont {P.}~\bibnamefont {Horodecki}}, \bibinfo {author} {\bibfnamefont {R.}~\bibnamefont {Horodecki}}, \bibinfo {author} {\bibfnamefont {T.}~\bibnamefont {V{\'e}rtesi}}, \ and\ \bibinfo {author} {\bibfnamefont {G.}~\bibnamefont {T{\'o}th}},\ }\bibfield  {title} {\enquote {\bibinfo {title} {Activation of metrologically useful genuine multipartite entanglement},}\ }\href {https://doi.org/10.1088/1367-2630/ad1e93 Focus to learn more} {\bibfield  {journal} {\bibinfo  {journal} {NJP}\ }\textbf {\bibinfo {volume} {26}},\ \bibinfo {pages} {023034} (\bibinfo {year} {2024})}\BibitemShut {NoStop}%
\bibitem [{\citenamefont {Fujiwara}(2001)}]{aux1}%
  \BibitemOpen
  \bibfield  {author} {\bibinfo {author} {\bibfnamefont {A.}~\bibnamefont {Fujiwara}},\ }\bibfield  {title} {\enquote {\bibinfo {title} {Quantum channel identification problem},}\ }\href {https://link.aps.org/doi/10.1103/PhysRevA.63.042304} {\bibfield  {journal} {\bibinfo  {journal} {Phys. Rev. A}\ }\textbf {\bibinfo {volume} {63}},\ \bibinfo {pages} {042304} (\bibinfo {year} {2001})}\BibitemShut {NoStop}%
\bibitem [{\citenamefont {Demkowicz-Dobrza\ifmmode~\acute{n}\else \'{n}\fi{}ski}\ and\ \citenamefont {Maccone}(2014)}]{Aux14}%
  \BibitemOpen
  \bibfield  {author} {\bibinfo {author} {\bibfnamefont {R.}~\bibnamefont {Demkowicz-Dobrza\ifmmode~\acute{n}\else \'{n}\fi{}ski}}\ and\ \bibinfo {author} {\bibfnamefont {L.}~\bibnamefont {Maccone}},\ }\bibfield  {title} {\enquote {\bibinfo {title} {Using entanglement against noise in quantum metrology},}\ }\href {https://link.aps.org/doi/10.1103/PhysRevLett.113.250801} {\bibfield  {journal} {\bibinfo  {journal} {Phys. Rev. Lett.}\ }\textbf {\bibinfo {volume} {113}},\ \bibinfo {pages} {250801} (\bibinfo {year} {2014})}\BibitemShut {NoStop}%
\bibitem [{\citenamefont {Huang}\ \emph {et~al.}(2016)\citenamefont {Huang}, \citenamefont {Macchiavello},\ and\ \citenamefont {Maccone}}]{Auxn}%
  \BibitemOpen
  \bibfield  {author} {\bibinfo {author} {\bibfnamefont {Z.}~\bibnamefont {Huang}}, \bibinfo {author} {\bibfnamefont {C.}~\bibnamefont {Macchiavello}}, \ and\ \bibinfo {author} {\bibfnamefont {L.}~\bibnamefont {Maccone}},\ }\bibfield  {title} {\enquote {\bibinfo {title} {Usefulness of entanglement-assisted quantum metrology},}\ }\href {\doibase 10.1103/PhysRevA.94.012101} {\bibfield  {journal} {\bibinfo  {journal} {Phys. Rev. A}\ }\textbf {\bibinfo {volume} {94}},\ \bibinfo {pages} {012101} (\bibinfo {year} {2016})}\BibitemShut {NoStop}%
\bibitem [{\citenamefont {Nichols}\ \emph {et~al.}(2016)\citenamefont {Nichols}, \citenamefont {Bromley}, \citenamefont {Correa},\ and\ \citenamefont {Adesso}}]{Aux16}%
  \BibitemOpen
  \bibfield  {author} {\bibinfo {author} {\bibfnamefont {R.}~\bibnamefont {Nichols}}, \bibinfo {author} {\bibfnamefont {T.~R.}\ \bibnamefont {Bromley}}, \bibinfo {author} {\bibfnamefont {L.~A.}\ \bibnamefont {Correa}}, \ and\ \bibinfo {author} {\bibfnamefont {G.}~\bibnamefont {Adesso}},\ }\bibfield  {title} {\enquote {\bibinfo {title} {Practical quantum metrology in noisy environments},}\ }\href {https://link.aps.org/doi/10.1103/PhysRevA.94.042101} {\bibfield  {journal} {\bibinfo  {journal} {Phys. Rev. A}\ }\textbf {\bibinfo {volume} {94}},\ \bibinfo {pages} {042101} (\bibinfo {year} {2016})}\BibitemShut {NoStop}%
\bibitem [{\citenamefont {Zhou}(2024{\natexlab{a}})}]{Aux23}%
  \BibitemOpen
  \bibfield  {author} {\bibinfo {author} {\bibfnamefont {S.}~\bibnamefont {Zhou}},\ }\bibfield  {title} {\enquote {\bibinfo {title} {Limits of noisy quantum metrology with restricted quantum controls},}\ }\href {https://link.aps.org/doi/10.1103/PhysRevLett.133.170801} {\bibfield  {journal} {\bibinfo  {journal} {Phys. Rev. Lett.}\ }\textbf {\bibinfo {volume} {133}},\ \bibinfo {pages} {170801} (\bibinfo {year} {2024}{\natexlab{a}})}\BibitemShut {NoStop}%
\bibitem [{\citenamefont {Rahim}\ \emph {et~al.}(2024)\citenamefont {Rahim}, \citenamefont {Al-Kuwari},\ and\ \citenamefont {Ali}}]{aux5}%
  \BibitemOpen
  \bibfield  {author} {\bibinfo {author} {\bibfnamefont {M.~T.}\ \bibnamefont {Rahim}}, \bibinfo {author} {\bibfnamefont {S.}~\bibnamefont {Al-Kuwari}}, \ and\ \bibinfo {author} {\bibfnamefont {A.}~\bibnamefont {Ali}},\ }\bibfield  {title} {\enquote {\bibinfo {title} {Entanglement-enhanced optimal quantum metrology},}\ }\href {https://doi.org/10.48550/arXiv.2411.04022} {\bibfield  {journal} {\bibinfo  {journal} {arXiv preprint arXiv:2411.04022}\ } (\bibinfo {year} {2024})}\BibitemShut {NoStop}%
\bibitem [{\citenamefont {Pezz\'e}\ and\ \citenamefont {Smerzi}(2009)}]{Ent2}%
  \BibitemOpen
  \bibfield  {author} {\bibinfo {author} {\bibfnamefont {L.}~\bibnamefont {Pezz\'e}}\ and\ \bibinfo {author} {\bibfnamefont {A.}~\bibnamefont {Smerzi}},\ }\bibfield  {title} {\enquote {\bibinfo {title} {Entanglement, nonlinear dynamics, and the {H}eisenberg limit},}\ }\href {\doibase 10.1103/PhysRevLett.102.100401} {\bibfield  {journal} {\bibinfo  {journal} {Phys. Rev. Lett.}\ }\textbf {\bibinfo {volume} {102}},\ \bibinfo {pages} {100401} (\bibinfo {year} {2009})}\BibitemShut {NoStop}%
\bibitem [{\citenamefont {Wang}\ \emph {et~al.}(2018)\citenamefont {Wang}, \citenamefont {Wu}, \citenamefont {Cui},\ and\ \citenamefont {Wang}}]{cohf}%
  \BibitemOpen
  \bibfield  {author} {\bibinfo {author} {\bibfnamefont {Zhihai}\ \bibnamefont {Wang}}, \bibinfo {author} {\bibfnamefont {Wei}\ \bibnamefont {Wu}}, \bibinfo {author} {\bibfnamefont {Guodong}\ \bibnamefont {Cui}}, \ and\ \bibinfo {author} {\bibfnamefont {Jin}\ \bibnamefont {Wang}},\ }\bibfield  {title} {\enquote {\bibinfo {title} {Coherence enhanced quantum metrology in a nonequilibrium optical molecule},}\ }\href {https://iopscience.iop.org/article/10.1088/1367-2630/aab03a} {\bibfield  {journal} {\bibinfo  {journal} {NJP}\ }\textbf {\bibinfo {volume} {20}},\ \bibinfo {pages} {033034} (\bibinfo {year} {2018})}\BibitemShut {NoStop}%
\bibitem [{\citenamefont {Chin}\ \emph {et~al.}(2012)\citenamefont {Chin}, \citenamefont {Huelga},\ and\ \citenamefont {Plenio}}]{NMR}%
  \BibitemOpen
  \bibfield  {author} {\bibinfo {author} {\bibfnamefont {A.~W.}\ \bibnamefont {Chin}}, \bibinfo {author} {\bibfnamefont {S.~F.}\ \bibnamefont {Huelga}}, \ and\ \bibinfo {author} {\bibfnamefont {M.~B.}\ \bibnamefont {Plenio}},\ }\bibfield  {title} {\enquote {\bibinfo {title} {Quantum metrology in non-{M}arkovian environments},}\ }\href {https://link.aps.org/doi/10.1103/PhysRevLett.109.233601} {\bibfield  {journal} {\bibinfo  {journal} {Phys. Rev. Lett.}\ }\textbf {\bibinfo {volume} {109}},\ \bibinfo {pages} {233601} (\bibinfo {year} {2012})}\BibitemShut {NoStop}%
\bibitem [{\citenamefont {Xie}\ and\ \citenamefont {Wang}(2014)}]{NMarkov}%
  \BibitemOpen
  \bibfield  {author} {\bibinfo {author} {\bibfnamefont {D.}~\bibnamefont {Xie}}\ and\ \bibinfo {author} {\bibfnamefont {A.~M.}\ \bibnamefont {Wang}},\ }\bibfield  {title} {\enquote {\bibinfo {title} {Quantum metrology in correlated environments},}\ }\href {https://doi.org/10.1016/j.physleta.2014.06.006} {\bibfield  {journal} {\bibinfo  {journal} {Physics Letters A}\ }\textbf {\bibinfo {volume} {378}},\ \bibinfo {pages} {2079} (\bibinfo {year} {2014})}\BibitemShut {NoStop}%
\bibitem [{\citenamefont {Huelga}\ \emph {et~al.}(1997{\natexlab{b}})\citenamefont {Huelga}, \citenamefont {Macchiavello}, \citenamefont {Pellizzari}, \citenamefont {Ekert}, \citenamefont {Plenio},\ and\ \citenamefont {Cirac}}]{SQZ1}%
  \BibitemOpen
  \bibfield  {author} {\bibinfo {author} {\bibfnamefont {S.~F.}\ \bibnamefont {Huelga}}, \bibinfo {author} {\bibfnamefont {C.}~\bibnamefont {Macchiavello}}, \bibinfo {author} {\bibfnamefont {T.}~\bibnamefont {Pellizzari}}, \bibinfo {author} {\bibfnamefont {A.~K.}\ \bibnamefont {Ekert}}, \bibinfo {author} {\bibfnamefont {M.~B.}\ \bibnamefont {Plenio}}, \ and\ \bibinfo {author} {\bibfnamefont {J.~I.}\ \bibnamefont {Cirac}},\ }\bibfield  {title} {\enquote {\bibinfo {title} {Improvement of frequency standards with quantum entanglement},}\ }\href {\doibase 10.1103/PhysRevLett.79.3865} {\bibfield  {journal} {\bibinfo  {journal} {Phys. Rev. Lett.}\ }\textbf {\bibinfo {volume} {79}},\ \bibinfo {pages} {3865} (\bibinfo {year} {1997}{\natexlab{b}})}\BibitemShut {NoStop}%
\bibitem [{\citenamefont {Greenberger}\ \emph {et~al.}(1989)\citenamefont {Greenberger}, \citenamefont {Horne},\ and\ \citenamefont {Zeilinger}}]{st1}%
  \BibitemOpen
  \bibfield  {author} {\bibinfo {author} {\bibfnamefont {D.~M}\ \bibnamefont {Greenberger}}, \bibinfo {author} {\bibfnamefont {M.~A}\ \bibnamefont {Horne}}, \ and\ \bibinfo {author} {\bibfnamefont {A.}~\bibnamefont {Zeilinger}},\ }\bibfield  {title} {\enquote {\bibinfo {title} {Going beyond bell’s theorem},}\ \ }(\bibinfo  {publisher} {Springer},\ \bibinfo {year} {1989})\ p.~\bibinfo {pages} {69}\BibitemShut {NoStop}%
\bibitem [{\citenamefont {Greenberger}\ \emph {et~al.}(1990)\citenamefont {Greenberger}, \citenamefont {Horne}, \citenamefont {Shimony},\ and\ \citenamefont {Zeilinger}}]{st2}%
  \BibitemOpen
  \bibfield  {author} {\bibinfo {author} {\bibfnamefont {D.~M.}\ \bibnamefont {Greenberger}}, \bibinfo {author} {\bibfnamefont {M.~A.}\ \bibnamefont {Horne}}, \bibinfo {author} {\bibfnamefont {A.}~\bibnamefont {Shimony}}, \ and\ \bibinfo {author} {\bibfnamefont {A.}~\bibnamefont {Zeilinger}},\ }\bibfield  {title} {\enquote {\bibinfo {title} {Bell’s theorem without inequalities},}\ }\href {https://doi.org/10.1119/1.16243} {\bibfield  {journal} {\bibinfo  {journal} {Am. J. Phys.}\ }\textbf {\bibinfo {volume} {58}},\ \bibinfo {pages} {1131} (\bibinfo {year} {1990})}\BibitemShut {NoStop}%
\bibitem [{\citenamefont {Hyllus}\ \emph {et~al.}(2012)\citenamefont {Hyllus}, \citenamefont {Laskowski}, \citenamefont {Krischek}, \citenamefont {Schwemmer}, \citenamefont {Wieczorek}, \citenamefont {Weinfurter}, \citenamefont {Pezz\'e},\ and\ \citenamefont {Smerzi}}]{MPE}%
  \BibitemOpen
  \bibfield  {author} {\bibinfo {author} {\bibfnamefont {P.}~\bibnamefont {Hyllus}}, \bibinfo {author} {\bibfnamefont {Wies\l{}aw}\ \bibnamefont {Laskowski}}, \bibinfo {author} {\bibfnamefont {R.}~\bibnamefont {Krischek}}, \bibinfo {author} {\bibfnamefont {C.}~\bibnamefont {Schwemmer}}, \bibinfo {author} {\bibfnamefont {W.}~\bibnamefont {Wieczorek}}, \bibinfo {author} {\bibfnamefont {H.}~\bibnamefont {Weinfurter}}, \bibinfo {author} {\bibfnamefont {L.}~\bibnamefont {Pezz\'e}}, \ and\ \bibinfo {author} {\bibfnamefont {A.}~\bibnamefont {Smerzi}},\ }\bibfield  {title} {\enquote {\bibinfo {title} {Fisher information and multiparticle entanglement},}\ }\href {https://link.aps.org/doi/10.1103/PhysRevA.85.022321} {\bibfield  {journal} {\bibinfo  {journal} {Phys. Rev. A}\ }\textbf {\bibinfo {volume} {85}},\ \bibinfo {pages} {022321} (\bibinfo {year} {2012})}\BibitemShut {NoStop}%
\bibitem [{\citenamefont {Shimony}(1995)}]{GM1}%
  \BibitemOpen
  \bibfield  {author} {\bibinfo {author} {\bibfnamefont {A.}~\bibnamefont {Shimony}},\ }\bibfield  {title} {\enquote {\bibinfo {title} {Degree of entanglement a},}\ }\href {https://doi.org/10.1111/j.1749-6632.1995.tb39008.xCitations: 239} {\bibfield  {journal} {\bibinfo  {journal} {Ann. N.Y. Acad. Sci.}\ }\textbf {\bibinfo {volume} {755}},\ \bibinfo {pages} {675} (\bibinfo {year} {1995})}\BibitemShut {NoStop}%
\bibitem [{\citenamefont {Wei}\ and\ \citenamefont {Goldbart}(2003)}]{GM2}%
  \BibitemOpen
  \bibfield  {author} {\bibinfo {author} {\bibfnamefont {T.-C.}\ \bibnamefont {Wei}}\ and\ \bibinfo {author} {\bibfnamefont {P.~M.}\ \bibnamefont {Goldbart}},\ }\bibfield  {title} {\enquote {\bibinfo {title} {Geometric measure of entanglement and applications to bipartite and multipartite quantum states},}\ }\href {\doibase 10.1103/PhysRevA.68.042307} {\bibfield  {journal} {\bibinfo  {journal} {Phys. Rev. A}\ }\textbf {\bibinfo {volume} {68}},\ \bibinfo {pages} {042307} (\bibinfo {year} {2003})}\BibitemShut {NoStop}%
\bibitem [{\citenamefont {Blasone}\ \emph {et~al.}(2008)\citenamefont {Blasone}, \citenamefont {Dell'Anno}, \citenamefont {De~Siena},\ and\ \citenamefont {Illuminati}}]{GGM1}%
  \BibitemOpen
  \bibfield  {author} {\bibinfo {author} {\bibfnamefont {M.}~\bibnamefont {Blasone}}, \bibinfo {author} {\bibfnamefont {F.}~\bibnamefont {Dell'Anno}}, \bibinfo {author} {\bibfnamefont {S.}~\bibnamefont {De~Siena}}, \ and\ \bibinfo {author} {\bibfnamefont {F.}~\bibnamefont {Illuminati}},\ }\bibfield  {title} {\enquote {\bibinfo {title} {Hierarchies of geometric entanglement},}\ }\href {https://link.aps.org/doi/10.1103/PhysRevA.77.062304} {\bibfield  {journal} {\bibinfo  {journal} {Phys. Rev. A}\ }\textbf {\bibinfo {volume} {77}},\ \bibinfo {pages} {062304} (\bibinfo {year} {2008})}\BibitemShut {NoStop}%
\bibitem [{\citenamefont {Sen(De)}\ and\ \citenamefont {Sen}(2010)}]{PhysRevA.81.012308}%
  \BibitemOpen
  \bibfield  {author} {\bibinfo {author} {\bibfnamefont {A.}~\bibnamefont {Sen(De)}}\ and\ \bibinfo {author} {\bibfnamefont {U.}~\bibnamefont {Sen}},\ }\bibfield  {title} {\enquote {\bibinfo {title} {Channel capacities versus entanglement measures in multiparty quantum states},}\ }\href {\doibase 10.1103/PhysRevA.81.012308} {\bibfield  {journal} {\bibinfo  {journal} {Phys. Rev. A}\ }\textbf {\bibinfo {volume} {81}},\ \bibinfo {pages} {012308} (\bibinfo {year} {2010})}\BibitemShut {NoStop}%
\bibitem [{\citenamefont {De}\ and\ \citenamefont {Sen}(2010)}]{GM3}%
  \BibitemOpen
  \bibfield  {author} {\bibinfo {author} {\bibfnamefont {A.~Sen}\ \bibnamefont {De}}\ and\ \bibinfo {author} {\bibfnamefont {U.}~\bibnamefont {Sen}},\ }\bibfield  {title} {\enquote {\bibinfo {title} {Bound genuine multisite entanglement: detector of gapless-gapped quantum transitions in frustrated systems},}\ }\href {https://doi.org/10.48550/arXiv.1002.1253} {\bibfield  {journal} {\bibinfo  {journal} {arXiv preprint arXiv:1002.1253}\ } (\bibinfo {year} {2010})}\BibitemShut {NoStop}%
\bibitem [{\citenamefont {Das}\ \emph {et~al.}(2010)\citenamefont {Das}, \citenamefont {Roy}, \citenamefont {Bagchi}, \citenamefont {Misra}, \citenamefont {Sen(De)},\ and\ \citenamefont {Sen}}]{GGM2}%
  \BibitemOpen
  \bibfield  {author} {\bibinfo {author} {\bibfnamefont {T.}~\bibnamefont {Das}}, \bibinfo {author} {\bibfnamefont {S.~S.}\ \bibnamefont {Roy}}, \bibinfo {author} {\bibfnamefont {S.}~\bibnamefont {Bagchi}}, \bibinfo {author} {\bibfnamefont {A.}~\bibnamefont {Misra}}, \bibinfo {author} {\bibfnamefont {A.}~\bibnamefont {Sen(De)}}, \ and\ \bibinfo {author} {\bibfnamefont {U.}~\bibnamefont {Sen}},\ }\bibfield  {title} {\enquote {\bibinfo {title} {Generalized geometric measure of entanglement for multiparty mixed states},}\ }\href {https://link.aps.org/doi/10.1103/PhysRevA.94.022336} {\bibfield  {journal} {\bibinfo  {journal} {Phys. Rev. A}\ }\textbf {\bibinfo {volume} {94}},\ \bibinfo {pages} {022336} (\bibinfo {year} {2010})}\BibitemShut {NoStop}%
\bibitem [{\citenamefont {Bennett}\ \emph {et~al.}(1996)\citenamefont {Bennett}, \citenamefont {Bernstein}, \citenamefont {Popescu},\ and\ \citenamefont {Schumacher}}]{Entropy}%
  \BibitemOpen
  \bibfield  {author} {\bibinfo {author} {\bibfnamefont {C.~H.}\ \bibnamefont {Bennett}}, \bibinfo {author} {\bibfnamefont {H.~J.}\ \bibnamefont {Bernstein}}, \bibinfo {author} {\bibfnamefont {S.}~\bibnamefont {Popescu}}, \ and\ \bibinfo {author} {\bibfnamefont {B.}~\bibnamefont {Schumacher}},\ }\bibfield  {title} {\enquote {\bibinfo {title} {Concentrating partial entanglement by local operations},}\ }\href {\doibase 10.1103/PhysRevA.53.2046} {\bibfield  {journal} {\bibinfo  {journal} {Phys. Rev. A}\ }\textbf {\bibinfo {volume} {53}},\ \bibinfo {pages} {2046} (\bibinfo {year} {1996})}\BibitemShut {NoStop}%
\bibitem [{\citenamefont {Ac\'{\i}n}\ \emph {et~al.}(2001)\citenamefont {Ac\'{\i}n}, \citenamefont {Bru\ss{}}, \citenamefont {Lewenstein},\ and\ \citenamefont {Sanpera}}]{ksepb}%
  \BibitemOpen
  \bibfield  {author} {\bibinfo {author} {\bibfnamefont {A.}~\bibnamefont {Ac\'{\i}n}}, \bibinfo {author} {\bibfnamefont {D.}~\bibnamefont {Bru\ss{}}}, \bibinfo {author} {\bibfnamefont {M.}~\bibnamefont {Lewenstein}}, \ and\ \bibinfo {author} {\bibfnamefont {A.}~\bibnamefont {Sanpera}},\ }\bibfield  {title} {\enquote {\bibinfo {title} {Classification of mixed three-qubit states},}\ }\href {\doibase 10.1103/PhysRevLett.87.040401} {\bibfield  {journal} {\bibinfo  {journal} {Phys. Rev. Lett.}\ }\textbf {\bibinfo {volume} {87}},\ \bibinfo {pages} {040401} (\bibinfo {year} {2001})}\BibitemShut {NoStop}%
\bibitem [{\citenamefont {G{\"u}hne}\ \emph {et~al.}(2005)\citenamefont {G{\"u}hne}, \citenamefont {T{\'o}th},\ and\ \citenamefont {Briegel}}]{ksepa}%
  \BibitemOpen
  \bibfield  {author} {\bibinfo {author} {\bibfnamefont {O.}~\bibnamefont {G{\"u}hne}}, \bibinfo {author} {\bibfnamefont {G.}~\bibnamefont {T{\'o}th}}, \ and\ \bibinfo {author} {\bibfnamefont {H.~J.}\ \bibnamefont {Briegel}},\ }\bibfield  {title} {\enquote {\bibinfo {title} {Multipartite entanglement in spin chains},}\ }\href {https://iopscience.iop.org/article/10.1088/1367-2630/7/1/229} {\bibfield  {journal} {\bibinfo  {journal} {NJP}\ }\textbf {\bibinfo {volume} {7}},\ \bibinfo {pages} {229} (\bibinfo {year} {2005})}\BibitemShut {NoStop}%
\bibitem [{\citenamefont {Seevinck}\ and\ \citenamefont {Uffink}(2008)}]{ksep}%
  \BibitemOpen
  \bibfield  {author} {\bibinfo {author} {\bibfnamefont {M.}~\bibnamefont {Seevinck}}\ and\ \bibinfo {author} {\bibfnamefont {J.}~\bibnamefont {Uffink}},\ }\bibfield  {title} {\enquote {\bibinfo {title} {Partial separability and entanglement criteria for multiqubit quantum states},}\ }\href {\doibase 10.1103/PhysRevA.78.032101} {\bibfield  {journal} {\bibinfo  {journal} {Phys. Rev. A}\ }\textbf {\bibinfo {volume} {78}},\ \bibinfo {pages} {032101} (\bibinfo {year} {2008})}\BibitemShut {NoStop}%
\bibitem [{\citenamefont {Escher}\ \emph {et~al.}(2011)\citenamefont {Escher}, \citenamefont {de~Matos~Filho},\ and\ \citenamefont {Davidovich}}]{Dcoh}%
  \BibitemOpen
  \bibfield  {author} {\bibinfo {author} {\bibfnamefont {B.~M.}\ \bibnamefont {Escher}}, \bibinfo {author} {\bibfnamefont {R.~L.}\ \bibnamefont {de~Matos~Filho}}, \ and\ \bibinfo {author} {\bibfnamefont {L.}~\bibnamefont {Davidovich}},\ }\bibfield  {title} {\enquote {\bibinfo {title} {General framework for estimating the ultimate precision limit in noisy quantum-enhanced metrology},}\ }\href {https://doi.org/10.1038/nphys1958} {\bibfield  {journal} {\bibinfo  {journal} {Nat. Phys.}\ }\textbf {\bibinfo {volume} {7}},\ \bibinfo {pages} {406} (\bibinfo {year} {2011})}\BibitemShut {NoStop}%
\bibitem [{\citenamefont {Demkowicz-Dobrza{\'n}ski}\ \emph {et~al.}(2012)\citenamefont {Demkowicz-Dobrza{\'n}ski}, \citenamefont {Ko{\l}ody{\'n}ski},\ and\ \citenamefont {Guță}}]{Qch}%
  \BibitemOpen
  \bibfield  {author} {\bibinfo {author} {\bibfnamefont {R.}~\bibnamefont {Demkowicz-Dobrza{\'n}ski}}, \bibinfo {author} {\bibfnamefont {J.}~\bibnamefont {Ko{\l}ody{\'n}ski}}, \ and\ \bibinfo {author} {\bibfnamefont {M.}~\bibnamefont {Guță}},\ }\bibfield  {title} {\enquote {\bibinfo {title} {The elusive {H}eisenberg limit in quantum-enhanced metrology},}\ }\href {https://doi.org/10.1038/ncomms2067} {\bibfield  {journal} {\bibinfo  {journal} {Nat. Commun.}\ }\textbf {\bibinfo {volume} {3}},\ \bibinfo {pages} {1063} (\bibinfo {year} {2012})}\BibitemShut {NoStop}%
\bibitem [{\citenamefont {Correa}\ \emph {et~al.}(2015)\citenamefont {Correa}, \citenamefont {Mehboudi}, \citenamefont {Adesso},\ and\ \citenamefont {Sanpera}}]{Therm}%
  \BibitemOpen
  \bibfield  {author} {\bibinfo {author} {\bibfnamefont {L.~A.}\ \bibnamefont {Correa}}, \bibinfo {author} {\bibfnamefont {M.}~\bibnamefont {Mehboudi}}, \bibinfo {author} {\bibfnamefont {G.}~\bibnamefont {Adesso}}, \ and\ \bibinfo {author} {\bibfnamefont {A.}~\bibnamefont {Sanpera}},\ }\bibfield  {title} {\enquote {\bibinfo {title} {Individual quantum probes for optimal thermometry},}\ }\href {https://link.aps.org/doi/10.1103/PhysRevLett.114.220405} {\bibfield  {journal} {\bibinfo  {journal} {Phys. Rev. Lett.}\ }\textbf {\bibinfo {volume} {114}},\ \bibinfo {pages} {220405} (\bibinfo {year} {2015})}\BibitemShut {NoStop}%
\bibitem [{\citenamefont {Zhou}(2024{\natexlab{b}})}]{op1}%
  \BibitemOpen
  \bibfield  {author} {\bibinfo {author} {\bibfnamefont {S.}~\bibnamefont {Zhou}},\ }\bibfield  {title} {\enquote {\bibinfo {title} {Limits of noisy quantum metrology with restricted quantum controls},}\ }\href {https://link.aps.org/doi/10.1103/PhysRevLett.133.170801} {\bibfield  {journal} {\bibinfo  {journal} {Phys. Rev. Lett.}\ }\textbf {\bibinfo {volume} {133}},\ \bibinfo {pages} {170801} (\bibinfo {year} {2024}{\natexlab{b}})}\BibitemShut {NoStop}%
\bibitem [{\citenamefont {Ragazzi}\ \emph {et~al.}(2024)\citenamefont {Ragazzi}, \citenamefont {Cavazzoni}, \citenamefont {Bordone},\ and\ \citenamefont {Paris}}]{op}%
  \BibitemOpen
  \bibfield  {author} {\bibinfo {author} {\bibfnamefont {G.}~\bibnamefont {Ragazzi}}, \bibinfo {author} {\bibfnamefont {S.}~\bibnamefont {Cavazzoni}}, \bibinfo {author} {\bibfnamefont {P.}~\bibnamefont {Bordone}}, \ and\ \bibinfo {author} {\bibfnamefont {Matteo G.~A.}\ \bibnamefont {Paris}},\ }\bibfield  {title} {\enquote {\bibinfo {title} {Generalized phase estimation in noisy quantum gates},}\ }\href {https://link.aps.org/doi/10.1103/PhysRevA.110.052425} {\bibfield  {journal} {\bibinfo  {journal} {Phys. Rev. A}\ }\textbf {\bibinfo {volume} {110}},\ \bibinfo {pages} {052425} (\bibinfo {year} {2024})}\BibitemShut {NoStop}%
\bibitem [{\citenamefont {Yang}(2024)}]{op2}%
  \BibitemOpen
  \bibfield  {author} {\bibinfo {author} {\bibfnamefont {J.}~\bibnamefont {Yang}},\ }\bibfield  {title} {\enquote {\bibinfo {title} {Quantum measurement encoding for quantum metrology},}\ }\href {https://link.aps.org/doi/10.1103/PhysRevResearch.6.043084} {\bibfield  {journal} {\bibinfo  {journal} {Phys. Rev. Res.}\ }\textbf {\bibinfo {volume} {6}},\ \bibinfo {pages} {043084} (\bibinfo {year} {2024})}\BibitemShut {NoStop}%
\bibitem [{\citenamefont {Berry}\ \emph {et~al.}(2010)\citenamefont {Berry}, \citenamefont {Xiang}, \citenamefont {Higgins}, \citenamefont {Wiseman},\ and\ \citenamefont {Pryde}}]{uni2}%
  \BibitemOpen
  \bibfield  {author} {\bibinfo {author} {\bibfnamefont {D.}~\bibnamefont {Berry}}, \bibinfo {author} {\bibfnamefont {G.}~\bibnamefont {Xiang}}, \bibinfo {author} {\bibfnamefont {B}~\bibnamefont {Higgins}}, \bibinfo {author} {\bibfnamefont {H.}~\bibnamefont {Wiseman}}, \ and\ \bibinfo {author} {\bibfnamefont {G.}~\bibnamefont {Pryde}},\ }\bibfield  {title} {\enquote {\bibinfo {title} {Entanglement-enhanced measurement of a completely unknown phase},}\ }\href {https://doi.org/10.1038/nphoton.2010.268} {\bibfield  {journal} {\bibinfo  {journal} {Nat. Photonics}\ }\textbf {\bibinfo {volume} {5}},\ \bibinfo {pages} {43} (\bibinfo {year} {2010})}\BibitemShut {NoStop}%
\bibitem [{\citenamefont {Ahnefeld}\ \emph {et~al.}(2025)\citenamefont {Ahnefeld}, \citenamefont {Theurer},\ and\ \citenamefont {Plenio}}]{COref}%
  \BibitemOpen
  \bibfield  {author} {\bibinfo {author} {\bibfnamefont {F.}~\bibnamefont {Ahnefeld}}, \bibinfo {author} {\bibfnamefont {T.}~\bibnamefont {Theurer}}, \ and\ \bibinfo {author} {\bibfnamefont {M.~B.}\ \bibnamefont {Plenio}},\ }\bibfield  {title} {\enquote {\bibinfo {title} {Coherence as a resource for phase estimation},}\ }\href {https://doi.org/10.48550/arXiv.2505.18544} {\bibfield  {journal} {\bibinfo  {journal} {arXiv preprint arXiv:2505.18544}\ } (\bibinfo {year} {2025})}\BibitemShut {NoStop}%
\bibitem [{\citenamefont {Fr\"owis}\ \emph {et~al.}(2016)\citenamefont {Fr\"owis}, \citenamefont {Sekatski},\ and\ \citenamefont {D\"ur}}]{uni5}%
  \BibitemOpen
  \bibfield  {author} {\bibinfo {author} {\bibfnamefont {F.}~\bibnamefont {Fr\"owis}}, \bibinfo {author} {\bibfnamefont {P.}~\bibnamefont {Sekatski}}, \ and\ \bibinfo {author} {\bibfnamefont {W.}~\bibnamefont {D\"ur}},\ }\bibfield  {title} {\enquote {\bibinfo {title} {Detecting large quantum fisher information with finite measurement precision},}\ }\href {\doibase 10.1103/PhysRevLett.116.090801} {\bibfield  {journal} {\bibinfo  {journal} {Phys. Rev. Lett.}\ }\textbf {\bibinfo {volume} {116}},\ \bibinfo {pages} {090801} (\bibinfo {year} {2016})}\BibitemShut {NoStop}%
\bibitem [{\citenamefont {Liu}\ \emph {et~al.}(2025)\citenamefont {Liu}, \citenamefont {Yang}, \citenamefont {Shi},\ and\ \citenamefont {Yu}}]{uni4}%
  \BibitemOpen
  \bibfield  {author} {\bibinfo {author} {\bibfnamefont {J.-X.}\ \bibnamefont {Liu}}, \bibinfo {author} {\bibfnamefont {J.}~\bibnamefont {Yang}}, \bibinfo {author} {\bibfnamefont {H.-L.}\ \bibnamefont {Shi}}, \ and\ \bibinfo {author} {\bibfnamefont {S.}~\bibnamefont {Yu}},\ }\bibfield  {title} {\enquote {\bibinfo {title} {Optimal local measurements in single-parameter quantum metrology},}\ }\href {https://link.aps.org/doi/10.1103/PhysRevA.111.022436} {\bibfield  {journal} {\bibinfo  {journal} {Phys. Rev. A}\ }\textbf {\bibinfo {volume} {111}},\ \bibinfo {pages} {022436} (\bibinfo {year} {2025})}\BibitemShut {NoStop}%
\bibitem [{\citenamefont {Wootters}(1981)}]{Wooters}%
  \BibitemOpen
  \bibfield  {author} {\bibinfo {author} {\bibfnamefont {W.~K.}\ \bibnamefont {Wootters}},\ }\bibfield  {title} {\enquote {\bibinfo {title} {Statistical distance and hilbert space},}\ }\href {https://link.aps.org/doi/10.1103/PhysRevD.23.357} {\bibfield  {journal} {\bibinfo  {journal} {Phys. Rev. D}\ }\textbf {\bibinfo {volume} {23}},\ \bibinfo {pages} {357} (\bibinfo {year} {1981})}\BibitemShut {NoStop}%
\bibitem [{\citenamefont {Braunstein}\ and\ \citenamefont {Caves}(1994{\natexlab{b}})}]{Braunstein1}%
  \BibitemOpen
  \bibfield  {author} {\bibinfo {author} {\bibfnamefont {S.~L.}\ \bibnamefont {Braunstein}}\ and\ \bibinfo {author} {\bibfnamefont {C.~M.}\ \bibnamefont {Caves}},\ }\bibfield  {title} {\enquote {\bibinfo {title} {Statistical distance and the geometry of quantum states},}\ }\href {https://link.aps.org/doi/10.1103/PhysRevLett.72.3439} {\bibfield  {journal} {\bibinfo  {journal} {Phys. Rev. Lett.}\ }\textbf {\bibinfo {volume} {72}},\ \bibinfo {pages} {3439} (\bibinfo {year} {1994}{\natexlab{b}})}\BibitemShut {NoStop}%
\bibitem [{\citenamefont {Holevo}(2011)}]{holevo}%
  \BibitemOpen
  \bibfield  {author} {\bibinfo {author} {\bibfnamefont {A.~S.}\ \bibnamefont {Holevo}},\ }\href {https://archive.org/details/probabilisticsta0000hole} {\emph {\bibinfo {title} {Probabilistic and statistical aspects of quantum theory}}},\ Vol.~\bibinfo {volume} {1}\ (\bibinfo  {publisher} {Springer Science \& Business Media},\ \bibinfo {year} {2011})\BibitemShut {NoStop}%
\bibitem [{\citenamefont {Tan}\ and\ \citenamefont {Jeong}(2019)}]{review1}%
  \BibitemOpen
  \bibfield  {author} {\bibinfo {author} {\bibfnamefont {K.~C.}\ \bibnamefont {Tan}}\ and\ \bibinfo {author} {\bibfnamefont {H.}~\bibnamefont {Jeong}},\ }\bibfield  {title} {\enquote {\bibinfo {title} {{Nonclassical light and metrological power: An introductory review}},}\ }\href {https://doi.org/10.1116/1.5126696} {\bibfield  {journal} {\bibinfo  {journal} {AVS Quantum Sci.}\ }\textbf {\bibinfo {volume} {1}},\ \bibinfo {pages} {014701} (\bibinfo {year} {2019})}\BibitemShut {NoStop}%
\bibitem [{\citenamefont {T\'oth}\ and\ \citenamefont {V\'ertesi}(2018)}]{Ent3}%
  \BibitemOpen
  \bibfield  {author} {\bibinfo {author} {\bibfnamefont {G.}~\bibnamefont {T\'oth}}\ and\ \bibinfo {author} {\bibfnamefont {T.}~\bibnamefont {V\'ertesi}},\ }\bibfield  {title} {\enquote {\bibinfo {title} {Quantum states with a positive partial transpose are useful for metrology},}\ }\href {\doibase 10.1103/PhysRevLett.120.020506} {\bibfield  {journal} {\bibinfo  {journal} {Phys. Rev. Lett.}\ }\textbf {\bibinfo {volume} {120}},\ \bibinfo {pages} {020506} (\bibinfo {year} {2018})}\BibitemShut {NoStop}%
\bibitem [{\citenamefont {Kwiat}\ \emph {et~al.}(1995)\citenamefont {Kwiat}, \citenamefont {Mattle}, \citenamefont {Weinfurter}, \citenamefont {Zeilinger}, \citenamefont {Sergienko},\ and\ \citenamefont {Shih}}]{Exst}%
  \BibitemOpen
  \bibfield  {author} {\bibinfo {author} {\bibfnamefont {P.~G.}\ \bibnamefont {Kwiat}}, \bibinfo {author} {\bibfnamefont {K.}~\bibnamefont {Mattle}}, \bibinfo {author} {\bibfnamefont {H.}~\bibnamefont {Weinfurter}}, \bibinfo {author} {\bibfnamefont {A.}~\bibnamefont {Zeilinger}}, \bibinfo {author} {\bibfnamefont {A.~V.}\ \bibnamefont {Sergienko}}, \ and\ \bibinfo {author} {\bibfnamefont {Y.}~\bibnamefont {Shih}},\ }\bibfield  {title} {\enquote {\bibinfo {title} {New high-intensity source of polarization-entangled photon pairs},}\ }\href {\doibase 10.1103/PhysRevLett.75.4337} {\bibfield  {journal} {\bibinfo  {journal} {Phys. Rev. Lett.}\ }\textbf {\bibinfo {volume} {75}},\ \bibinfo {pages} {4337} (\bibinfo {year} {1995})}\BibitemShut {NoStop}%
\bibitem [{\citenamefont {Kim}\ \emph {et~al.}(2003)\citenamefont {Kim}, \citenamefont {Kulik}, \citenamefont {Chekhova}, \citenamefont {Grice},\ and\ \citenamefont {Shih}}]{Exst2}%
  \BibitemOpen
  \bibfield  {author} {\bibinfo {author} {\bibfnamefont {Y.-H.}\ \bibnamefont {Kim}}, \bibinfo {author} {\bibfnamefont {S.~P.}\ \bibnamefont {Kulik}}, \bibinfo {author} {\bibfnamefont {M.~V.}\ \bibnamefont {Chekhova}}, \bibinfo {author} {\bibfnamefont {W.~P.}\ \bibnamefont {Grice}}, \ and\ \bibinfo {author} {\bibfnamefont {Y.}~\bibnamefont {Shih}},\ }\bibfield  {title} {\enquote {\bibinfo {title} {Experimental entanglement concentration and universal bell-state synthesizer},}\ }\href {https://link.aps.org/doi/10.1103/PhysRevA.67.010301} {\bibfield  {journal} {\bibinfo  {journal} {Phys. Rev. A}\ }\textbf {\bibinfo {volume} {67}},\ \bibinfo {pages} {010301} (\bibinfo {year} {2003})}\BibitemShut {NoStop}%
\bibitem [{\citenamefont {Kim}\ \emph {et~al.}(2001)\citenamefont {Kim}, \citenamefont {Chekhova}, \citenamefont {Kulik}, \citenamefont {Rubin},\ and\ \citenamefont {Shih}}]{Exst3}%
  \BibitemOpen
  \bibfield  {author} {\bibinfo {author} {\bibfnamefont {Yoon-Ho}\ \bibnamefont {Kim}}, \bibinfo {author} {\bibfnamefont {Maria~V.}\ \bibnamefont {Chekhova}}, \bibinfo {author} {\bibfnamefont {Sergei~P.}\ \bibnamefont {Kulik}}, \bibinfo {author} {\bibfnamefont {Morton~H.}\ \bibnamefont {Rubin}}, \ and\ \bibinfo {author} {\bibfnamefont {Yanhua}\ \bibnamefont {Shih}},\ }\bibfield  {title} {\enquote {\bibinfo {title} {Interferometric bell-state preparation using femtosecond-pulse-pumped spontaneous parametric down-conversion},}\ }\href {\doibase 10.1103/PhysRevA.63.062301} {\bibfield  {journal} {\bibinfo  {journal} {Phys. Rev. A}\ }\textbf {\bibinfo {volume} {63}},\ \bibinfo {pages} {062301} (\bibinfo {year} {2001})}\BibitemShut {NoStop}%
\end{thebibliography}%
\end{document}